\documentclass{article}

\usepackage{microtype}
\usepackage{graphicx}
\usepackage{subfigure}
\usepackage{booktabs} % for professional tables
\usepackage{amsmath,amsfonts,amssymb, amsthm} % beautiful math
\usepackage{dsfont}
\usepackage{xcolor}         % colors
\usepackage{floatflt}
\usepackage{enumitem}
\usepackage{multicol}
\usepackage{float}
\usepackage{wrapfig} 
\usepackage{hyperref}
\usepackage{float}
\DeclareMathOperator{\Var}{Var}
\DeclareMathOperator{\Cov}{Cov}
\DeclareMathOperator{\E}{\mathbb{E}}

\usepackage[accepted]{icml2025}

\usepackage{amsmath}
\usepackage{amssymb}
\usepackage{mathtools}
\usepackage{amsthm}

\usepackage[capitalize,noabbrev]{cleveref}

\theoremstyle{plain}
\newtheorem{theorem}{Theorem}[section]
\newtheorem{proposition}[theorem]{Proposition}
\newtheorem{lemma}[theorem]{Lemma}

\theoremstyle{definition}
\newtheorem{definition}[theorem]{Definition}
\newtheorem{assumption}[theorem]{Assumption}
\theoremstyle{remark}

\usepackage[disable,textsize=tiny]{todonotes}
\icmltitlerunning{Quantifying Treatment Effects: Estimating Risk Ratios via Observational Studies}

\begin{document}

\twocolumn[
\icmltitle{Quantifying Treatment Effects: Estimating Risk Ratios via Observational Studies}



\icmlsetsymbol{equal}{*}

\begin{icmlauthorlist}
\icmlauthor{Ahmed Boughdiri}{INRIA}
\icmlauthor{Julie Josse}{INRIA}
\icmlauthor{Erwan Scornet}{Sorbonne}

\end{icmlauthorlist}

\icmlaffiliation{INRIA}{INRIA Sophia-Antipolis}
\icmlaffiliation{Sorbonne}{Sorbonne Université and Université Paris Cité}

\icmlcorrespondingauthor{Ahmed Boughdiri}{ahmed.boughdiri@inria.fr}

\icmlkeywords{Causal inference, Observational data, Risk Ratio}

\vskip 0.3in
]



\printAffiliationsAndNotice{\icmlEqualContribution} % otherwise use the standard text.

\begin{abstract}
The Risk Difference (RD), an absolute measure of effect, is widely used and well-studied in both randomized controlled trials (RCTs) and observational studies. Complementary to the RD, the Risk Ratio (RR), as a relative measure, is critical for a comprehensive understanding of intervention effects: RD can downplay small absolute changes, while RR can highlight them. Despite its significance, the theoretical study of RR has received less attention, particularly in observational settings. This paper addresses this gap by tackling the estimation of RR in observational data. We propose several RR estimators and establish their theoretical properties, including asymptotic normality and confidence intervals. Through analyses on simulated and real-world datasets, we evaluate the performance of these estimators in terms of bias, efficiency, and robustness to generative data models. We also examine the coverage and length of the associated confidence intervals. Due to the non-linear nature of RR, influence function theory yields two distinct efficient estimators with different convergence assumptions. Based on theoretical and empirical insights, we recommend, among all estimators, one of the two doubly-robust estimators, which, intriguingly, challenges conventional expectations.
\end{abstract}

\section{Introduction}
\label{Introduction}

\paragraph{Treatment effect estimation in trials.} 
Modern evidence-based medicine prioritizes Randomized Controlled Trials (RCTs) as the cornerstone of clinical evidence. Randomization in RCTs allows for the quantification of the average treatment effect (ATE) by removing confounding influences from extraneous or undesirable factors. 
The medical guideline CONSORT \citep{CONSORT2010} recommends reporting the treatment effect with relative measures like the Risk Ratio (RR) along with absolute measures like the Risk Difference (RD) to provide a more comprehensive understanding of the effect and its implications, as neither measure alone offers a complete picture. 
Indeed, selecting one measure over another carries several implication. For instance, with 3\% baseline mortality reduced to 1\% by treatment, RD shows a 2\% drop, while RR shows controls have three times the risk: RD suggests a small effect; RR highlights a larger one. In addition, \cite{Naylor1992} and \cite{Forrow1992HowResultsAreSummarized} demonstrated that physicians' inclination to treat patients, based on their perception of therapeutic impact, is influenced by the scale utilized to present clinical effects. Finally, the treatment effect may be heterogeneous in one scale, i.e. the treatment effect varies according to patient characteristics, but homogeneous in another scale \citep{Rothman2011bookEpidemiologyIntrod}, which significantly disrupts interpretation. \cite{colnet2023risk} discusses causal measure properties with a focus on generalization of the treatment effect from a trial to a target population. 

Consequently, both RD and RR measures are widely used in the analysis of clinical trial data as explained in \citet{malenka1993framing}, \citet{sinclair1994clinically} and \citet{NAKAYAMA19981179}. The Risk Ratio is particularly relevant in scenarios where outcomes are always either positive or negative \citep{malenka1993framing} and in cases where both expected potential outcomes being compared are small, as it is more stable and interpretable than the RD. Furthermore, in cases of rare treatment occurrences, the Risk Ratio closely approximates the Odds Ratio (OR), further enhancing its relevance and utility in clinical analyses \cite{schechtman2002odds}. \citet{barratt2004tips} recommend using the RR from clinical trials with an estimation of the individual patient baseline to provide the right treatment. 

\paragraph{Treatment effect estimation in observational data.} 
Despite being the gold standard to assess treatment effects, RCTs may face limitations due to stringent eligibility criteria, unrealistic real-world compliance, short study durations, and limited sample sizes. Medical journals such as JAMA \citep{JAMA} and others \citep{Hernan} have advocated the use of real-world data, often referred to as observational data, to provide additional sources of evidence. These data sets are typically less expensive to collect, more representative of the target population, and usually encompass large sample sizes. 

In the context of observational studies, various estimators exist to measure the treatment effect, mainly on an absolute scale. Different methods such as re-weighting using Inverse Probability Weighting (IPW) \citep{Hirano}, outcome modeling with the G-formula, and doubly-robust approaches like Augmented Inverse Probability Weighting (AIPW) \citep{Robins1986} aim to estimate the RD while handling confounding effects. 

However, to the best of our knowledge, there exist no work proposing or using estimators of the ATE measured with  the Risk Ratio in observational studies based on (non-)parametric estimation (G-formula or AIPW), nor derivations of their theoretical properties. A clear gap exists, needing robust estimators and analyses to improve treatment effect assessments and support medical recommendations in observational studies. Noteworthily, compared to the Risk Difference, studying the Risk Ratio induces additional technical difficulties, due to its non-linear nature.

\paragraph{Contributions.} 
In this paper, we propose and analyze different estimators of the Risk Ratio in observational studies. Considering first the well studied RCT setting in \Cref{risk_ratio_RCTs}, we analyze the first RR estimator introduced by \citet{cornfield1956statistical}, establishing a Central Limit Theorem and asymptotic confidence intervals. We prove that in RCT adjusting for covariates or estimating the probability of being treated reduces the estimator variance. 
As the probability of being treated varies across individuals in observational studies, the above estimator is no longer valid. In \Cref{risk_ratio_OBS}, we detail different estimators that can be used for estimating the RR in observational studies:  Inverse Propensity Weighting (RR-IPW), G-formula (RR-G) and two doubly-robust estimators.  For the first two, we prove their asymptotic normality when the nuisance functions (surface responses and propensity score) are known or estimated via parametric models. Our analyses show that estimating nuisance function increases the variance of RR-G and decreases that of RR-IPW. Besides, using influence function theory  \citep[see, e.g.,][]{kennedy2022semiparametric}, we derive two doubly-robust estimators, RR-OS, for one-step correction and RR-AIPW based on estimating equations. Contrary to the RD, due to the non-linearity of the RR, both estimators differ. We prove that they are both asymptotically unbiased and have the minimal variance among all asymptotically unbiased estimators. Surprisingly, the RR-AIPW estimator turns out to be a plug-in version of AIPW estimators for both the numerators and denominators, which requires weaker assumptions than RR-OS to be asymptotically normal. For all estimators, our asymptotic analyses allow us to build asymptotic confidence intervals.  
In \Cref{sec:simulations}, we evaluate all estimators on observational data, and study the empirical properties of confidence intervals in terms of coverage and lengths. In \Cref{sec:real_world_exp}, we extend our analysis to a semi-synthetic and a real-world dataset. We recommend the RR-AIPW due to $(i)$ less stringent assumptions needed to ensure asymptotic normality, $(ii)$ empirical performances (consistent and with a small asymptotical variance), $(iii)$ the ease of implementation, as a ratio of two simple AIPW estimators.

\paragraph{Related work - Estimation of RR.} 
To the best of our knowledge, \citet{cornfield1956statistical} was the first to propose an estimate of the RR, together with exact and asymptotic confidence intervals, for binary responses, in a RCT scenario, followed by \citet{kupper1975hybrid, katz1978obtaining, bailey1987confidence, morris1988statistics, sato1992maximum}. Considering a logistic model, \citet{schouten1993risk} propose a RR estimator.  Later on, exact confidence intervals were derived by \citet{wang2015exact}.
Recently, Inverse Propensity Weighting schemes have been used in different study designs to estimate the Odds Ratios \citep{staus2022inverse} or directly the RR in some simple settings \cite{hernan2006estimating}.
Besides, pseudo-Poisson and pseudo normal distribution have been proposed with IPW strategies to estimate RD and RR in clinical trials \citep{noma2023risk}. Other methods include G-formula approaches \citep{dukes2018note} and semiparametric density estimation \citep{kennedy2023semiparametric}.
  
In observational studies, one can mention the work of \citet{richardson2017modeling},\citet{yadlowsky2021estimation} and \citet{shirvaikar2023targeting} who focus on estimators of the conditional  average treatment effect (CATE) for the RR. \citet{curth2020estimating} introduces an ``IF-learning'' approach with pseudo outcome regression and derive the influence function for the CATE of the RR. However, since the expectation of a ratio is not the ratio of the expectations,  CATE estimations do not directly yield ATE for the RR. This departs from the RD, for which the ATE is simply the expectation of the CATE. 

\section{A Well-Known Risk Ratio Estimator in RCT}
\label{risk_ratio_RCTs}

\paragraph{Problem setting} 
Following the potential outcome framework \citep[see][]{ RubinDonaldB1974Eceo,Neyman}, we consider the random variables $(X, T, Y^{(0)}, Y^{(1)})$, where $X \in  \mathbb{R}^p$ denotes covariates describing a patient, $T$ is the treatment assignment ($T =1$ when the treatment is given to an individual, $T=0$ otherwise) and $Y^{(0)}$ (resp.  $Y^{(1)}$) is the outcome of interest, describing the status of a patient without treatment (and with treatment respectively).
In practice, we do not have access simultaneously to  $Y^{(1)}$ and  $Y^{(0)}$, and we only observe $$Y = T Y^{(1)} + (1-T) Y^{(0)}.$$
Causal effect measures are functions of the joint distribution of potential outcomes \citep[see][]{Pearl_2009}. In particular, the Risk Difference (RD) and the Risk Ratio (RR) contrast the two states as followed
\begin{align}
 \tau_{{\textrm{\tiny RD}}} = \mathbb{E}[Y^{(1)}] - \mathbb{E}[Y^{(0)}]  \quad\text{and} \quad \tau_{{\textrm{\tiny RR}}} = \frac{\mathbb{E}[Y^{(1)}]}{\mathbb{E}[Y^{(0)}]}.   \label{def_tau_RR}
\end{align}
The aim of this paper is to propose and study estimators of \(\tau_{{\textrm{\tiny RR}}}\) for binary and continuous outcomes. Indeed, all our theoretical results (except \Cref{OLS_RR_OBS}) are valid for binary and continuous outcomes. To estimate \(\tau_{{\textrm{\tiny RR}}}\), we assume to be given an i.i.d dataset \((X_1, T_1, Y_1), \hdots, (X_n, T_n, Y_n)\).

The most simple estimator of the risk ratio consists in replacing expectations by empirical means in $\tau_{{\textrm{\tiny RR}}}$ (Equation~\ref{def_tau_RR}). Such an estimator has already been proposed outside the causal inference framework, with confidence intervals for the binary case \citep[where $Y \in \{0,1\}$, see ][]{katz1978obtaining, bailey1987confidence}. In the potential outcome framework, inspired by the Neyman estimator of the Risk Difference \citep{Neyman}, we call this estimator the Risk Ratio Neyman estimator. 

%The Neyman estimator, also known as the difference-in-means estimator, calculates the difference between means of the outcome in treated and control groups. In order to estimate the RR, we compute the ratio of the means instead.

\begin{definition}[\textbf{Risk Ratio Neyman estimator}]
Let \(N_1 = \sum_{i=1}^{n} T_i\) and \(N_0 = n-n_1\). The Risk Ratio Neyman estimator, denoted \(\hat{\tau}_{\textrm{\tiny RR,N}}\), is defined as
    \begin{align}
      \hat{\tau}_{\textrm{\tiny RR,N}} = \frac{\frac{1}{N_1}\sum_{i=1}^n T_i Y_i}{\frac{1}{N_0}\sum_{i=1}^n (1 - T_i) Y_i},
      \end{align}
      if the denominator is nonzero and \(0\) otherwise.
\end{definition}

With our notation, the $95\%$ confidence interval for $\tau_{{\textrm{\tiny RR}}}$ for binary outcomes    \cite{katz1978obtaining} takes the form
\begin{align}
\left[ \hat{\tau}_{\textrm{\tiny RR,N}} \exp({- z_{1 - \alpha/2} \hat{\sigma}_n }),  \hat{\tau}_{\textrm{\tiny RR,N}} \exp({ z_{1 - \alpha/2} \hat{\sigma}_n }) \right], 
\label{eq_CI_existing_literature}
\end{align}
where $z_{1 - \alpha/2}$ is the $1-\alpha/2$ quantile of a standard Gaussian $\mathcal{N}(0,1)$ and 
\begin{align}
   \hat{\sigma}_n = \sqrt{\frac{1}{\sum_{i=1}^n T_i Y_i} - \frac{1}{N_1} + \frac{1}{\sum_{i=1}^n (1-T_i) Y_i} - \frac{1}{N_0}}.
\end{align}
In the sequel, we establish under which theoretical assumptions the RR-N is an accurate estimator of the Risk Ratio in Randomized Clinical Trials.

RCT randomly assign treatment to patients in order to evaluate treatment effects. We focus on the Bernoulli design, one of the most widely used RCT designs \citep[][]{RubinDonaldB1974Eceo, Imbens_Rubin_2015}, where each participant has the same probability \(e \in (0,1)\) of being treated, independently of the treatments of others. 
%In this section, we use the following assumptions. 

\begin{assumption}[Bernoulli Trial]\label{a:bernoulli_trial}
We assume:
\begin{enumerate}[leftmargin=*,label=(\roman*)]
    \item\label{a:ignorability} \textbf{Ignorability/Exchangeability:} \(T\perp\mkern-9.5mu\perp (Y^{(0)}, Y^{(1)})\).
    \item\label{a:SUTVA} \textbf{SUTVA:} \(Y = T Y^{(1)} + (1-T)Y^{(0)}\).
    \item\label{a:i.i.d.} \textbf{i.i.d.:} \((X_i, T_i, Y_i^{(0)}, Y_i^{(1)})_{i\in [n]} \stackrel{\textrm{i.i.d.}}{\thicksim} \mathcal{P}\).
    \item\label{a:Trial_positivity} \textbf{Trial Positivity:} Each participant has a fixed probability \(e \in (0,1)\) of assignment: \(\mathbb{P}[T_i = 1] = e\).
\end{enumerate}
\end{assumption}

To ensure our estimates are valid, we need to guarantee the existence of the ratio we aim to estimate.
  
 \begin{assumption}[\textbf{Outcome positivity}]\label{a:Outcome_positivity}
    We suppose that both \(Y^{(0)}\) and \(Y^{(1)}\) are squared integrable and that $ \mathbb{E}\left[Y^{(0)}\right] > 0.$
\end{assumption}

\begin{proposition}[\textbf{Asymptotic normality of \(\hat{\tau}_{\textrm{\tiny RR,N}}\)}]\label{prop:N}
Grant \Cref{a:bernoulli_trial} and \Cref{a:Outcome_positivity}, the Risk Ratio Neyman estimator is asymptotically unbiased and satisfies
\begin{align}\label{HT_VS_N}
 \sqrt{n}\left(\hat{\tau}_{\textrm{\tiny RR,N}} - \tau_{\textrm{\tiny RR}} \right) \stackrel{d}{\rightarrow} \mathcal{N}\left(0, V_{\textrm{\tiny RR,N}} \right)
 \end{align}
where 
\begin{align*}
V_{\textrm{\tiny RR,N}} &=\tau_{\textrm{\tiny RR}}^2\left(\frac{\Var\left(Y^{(1)}\right)}{e\mathbb{E}[Y^{(1)}]^2} + \frac{\Var\left(Y^{(0)}\right)}{(1-e)\mathbb{E}[Y^{(0)}]^2}\right).
%\\ &= V_{\textrm{\tiny RR,HT}} - \frac{\tau_{\textrm{\tiny RR}}^2}{e(1-e)}.
 \end{align*}
%Consequently, a $(1-\alpha)$ asymptotic confidence interval for $\tau_{\textrm{\tiny RR}}$ is given by 
%\begin{align}
%\left[ \hat{\tau}_{\textrm{\tiny RR,N,n}} \pm z_{1-\alpha/2} \sqrt{\frac{\widehat{V_{\textrm{\tiny RR,N}}}}{n}} \right]
%\end{align}

%for all estimator $\widehat{V_{\textrm{\tiny RR,N}}} $ of  $V_{\textrm{\tiny RR,N}}$.
%Furthermore, if we assume that for all \(i\), \( M \geq Y_i \geq m > 0\) and \(0 < \sum_{i=1}^n T_i < n\), we have:
%\[\left|Bias(\hat{\tau}_{\textrm{\tiny RR, HT, n}})\right| \leq \frac{2M^3}{nm^3} \]\[\left|\Var(\hat{\tau}_{\textrm{\tiny RR, HT, n}})\right| \leq \frac{4M^4}{nm^6}\]
\end{proposition} 

\Cref{prop:N} establishes the asymptotic normality of the RR-N estimator, a simple ratio of mean estimates, which leads to asymptotic confidence intervals (CI). Indeed, according to \Cref{prop:N}, for all $\alpha \in (0,1)$, a $(1-\alpha)$ asymptotic confidence interval for $\tau_{\textrm{\tiny RR}}$ is given by 
\begin{align}
\left[ \hat{\tau}_{\textrm{\tiny RR,N}} \pm z_{1-\alpha/2} \sqrt{\widehat{V_{\textrm{\tiny RR,N}}}/n} \right],
\end{align}
where $\widehat{V_{\textrm{\tiny RR,N}}}$ is any estimation of $V_{\textrm{\tiny RR,N}}$.
Throughout this paper, based on the Central Limit Theorems we establish, we will consider such CI. The properties of the different CI are studied in \Cref{subsec:CI}.

Contrary to \citet{katz1978obtaining}, \Cref{prop:N} is valid for both continuous and binary outcomes. However, considering binary outcomes in \Cref{prop:N} leads to an asymptotic confidence interval equivalent to that presented in \citet{katz1978obtaining} (see Appendix~\ref{app_link_existing_CI}). Deriving a Central Limit Theorem for $\log (\hat{\tau}_{\textrm{\tiny RR,N}})$ instead of $\hat{\tau}_{\textrm{\tiny RR,N}}$ would lead to the exact same CI (see Appendix~\ref{app_delta_method_log}). 

\paragraph{Probability of receiving treatment} As the probability of treatment $e$ is known in an RCT, one could be tempted to consider what we call the Risk Ratio Horvitz-Thomson estimator  \citep[in reference of the Risk Difference Horvitz-Thomson estimator, see][]{Horvitz_Thompson} defined as 
    \begin{align}
     \hat{\tau}_{\textrm{\tiny RR,HT}} = \frac{\frac{1}{n}\sum_{i=1}^n \frac{T_iY_i}{e}}{\frac{1}{n}\sum_{i=1}^n \frac{(1-T_i)Y_i}{1-e}}
     \end{align}
     if \(\sum_{i=1}^n T_i < n\) and \(0\) otherwise.
Indeed, the  frequency of treatments assignments in the sample may be different from the actual probability of receiving treatment  \(e\). Similarly to what \citet{Hirano, Hahn1998efficiencybound, Robins1992EstimatingEE} noticed for the RD, we prove in Appendix~\ref{proof:HT} that opting for \(\hat e\) over \(e\) in the Risk Ratio estimator (thereby employing the RR-N instead of the RR-HT) results in a reduced asymptotic variance, with a larger reduction when $e$ is close to zero or one. More precisely, letting $V_{\textrm{\tiny RR,HT}}$ the asymptotic variance of $\hat{\tau}_{\textrm{\tiny RR,HT}}$, we have 
\begin{align}
V_{\textrm{\tiny RR,N}} &= V_{\textrm{\tiny RR,HT}} - \tau_{\textrm{\tiny RR}}^2/e(1-e). 
\end{align}
%This advantage arises because \(\hat e\)  corresponds more closely to the actual probability of receiving treatment in the study.
%Similarly to the RR-HT, one can select an optimal value of \(e_\mathrm{opt}\) at which the RCT should be conducted prior to the actual experiment to minimize the variance of the RR-N. This procedure is also in detailed in Appendix \ref{proof:N}. 
%As for the RD, the RR-N can be viewed as a variant of RR-HT, where the probability \(e\) to be treated  (or propensity score) is estimated by \(\hat{e} = n_1/n\). 

%We propose three RR estimators: the RR Horvitz-Thompson, the RR Neyman and the RR G-formula. We provide finite-sample as well as asymptotic bias and variance of these estimators. 

%Note that the asymptotic variance of the RR-HT is influenced by \(e\), so one can minimize the variance of the estimator by choosing an optimal \(e_\mathrm{opt}\) prior to conducting the RCT. \(e_\mathrm{opt}\) is the solution of a polynomial equation which depends on whether \(\Var(Y^{(1)})\) or \(\Var(Y^{(0)})\) is larger. An estimation of the two variances thus lead to an estimation of the optimal $e_\mathrm{opt}$, which equals 0.5 when both variance are equals (see Appendix \ref{proof:HT} for details). 

%\paragraph{Finite bias and variance given in Appendix - To be discussed}

\section{Risk Ratio Estimators in Observational Studies}\label{risk_ratio_OBS}

%Observational studies are forms of empirical research that stand in contrast to randomized controlled trials (RCTs) by not engaging in the direct manipulation or intervention of variables. These studies are particularly prevalent in fields like epidemiology, social sciences, and economics, where they capture the nuances of real-world interactions. 

%Observational studies reveal the complexities of real-world scenarios, which may be missed by  the controlled designs of RCTs. 
A key difference between RCTs and observational studies is the handling of confounding variables. If not properly addressed, these can distort the true causal association between exposure and outcome due to their correlation with both. Therefore, estimating the Risk Ratio in observational studies is more complex than in RCTs, as randomization assumptions do not apply (i.e. the propensity score now depends on the covariates $X$). %\jj{on pourra faire plus court si besoin}

%\subsection{Model design and assumptions for identifiability in observational studies}

\begin{assumption}[\textbf{Observational Study Identifiability Assumptions}]
\label{a:observational_trial}
We assume:
\begin{enumerate}[leftmargin=*,label=(\roman*)]
    \item\label{a:Unconfoundedness} \textbf{Unconfoundedness/Conditional Exchangeability:} \( (Y(0), Y(1))  \perp\mkern-9.5mu\perp T~|~X \).
    \item\label{a:Overlap} \textbf{Overlap/Positivity:} \(\exists \eta \in (0,1/2]\) such that \(\eta \le \mathbb{P}[T=1|X] \le 1-\eta\) almost surely.
    \item\label{a:SUTVA_OBS} \textbf{SUTVA:} \(Y = T Y^{(1)} + (1-T)Y^{(0)}\).
    \item\label{a:i.i.d._OBS} \textbf{i.i.d.:} The data is \((X_i, T_i, Y_i^{(0)}, Y_i^{(1)})_{i\in [n]} \stackrel{\textrm{i.i.d.}}{\thicksim} \mathcal{P}\).
\end{enumerate}
\end{assumption}

 \hyperref[a:Unconfoundedness]{Unconfoundedness} means that after accounting for known confounding variables, no hidden factors affect both treatment assignment and outcomes. %This implies that \jj{si besoin de place tu }, given pre-treatment variables, treatment allocation is independent of potential outcomes. 
 It is a relaxed form of \hyperref[a:ignorability]{exchangeability}. %\hyperref[a:Overlap]{Overlap} ensures that the propensity score, \(e(X)=\mathbb{P}[T=1|X]\) is never 0 or 1. 
 %Every participant has a chance of both receiving and not receiving the treatment.% We also still assume  that the data set is \textrm{i.i.d.} and SUTVA assumption. 
%Having established the necessary assumptions and groundwork, we can now proceed to construct estimators of the risk ratio for observational studies. This next step will enable us to better understand and quantify the relationships observed in our data.
%We introduce three estimators for estimating the RR in observational studies. The  \hyperref[RG]{Risk Ratio G-formula} introduced in \Cref{risk_ratio_RCTs} can be used in the context of observational studies, contrary to

RR-N and RR-HT estimators cannot be used in the context of observational studies, since they are built on the assumption of a constant propensity score. But the RR-N estimator can be extended to observational studies as follows.

%can also serve as an estimator in observational studies. Its convergence towards the risk ratio, finite sample or asymptotic bias variance and confidence interval do not rely on assumptions specific to (RCTs). Consequently, \Cref{prop:g_formula} remains valid within the context of \hyperref[a:observational_trial]{observational studies}. This principle applies to all the estimators we will introduce subsequently, as an RCT can essentially be considered a type of observational study.

\subsection{Risk Ratio Inverse Propensity Weighting (RR-IPW)}

Treatment effect in observational studies can be estimated via reweighting individuals by the inverse of their propensity score, thus giving more weights to people who are very likely/unlikely to be treated. Such a method, called Inverse Propensity Weighting \citep[IPW, see][]{Hirano} for estimating the Risk Difference, can be straightforwardly extended to build Risk Ratio estimators. 
%Taking into account the different probabilities of being treated across individuals can be done via 
%The Inverse Propensity Weighting (IPW) \citep{Hirano}, for estimating the RD, can be interpreted as a simple difference-in-means,  where individuals are weighted by the inverse of their propensity score. The same idea can be used  to build the RR-IPW. However, despite the similarities between the approaches, proofs are very different between RD and RR, and much more involved with the ratio. This remarks applies to all estimators presented. 

\begin{definition}[\textbf{RR-IPW}]\label{RIPW}
 Grant \Cref{a:Outcome_positivity} and \Cref{a:observational_trial}.
 Given an estimator \(0 < \hat{e}(\cdot) < 1\) of the propensity score \(e(x) = \mathbb{P}\left[T=1 | X = x\right]\), the Risk Ratio IPW, denoted by \(\hat{\tau}_{\textrm{\tiny RR,IPW,n}}\), is defined as
\[
  \hat{\tau}_{\textrm{\tiny RR,IPW}} = \frac{\frac{1}{n}\sum_{i=1}^n \frac{T_iY_i}{\hat{e}(X_i)}}{\frac{1}{n}\sum_{i=1}^n \frac{(1-T_i)Y_i}{1-\hat{e}(X_i)}}.
\]
\end{definition}
\Cref{prop:ipw_normality_obs}  demonstrates the asymptotic normality of the Oracle Ratio IPW estimator, defined as the RR-IPW but where $\hat{e}(\cdot)$ is replaced by the oracle propensity score $e(\cdot)$. 
%Despite the similarities between the approaches, proofs are very different between RD and RR, and much more involved with the ratio.
%\[
%  \tau_{\textrm{\tiny RR,IPW}}^{\star} = \frac{\sum_{i=1}^n \frac{T_iY_i}{e(X_i)}}{\sum_{i=1}^n \frac{(1-T_i)Y_i}{1-e(X_i)}}.
%\]
\begin{proposition}[\textbf{RR-IPW asymptotic normality}]
\label{prop:ipw_normality_obs}
 Grant Assumptions \ref{a:Outcome_positivity} and \ref{a:observational_trial}. Then the Oracle Risk Ratio IPW is  asymptotically unbiased and satisfies 
 \begin{align}
 \sqrt{n}\left(\tau_{\textrm{\tiny RR,IPW}}^{\star} - \tau_{\textrm{\tiny RR}} \right) \stackrel{d}{\rightarrow} \mathcal{N}\left(0, V_{\textrm{\tiny RR,IPW}} \right)  
 \end{align}
  where \(\frac{V_{\textrm{\tiny RR,IPW}}}{\tau_{\textrm{\tiny RR}}^2} = \mathbb{E}\left[\frac{(Y^{(1)})^2}{e(X)\mathbb{E}\left[Y^{(1)}\right]^2}\right] +\mathbb{E}\left[\frac{(Y^{(0)})^2}{(1-e(X))\mathbb{E}\left[Y^{(0)}\right]^2}\right]\).
\end{proposition}

Note that when the propensity score is constant, one can retrieve the variance of the \hyperref[Var_HT]{RR-HT} as expected. Note also that the asymptotic variance may be large, due to strata on which the propensity score is close to zero or one. In other words, a correct estimation is difficult when some subpopulations are unlikely to be treated (or untreated).

If we assume a logistic model for the true propensity score and estimate it using maximum likelihood estimation (MLE), the variance of the RR-IPW can be derived.
\begin{assumption}\label{a:logistic}
We assume \( \mathbb{E}[X X^\top] \) is positive definite, \( X \) is Sub-Gaussian, and for all \( X \in \mathbb{R}^p \),
\[
\mathbb{P}(T = 1 | X) = \{1 + \exp(-X^\top \beta_\infty^1 - \beta_\infty^0)\}^{-1},
\]
where \( \boldsymbol{\beta}_\infty := (\beta_\infty^0, \beta_\infty^1) \in \mathbb{R}^{p+1} \).
\end{assumption}
For any positive semi-definite matrix $A$ and any vector $X$, let \(\|X\|_A = \sqrt{X^{\top} A X}\). 
\begin{proposition}[Asymptotics of \( \hat{\tau}_{\textrm{RR,IPW}} \) under a logistic model]\label{prop:IPW_MLE}
Under Assumptions \ref{a:observational_trial} and \ref{a:logistic}, the RR-IPW estimator, with the propensity score estimated via MLE, satisfies
\[
\sqrt{n}(\hat{\tau}_{\textrm{RR,IPW}} - \tau_{\textrm{RR}}) \stackrel{d}{\rightarrow} \mathcal{N}(0, V_{\textrm{RR-MLE}}),
\]
with
$
V_{\textrm{RR-MLE}} = V_{\textrm{RR-IPW}} - \tau_{\textrm{RR}}^2 
\bigg\| \frac{c_{10}}{\mathbb{E}[Y(0)]} + \frac{c_{01}}{\mathbb{E}[Y(1)]} \bigg\|_{Q^{-1}}^2,
$\\
where \(
c_{10} = \mathbb{E}[\Tilde{X} \, e(X) Y^{(0)}]\),
\(c_{01} = \mathbb{E}[\Tilde{X}(1 - e(X)) Y^{(1)}]
\) and \(Q = \mathbb{E}[e(X)(1 - e(X)) \Tilde{X}\Tilde{X}^\top]\) with \( \Tilde{X} := (1, X) \).
\end{proposition}
The variance of \(\hat{\tau}_{\textrm{RR,IPW}}\) is notably smaller than that of the oracle estimator \(\tau_{\textrm{RR,IPW}}^\star\). While this might initially seem counterintuitive, similar observations have been made in RCT, as highlighted in studies by \citet{Hirano, Hahn1998efficiencybound, Robins1992EstimatingEE} and in observational studies by \citet{lunceford2004stratification}. Choosing \(\hat{e}\) over \(e\) in the Risk Ratio estimator (thus using \(\hat{\tau}_{\textrm{RR,IPW}}\) instead of \(\tau_{\textrm{RR,IPW}}^\star\)) leads to a reduction in asymptotic variance.
\subsection{Risk Ratio G-formula estimator (RR-G)}

%\es{pour être cohérent, il faudrait remplacer $\mu_0$ et $\mu_1$ par $\mu_{(0)}$ et $\mu_{(1)}$ partout}
For all \((x,t) \in \mathbb{R}^p\times\{0,1\}\), let \(\mu_{(t)}(x) = \mathbb{E}\left[Y^{(t)} | X = x \right]\) be the surface response of the potential outcome.
Assume that we have at our disposal two estimators \(\hat \mu_{(0)}(\cdot)\) and \(\hat \mu_{(1)}(\cdot)\) which respectively estimate \(\mu_{(0)}(\cdot)\) and \(\mu_{(1)}(\cdot)\). We then employ the ratio of these two potential outcome estimations to compute the Risk Ratio. This method, termed the plug-in G-formula or outcome-based modeling, was first introduced by \citet{Robins1986} for the Risk Difference. 

\begin{definition}[\textbf{RR G-formula}]\label{RG}
    Given two estimators \(\hat{\mu}_{(0)}(\cdot)\) and \(\hat{\mu}_{(1)}(\cdot)\), the Risk Ratio G-formula estimator, denoted \(\hat{\tau}_{\textrm{\tiny RR,G}}\), is defined as 
    \begin{align}
      \hat{\tau}_{\textrm{\tiny RR,G}} = \frac{\frac{1}{n}\sum_{i=1}^n \hat \mu_{(1)}(X_i)}{\frac{1}{n}\sum_{i=1}^n \hat \mu_{(0)}(X_i)},
      \end{align}
      if \(\frac{1}{n}\sum_{i=1}^n \hat \mu_{(0)}(X_i) \neq 0\) and zero otherwise.

\end{definition}

The properties of RR-G depend on the estimators $\hat{\mu}_{(0)}$ and $\hat{\mu}_{(1)}$. We analyze in the following the behavior of Oracle Risk Ratio G-formula estimator defined as $
    \tau_{\textrm{\tiny RR,G}}^{\star} = (\frac{1}{n}\sum_{i=1}^n \mu_{(1)}(X_i))/(\frac{1}{n}\sum_{i=1}^n \mu_{(0)}(X_i))$.

\begin{proposition}[\textbf{Asymptotic Normality of \(\tau_{\textrm{\tiny RR,G}}^{\star}\)}]
\label{prop:g_formula}
Grant Assumptions \ref{a:bernoulli_trial} and \ref{a:Outcome_positivity}.
Then, the Oracle Risk Ratio G-formula estimator, \(\tau_{\textrm{\tiny RR,G}}^{\star}\), is asymptotically unbiased and satisfies
\begin{align}
\sqrt{n}\left(\tau_{\textrm{\tiny RR,G}}^{\star} - \tau_{\textrm{\tiny RR}} \right) \stackrel{d}{\rightarrow} \mathcal{N}\left(0, V_{\textrm{\tiny RR,G}} \right),
\end{align}
where \(V_{\textrm{\tiny RR,G}} = \tau_{\textrm{\tiny RR}}^2 \Var\left(\frac{\mu_{(1)}(X)}{\mathbb{E}\left[Y^{(1)}\right]} - \frac{\mu_{(0)}(X)}{\mathbb{E}\left[Y^{(0)}\right]}\right).\)

%Assume for all \(i\), \( M_0 \geq \mu_0^{\star}(x_i) \geq m_0 > 0\) and \(M_1 \geq \mu_1^{\star}(x_i)\). Then, we have:
%\[\left|Bias(\hat{\tau}_{\textrm{\tiny RR, G, n}})\right| \leq \frac{2M_1M_0^2}{nm_0^3} \]\[\left|\Var(\hat{\tau}_{\textrm{\tiny RR, G , n}})\right| \leq \frac{2M_0^2M_1(M_1+M_0)}{nm_0^6}\]
\end{proposition}

\Cref{prop:g_formula} establishes that the Oracle Risk Ratio G-formula estimator is asymptotically normal. Surprisingly, in the case where there is no effect (i.e. $\tau_{\textrm{RR}}=1$), the asymptotic variance is driven by the variance of the Risk Difference on each strata determined by $X$, namely $\Var(\mu_{(1)}(X) - \mu_{(0)}(X))$. By considering the Oracle RR-G instead of RR-G, we remove the additional randomness related to the estimation of the surface responses. It is thus likely that the true variance of RR-G is larger than that of Oracle RR-G, as shown below. 

Assuming a linear model for \( Y^{(t)} \) and estimating both response surfaces  \(\hat \mu_{(0)}\) and \(\hat \mu_{(1)}\) using ordinary least squares, the variance of the RR-G can be derived.
\begin{assumption}[Linear model]\label{a:linear_model}
For all \(t \in \{0,1\}\),
\begin{align*}
    Y^{(t)} = c_{(t)} +  X^{\top} \beta_{(t)}  + \varepsilon_{(t)} \quad \mathbb{E}[X] = \mu\\
 \mathbb{E}[\varepsilon_{(t)} | X] = 0 \quad \text{Var}[\varepsilon_{(t)} | X] = \sigma^2,
\end{align*}
where we assume that $ Y^{(t)} \geq c >0$ for some $c$. 
\end{assumption}
\begin{proposition}[Asymptotic normality of \(\hat{\tau}_{\textrm{RR,OLS}}\)]\label{OLS_RR_OBS}
Grant Assumptions \ref{a:observational_trial} and \ref{a:linear_model}. Then, the RR G-formula estimator $\hat{\tau}_{\textrm{RR,OLS}}$ that uses linear regression to estimate $\mu_{(t)}$ satisfies
\[\sqrt{n}(\hat{\tau}_{\textrm{RR,OLS}}-\tau_{\textrm{RR}}) \stackrel{d}{\rightarrow} \mathcal{N}\left(0, V_{\textrm{RR-OLS}}\right),\]
where, letting $\nu_{t} = \mathbb{E}[X|T=t]$ and $\Sigma_t = \Var(X|T=t)$,
\begin{align*}
& \frac{V_{\textrm{RR-OLS}}}{\tau_{\textrm{\tiny RR}}^2} =  \left\Vert\frac{\beta_{(1)}}{\mathbb{E}\left[Y^{(1)}\right]} - \frac{\beta_{(0)}}{\mathbb{E}\left[Y^{(0)}\right]}\right\Vert_{\Sigma}^2 + \sigma^2 \\
& \times \left(\frac{1+(1-e)^2\|\nu_1 - \nu_0\|^2_{\Sigma_1^{-1}}}{e\mathbb{E}\left[Y^{(1)}\right]^2} + \frac{1+e^2\|\nu_1 - \nu_0\|^2_{\Sigma_0^{-1}}}{(1-e)\mathbb{E}\left[Y^{(0)}\right]^2}\right).
\end{align*}
\end{proposition}
The variance of RR-OLS can be decomposed in two terms: the first term corresponds to the oracle variance of RR-G, that is $V_{\textrm{\tiny RR,G}} /\tau_{\textrm{\tiny RR}}^2$; the second term appears due to the estimation of response surfaces via OLS. Contrary to RR-IPW, the variance of the oracle estimator for the G-formula is smaller than that of the OLS estimator. Note that if we use RR-OLS in a \hyperref[a:bernoulli_trial]{Bernoulli Trial}, then one can show that even in an RCT setting, adjusting for covariates is beneficial as the variance of the RR-OLS is smaller than the variance of RR-N. These results are provided in Appendix~\ref{app:sec_linear_comparison}.

\subsection{Risk Ratio One-step estimator (RR-OS)}

A popular estimator for the RD is the augmented inverse probability weighted estimator  \citep[AIPW, see][]{Robins1992EstimatingEE}. AIPW combines the properties of G-formula and IPW estimators and is \textit{doubly-robust} in the sense that it is consistent as soon as either the propensity or outcome models are correctly specified. By calculating the influence function of the statistical estimand \(\psi_{\textrm{\tiny RD}} = \mathbb{E}\left[\mathbb{E}\left[Y| T= 1, X\right] - \mathbb{E}\left[Y| T= 0, X\right] \right]\) we obtain an efficient estimator: it has no asymptotic bias and the minimal asymptotic variance \citep{kennedy2022semiparametric}. Therefore, to estimate the Risk Ratio (RR), a natural approach is to derive an efficient estimator using semi-parametric theory \citep{Tsiatis}, as presented below. 
\begin{definition}[\textbf{Crossfitted RR-OS}]\label{prop:one_step}
   For all $t \in \{0,1\}$ and all $x$, let $\mu_{(t)}(x) = \mathbb{E}\left[Y^{(0)} | X = x \right]$ and $e_{(t)}(x) = \mathbb{P}\left[T=t | X = x\right]$. We denote \(\mathcal{I} = \{1, \hdots, n\}\), let \(\mathcal{I}_1, \mathcal{I}_2,..., \mathcal{I}_K\) be a partition of \(\mathcal{I}\). Let \(\hat\mu_{(t)}^{\mathcal{I}_{-k}}(X)\) and \(\hat e_{(t)}^{\mathcal{I}_{-k}}(X)\) be  estimators of \(\mu_{(t)}\) and \(e_{(t)}\) built on the sample \(\mathcal{I}_{-k} = \mathcal{I} \backslash \mathcal{I}_{k}\).  
   For all $t \in \{0,1\}$, let 
\begin{align}
    & \hat\tau_{\textrm{\tiny AIPW,t}}  = \frac{1}{n}\sum_{k = 1}^K \sum_{i \in \mathcal{I}_k}\left(\hat\mu_{(t)}^{\mathcal{I}_{-k}}(X_i)+ \frac{Y_i - \hat\mu_{(t)}^{\mathcal{I}_{-k}}(X_i)}{\hat e_{(t)}^{\mathcal{I}_{-k}}(X_i)}\mathds{1}_{T_i=t}\right)\\
    & \textrm{and} \quad 
    \hat\tau_{\textrm{\tiny G,t}}   = \frac{1}{n}\sum_{k = 1}^K \sum_{i \in \mathcal{I}_k}\hat\mu_{(t)}^{\mathcal{I}_{-k}}(X_i),
 \end{align}
    The crossfitted Risk Ratio One-Step (RR-OS) estimator \(\hat{\tau}_{\textrm{\tiny RR-OS}}\) is defined as
    \begin{align*}
        \hat{\tau}_{\textrm{\tiny RR-OS}} &=  \frac{\hat\tau_{\textrm{\tiny G,1}}}{\hat\tau_{\textrm{\tiny G,0}}} \left(1- \frac{\hat\tau_{\textrm{\tiny AIPW,0}}}{\hat\tau_{\textrm{\tiny G,0}}}\right) + \frac{\hat\tau_{\textrm{\tiny AIPW,1}}}{\hat\tau_{\textrm{\tiny G,0}}}.
    \end{align*}
\end{definition}
Considering the statistical estimand \( \psi_{\textrm{\tiny RR}} = \frac{\mathbb{E}\left[\mathbb{E}[Y \mid T = 1, X]\right]}{\mathbb{E}\left[\mathbb{E}[Y \mid T = 0, X]\right]} \), we obtained the estimator RR-OS, which is efficient when nuisance components are estimated via cross-fitting \citep[][]{Chernozhukov} and with non-parametric methods.

\begin{proposition}[\textbf{Asymptotic normality of $\hat{\tau}_{\textrm{\tiny RR-OS}}$}]\label{prop:one_step2}
Grant \Cref{a:Outcome_positivity} and  \Cref{a:observational_trial}. Assume that for all \(1 \leq k \leq K \), and for all $t \in \{0,1\}$, 
\begin{align}\label{double_robust}
\scriptstyle
    \mathbb{E} \left[ \left( \hat{\mu}^{\mathcal{I}_{-k}}_{(t)}(X) - \mu_{(t)}(X) \right)^2 \right] \mathbb{E} \left[ \left( \hat{e}^{\mathcal{I}_{-k}}(X) - e(X) \right)^2 \right] = o \left(\frac{1}{n}\right)
\end{align} 
\vspace{-0.6cm}
\begin{align}
\scriptstyle
\mathbb{E}\left[\hat \mu^{\mathcal{I}_{-k}}_{(0)}(X)\right]  - \mathbb{E}\left[\mu_{(0)}(X))\right] = o\left(n^{-1/4}\right)
\end{align}
\vspace{-0.6cm}
\begin{align}
\scriptstyle
\mathbb{E}\left[(\hat \mu^{\mathcal{I}_{-k}}_{(0)}(X) - \mu_{(0)}(X))^2\right]\mathbb{E}\left[(\hat \mu^{\mathcal{I}_{-k}}_{(1)}(X) - \mu_{(1)}(X))^2\right] = o\left(\frac{1}{n}\right),
\end{align}
with $\eta \leq \hat{e}^{\mathcal{I}_{-k}} (\cdot)\leq 1 -\eta$ (see \Cref{a:observational_trial}). 
Then the One-Step estimator is asymptotically unbiased and satisfies  
\[\sqrt{n}\left(\hat{\tau}_{\textrm{\tiny RR-OS}} - \tau_{\textrm{\tiny RR}} \right) \stackrel{d}{\rightarrow} \mathcal{N}\left(0, V_{\textrm{\tiny RR,OS}} \right),\]
where
\begin{align*}
\scriptstyle \quad& \frac{V_{\textrm{\tiny RR,OS}}}{\tau_{\textrm{\tiny RR}}^2} = \Var\left(\frac{\mu_{(1)}(X)}{\mathbb{E}\left[Y^{(1)}\right]} - \frac{\mu_{(0)}(X)}{\mathbb{E}\left[Y^{(0)}\right]}\right) \\ & \quad +  \mathbb{E}\left[ \frac{\Var\left(Y^{(1)}|X\right)}{e(X)\mathbb{E}\left[Y^{(1)}\right]^2} \right] + \mathbb{E}\left[\frac{\Var\left(Y^{(0)}|X\right)}{(1-e(X))\mathbb{E}\left[Y^{(0)}\right]^2} \right].
\end{align*}
\end{proposition}

This estimator is efficient: its asymptotic variance is minimal. 
%exhibits the lowest mean squared error among all estimators and therefore the smallest variance among all unbiased estimators. 
The semi parametric theory develops efficient estimators by compensating for the first-order bias \citep{kennedy2022semiparametric}, this can be achieved either by estimating and subtracting the first-order bias, leading to the RR-OS estimator or by finding values for the target parameter and nuisance parameters that solve the estimating equation  \citep[see][for the RD case]{Schuler}, resulting in RR-AIPW presented below (calculations are detailed in Appendix~\ref{proof:aipw_normality_obs}).
\subsection{Risk Ratio Augmented Inverse Propensity Weighting (RR-AIPW)}
\begin{definition}[\textbf{Crossfitted RR-AIPW}]\label{RAIPW}
The Risk Ratio AIPW crossfitted is defined as
    \begin{align*}
        \hat\tau_{\textrm{\tiny RR,AIPW}} &:=\frac{\hat\tau_{\textrm{\tiny AIPW,1}}}{\hat\tau_{\textrm{\tiny AIPW,0}}},
    \end{align*}
        where \(\hat\tau_{\textrm{\tiny AIPW,0}}\), and \(\hat\tau_{\textrm{\tiny AIPW,1}}\) are defined in \ref{prop:one_step}.
\end{definition}
The \hyperref[RAIPW]{RR-AIPW} is simply the ratio of two one-step estimators, one for \(\mathbb{E}\left[Y^{(1)}\right]\) and one for \(\mathbb{E}\left[Y^{(0)}\right]\). This method may seem simplistic at first glance, since approximating both the numerator and denominator usually results in a non-zero asymptotic bias. However, RR-AIPW is derived via the estimating equation method using influence function theory, which results in an efficient (asymptotically unbiased) estimator. Note that in the case of the Risk Difference (RD), both approaches (One-step bias correction and estimating equation) yield the same AIPW estimator. However, because our statistical estimand for the Risk Ratio is nonlinear, the resulting estimators differ. It remains that they are both efficient, as shown below.

\begin{proposition}[\textbf{RR AIPW asymptotic normality}]
    \label{prop:aipw_normality_obs}
    Grant Assumptions \ref{a:Outcome_positivity} and \ref{a:observational_trial}. Assume that \eqref{double_robust} holds and that, for all \(1 \leq k \leq K\) and all $t \in \{0,1\}$,
\begin{align}
\scriptstyle
    \mathbb{E}\left[\left(\hat \mu_{(t)}^{\mathcal{I}_{-k}}\left(X\right)-\mu_{(t)}\left(X\right)\right)^2 \right] = o(1),\quad \mathbb{E}\left[ \left(\hat{e}^{\mathcal{I}_{-k}}\left(X\right)-e\left(X\right)\right)^2 \right] = o(1),
\end{align}  
with $\eta \leq \hat{e}^{\mathcal{I}_{-k}} (\cdot)\leq 1 -\eta$.
    Then, the crossfitted Risk Ratio AIPW estimator is asymptotically unbiased and satisfies  
    \[\sqrt{n}\left(\hat\tau_{\textrm{\tiny RR,AIPW}} - \tau_{\textrm{\tiny RR}} \right) \stackrel{d}{\rightarrow} \mathcal{N}\left(0, V_{\textrm{\tiny RR,OS}} \right),
    \]
 where \(V_{\textrm{\tiny RR,OS}}\) is defined in \Cref{prop:one_step2}.
\end{proposition}
 
Assumptions in \Cref{prop:aipw_normality_obs} are the same as those used in the Risk Difference AIPW estimator \citep{wager2020stats} to achieve double robustness. Specifically, RR-AIPW benefits from \textit{weak} double robustness, meaning consistency is maintained as long as either the outcome model or the propensity score model is estimated consistently. This contrasts with RR-OS, which lacks this flexibility: it requires consistency of both outcome models $(\mu_0, \mu_1)$ or the joint consistency of $(e, \mu_0)$ to achieve reliable inference. Additionally, condition \ref{double_robust} often referred to as a risk decay condition ensures \textit{strong} double robustness for RR-AIPW. This property ensures asymptotic normality when both the propensity score and outcome regression converge at sufficiently fast rates. We recommend RR-AIPW over RR-OS, not only because both share identical efficiency in their asymptotic distributions, but also because RR-AIPW operates under weaker assumptions. RR-AIPW’s strong and weak double robustness properties provide greater resilience to model misspecification compared to RR-OS.
%Additionally, Assumption \ref{double_robust}, often referred to as risk decay, holds when either surfaces responses or propensity score achieve a parametric rate, while the other is only consistent.  This departs from RR-OS where asymptotic normality is not achieved when only the propensity score has a parametric convergence rate. We recommend RR-AIPW over RR-OS due to identical efficient asymptotic properties but weaker assumptions for RR-AIPW: while RR-OS depends on either \((e, \mu_0)\) or \((\mu_1, \mu_0)\) to achieve consistency, RR-AIPW relies solely on \(e\) or \((\mu_1, \mu_0)\) to be consistent since AIPW is weakly doubly robust \cite{wager2020stats}.

\section{Simulation}
\label{sec:simulations}
\begin{figure*}[h!]
  \includegraphics[width=\textwidth]{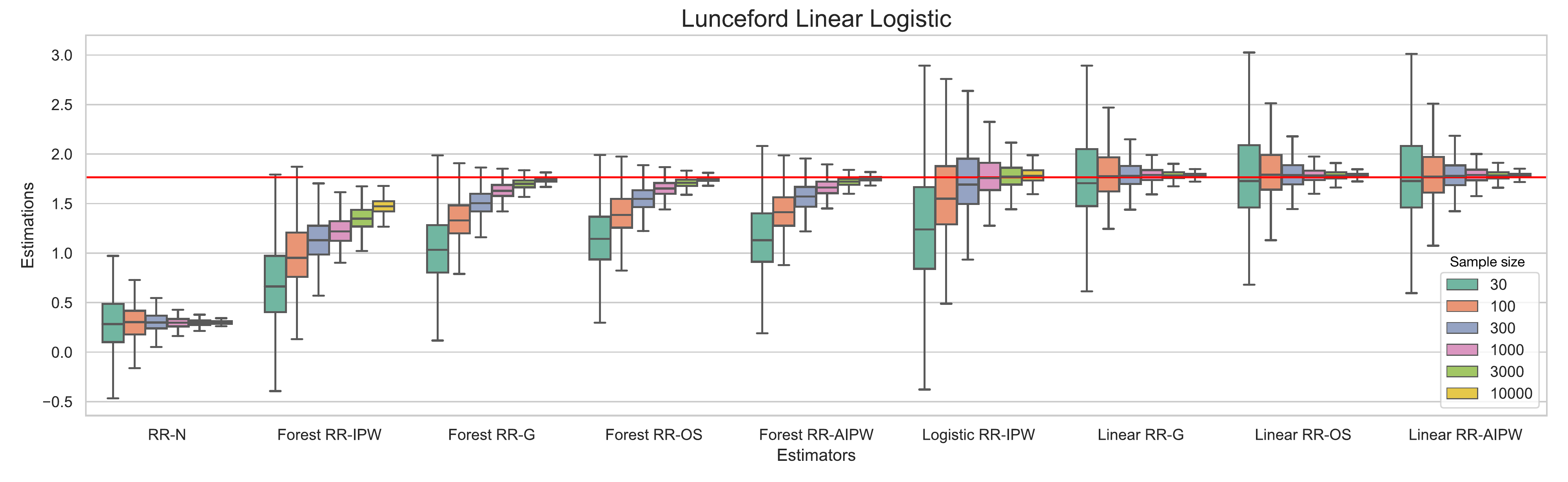}
  \caption{Risk Ratio estimators computed for a Linear/Logistic DGP, with $3000$ repetitions.}
\label{fig:Lunceford}
\end{figure*}
\begin{figure*}[h]
  \includegraphics[width=\textwidth]{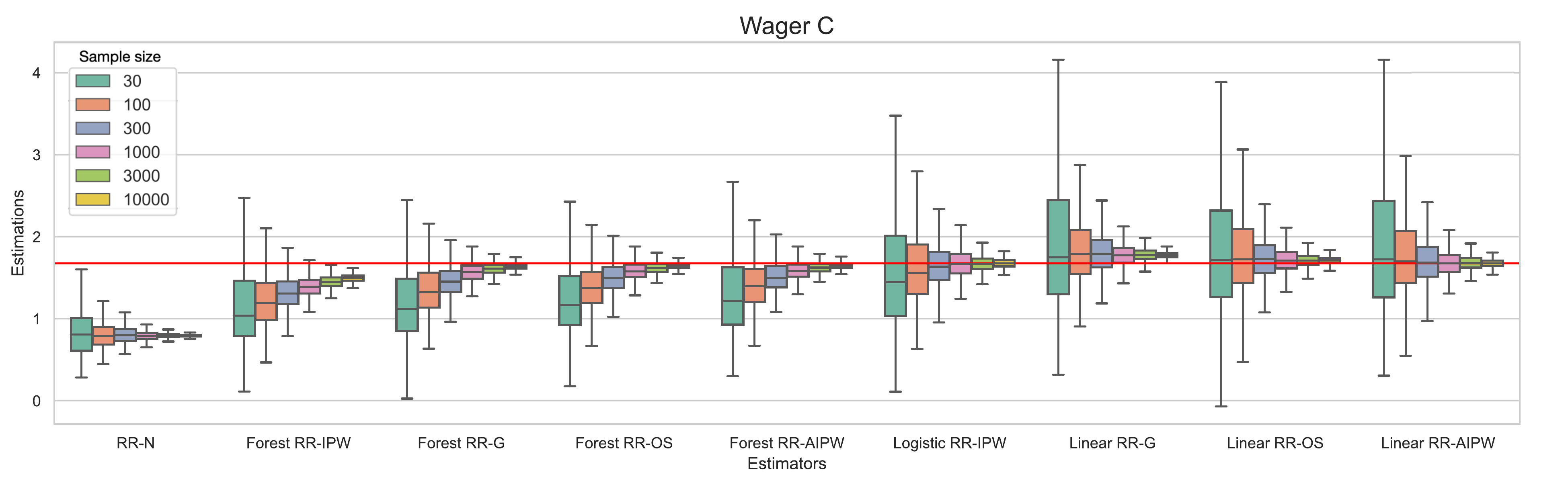}
    \caption{Risk Ratio estimators computed for a non-Linear-Logistic DGP, with $3000$ repetitions.}
    %Estimations of the Risk Ratio with weighting, outcome based and augmented estimators as a function of the sample size for the non-Linear-Logistic DGP. Parametric (Linear) and non-parametric (Forest) estimations of nuisance are displayed.}
    \label{fig:Wager_C}
\end{figure*}
Simulations for RCT are provided in Appendix \ref{Simulation-2}.
For observational studies, we generate datasets \((X, T, Y^{(0)}, Y^{(1)})\) according to the general model
\begin{align*}
    \begin{array}{ll}
    Y^{(1)} = m(X) + b(X) + \varepsilon_{(1)}  &  \mathbb{P}\left[T=1|X\right] = e(X),\\
    Y^{(0)} = b(X) + \varepsilon_{(0)}  & \text{with } \varepsilon_{(t)}\sim \mathcal{N}\left(0,\sigma^2\right).
 \end{array},
\end{align*}
 where $m(.)$, \(b(.)\), and \(e(.)\)  respectively correspond to the treatment effect, the baseline and the propensity score. The true Risk Ratio can be expressed as $\tau_{\textrm{RR}} = \mathbb{E}\left[Y^{(1)}\right]/\mathbb{E}\left[Y^{(0)}\right] = \mathbb{E}\left[m(X)\right]/\mathbb{E}\left[b(X)\right]+1$.
We compare the performances of all estimators defined in \Cref{risk_ratio_OBS} where nuisance components (regression surfaces and propensity score) are estimated via parametric (linear/logistic regression) or non-parametric methods (random forests).  %Each simulation is repeated 3000 times. 
More details  are provided in Appendix~\ref{Simulation-2}.

%\subsection{Observational studies}
\subsection{Linear and Logistic DGP}\label{sim:LL_DGP}
The first observational data-generating process (DGP), introduced in \citet{Lunceford}, uses linear outcome models (treatment effect and baseline) and a logistic propensity score, defined as \(m(X,V) = 2\) with:  
\begin{align*}
    \begin{array}{llll}
        b(X, V) = \beta_0^{\top} [X, V], \quad e(X) = (1 + \exp(-\beta_e^\top X))^{-1} \\
        \beta_0 = (-1, 1, -1, -1, 1, 1) \quad \beta_e = (-0.6, 0.6, -0.6).
    \end{array} 
\end{align*}
Covariates \(X = (X_1, X_2, X_3)^\top\) influence treatment and response, while \(V = (V_1, V_2, V_3)^\top\) only affect the response. \([X, V]\) are jointly distributed, with \(X_3 \sim \text{Bernoulli}(0.2)\) and \(V_3 \sim \text{Bernoulli}(P(V_3 = 1 \mid X_3) = 0.75X_3 + 0.25(1-X_3))\). Conditional on \(X_3\),
\((X_1, V_1, X_2, V_2)^\top \sim \mathcal{N}(\lambda_{X_3}, \Sigma)\), where:
\[
\lambda_1 = -\lambda_0 = \begin{bmatrix}
1 \\ 
1 \\ 
-1 \\ 
-1 
\end{bmatrix},
\Sigma = \begin{bmatrix}
1 & 0.5 & -0.5 & -0.5 \\
0.5 & 1 & -0.5 & -0.5 \\
-0.5 & -0.5 & 1 & 0.5 \\
-0.5 & -0.5 & 0.5 & 1
\end{bmatrix}
\]
Results are depicted in \Cref{fig:Lunceford}. Only confounding variables are used as inputs in the different estimators.
As expected, since the generative process is linear, methods that use  parametric estimators (logistic/linear regression) outperform those using non-parametric approaches (random forests) in finite-sample settings. While all methods (except RR-N) converge to the correct RR, methods based on parametric estimators exhibit a faster rate of convergence and are unbiased (except for Logistic RR-IPW) even for small sample sizes. Indeed, random forests are not suited for linear generative process and require here more than  10000 samples to estimate correctly the RR.  
%This is because estimators using non-parametric have a slow convergence rate to estimate linear quantities. Consequently, in these simulations, they require more than 
%On the other hand, parametric method estimators are almost all unbiaised except for the Logistic RIPW, for small sample sizes and converge quickly. This is expected given data generation process. 
All in all, when the outcome modelling and the propensity scores are linear, the two doubly-robust estimators (RR-AIPW and RR-OS) and the RR-G, all based on linear estimators, achieve the best performances: they are unbiased, even for small sample sizes, and converge quickly to the true RR. 
\begin{figure*}[h!]
  \includegraphics[width=\textwidth,height=4cm]{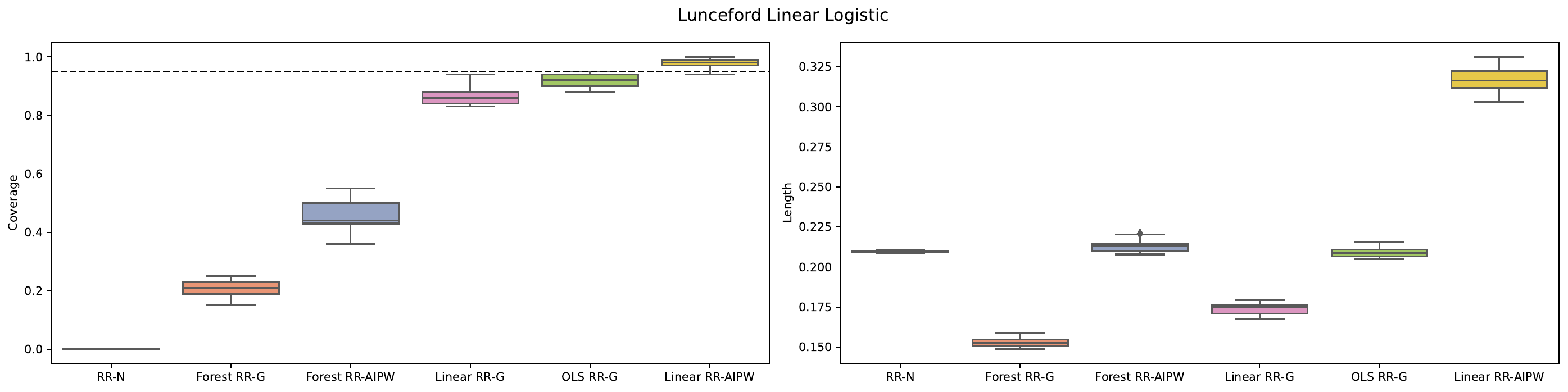}
  \caption{Coverage (left) and Length (right) of asymptotic CI derived from \Cref{risk_ratio_RCTs} and \Cref{risk_ratio_OBS} with $n=1000$ and $300$ repetitions.}
\label{fig:Lunceford_CI}
\end{figure*}

\begin{figure*}[h!]
  \includegraphics[width=\textwidth,height=4cm]{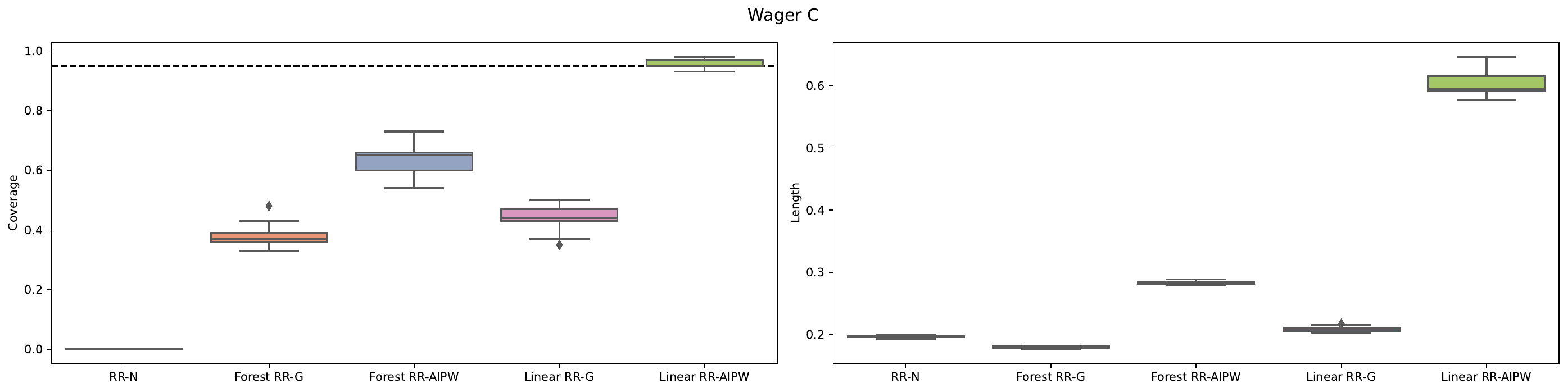}
    \caption{Average coverage (left) and average length (right) of asymptotic confidence interval derived from \Cref{risk_ratio_RCTs} and \Cref{risk_ratio_OBS} for different estimators with $n=1000$ and $300$ repetitions for a Non-Linear and Logistic DGP.}
    \label{fig:Wager_CI_C}
\end{figure*}

 \begin{figure*}[h!!]
\includegraphics[width=\textwidth]{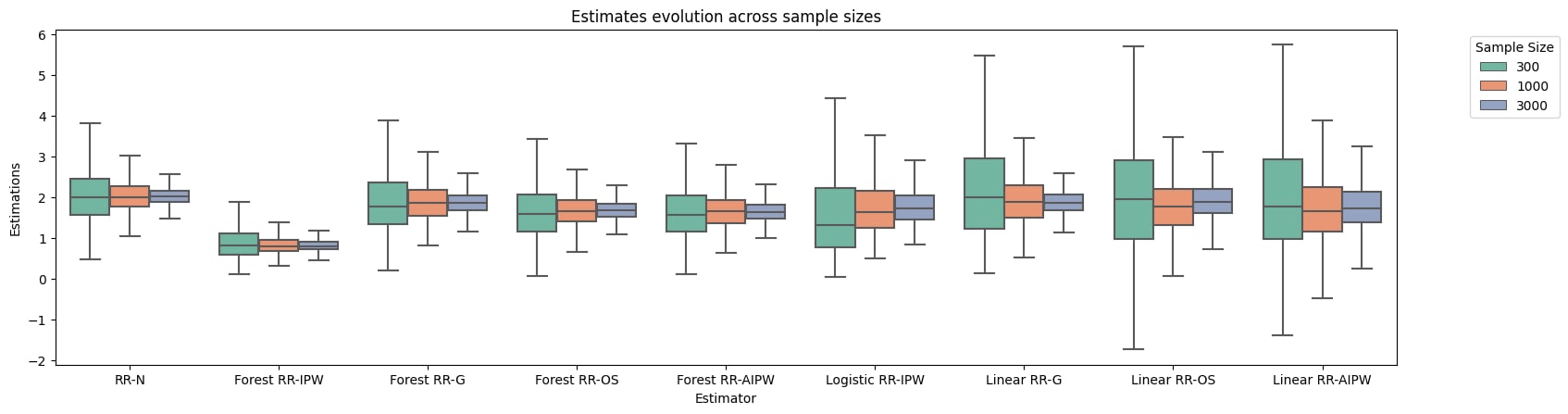}
    \caption{RR estimations with weighting, outcome based and augmented estimators as a function of the sample size for the Traumabase. Parametric (Linear) and non parametric (Forest) estimations of nuisance are displayed.}
    \label{fig:Traumabase}
\end{figure*}
%Note also that the RIPW generally exhibits greater variance compared to other estimators. This increased variance can be attributed to the fact that this estimator involves division by a probability, which is a small number, thereby contributing to its higher variance. \es{je ne vois pas vraiment ça sur les simus}
Furthermore, both Linear RR-OS and RR-AIPW estimators give very similar results. More simulations on a Non-Linear and (Non-)Logistic DGPs are available in \ref{sim:CI}. As anticipated, estimators with non-parametric nuisance functions outperform their parametric counterparts when the DGP deviates from linear or logistic structures.

\subsubsection{Non-Linear and Logistic DGP}

We use a semi-parametric setup \citep[see][]{Wager} with non-linear baseline models, a constant treatment effect and a logistic propensity score:
\begin{align*}
        b(X) &= 2 \log \left(1+e^{X_{1}+X_{2}+X_{3}}\right),\\ 
        e(X) &= 1 /\left(1+e^{X_{2}+X_{3}}\right) \quad \text{and} \quad  m(X) = 1,
\end{align*}
where $X \sim \mathcal{N}\left(0, I_{d \times d}\right)$. Results are presented in \Cref{fig:Wager_C}. 
 The Forest IPW and Linear G-formula estimators yield poor RR estimates for the largest sample size. The forest IPW uses random forests to estimate a logistic model, which may still converge, but at a slower rate than other methods. The Linear G-formula employs linear regressions to estimate the response surfaces, potentially leading to an irreducible asymptotic bias.
The Forest RR-G, Forest RR-AIPW, and Forest RR-OS estimators converge slowly to the true RR. This simulation highlights the doubly robust properties of the Linear RR-AIPW and Linear RR-OS estimators: they target the true RR even at small sample sizes, as they have at least one well-specified model. 

\subsection{Confidence intervals (CI)}
\label{subsec:CI}

In both the Linear/Logistic and Non-Linear/Logistic data-generating processes (DGPs), we build asymptotic 95\% CI for the RR-AIPW, RR-G and RR-N estimators based on their asymptotic normality. Variances were estimated via equation \ref{V_RR_OS}, \ref{V_RR-HT}, \ref{V_RR-IPW}, \ref{V_RR,G} and \ref{V_R-N}. In \Cref{fig:Lunceford_CI}, we present the distribution of the length and coverage (probability that the CI contains the risk ratio) for each estimator (300 repetitions). IPW was excluded for its poor performance (excessively large CIs), caused by propensity scores close to 0 or 1 (see \Cref{fig:hist_e}). As expected, RR-N CI has nearly zero coverage due to non-RCT setting. The Forest RR-G and RR-AIPW confidence intervals also exhibit poor coverage, which is in agreement with the Linear/Logistic DGP.
In contrast, Linear RR-G and RR-AIPW demonstrate good coverage. Note that the CI for OLS RR-G, built on \Cref{OLS_RR_OBS}, has a better coverage than the Linear RR-G method, as it includes the uncertainty due to linear estimations. Although only the RR AIPW has coverage above 95\%, the OLS RR-G has a shorter average predicted length compared to the Linear RR AIPW. 

Turning to the Non-Linear/Logistic DGP, \Cref{fig:Wager_CI_C} shows the CI length and coverage for the same estimators. In this setting, only the Linear RR-AIPW estimator maintains acceptable coverage—likely because propensity scores are well-estimated, whereas modeling complex, non-linear response surfaces is difficult with either linear models or forest-based methods when $n = 1000$. As before, Forest RR-G and Forest RR-AIPW struggle with coverage, and the RR-N estimator again shows almost no coverage, reflecting its limitations in observational settings. More simulations on a Non-Linear/Non-Logistic DGP are available in Appendix~\ref{sim:CI}.

\section{Real-World Experiment}\label{sec:real_world_exp}

To better illustrate the practical application and behavior of our estimators, we include a real-world study from \citet{mayer2020doublyrobusttreatmenteffect} involving 8,270 patients with traumatic brain injury (TBI), using data extracted from the Traumabase. The Traumabase is a continuously updated database that collects comprehensive clinical data from the scene of the accident through to hospital discharge. 
%This study investigates the effect of tranexamic acid (TXA) on mortality. %(\(Y \in \{0,1\}\), where \(Y=1\) indicates death). 
The causal effect of tranexamic acid (TXA) on 28-day mortality
%both 28-day intrahospital TBI-related mortality and all-cause intrahospital mortality 
is estimated by adjusting for 17 confounding variables. These variables include key metrics for severe trauma cases, such as systolic and diastolic blood pressure, heart rate, oxygen saturation, and details of interventions.
%like catecholamine administration. 
We subsample the real data to obtain different sample sizes, with results averaged over 3,000 simulations for each sample size. Since this is a real dataset, the true value of \(\tau_{RR}\) is unknown. Semi-synthetic simulations, in which the treatment and potential outcomes are generated (allowing us to know the true RR value), are described in detail in \ref{an:semi-synthetic}. Results are displayed in \Cref{fig:Traumabase}.

As observed, most estimators yield values larger than one, suggesting a potential deleterious effect of TXA (increased mortality). This trend aligns with findings from previous studies (e.g., \cite{mayer2020doublyrobusttreatmenteffect}). Notably, Forest RR-IPW is the only estimator indicating a beneficial effect of the treatment. RR-N, which is not suited for use with observational data, produces a higher value compared to other estimators. Interestingly, estimators based on linear models (Linear RR-AIPW, Linear RR-G, Linear RR-OS) exhibit the largest variances, particularly for small sample sizes.

\section{Conclusion}

Quantifying treatment effects presents challenges, since different measures may lead to different understanding of the same phenomenon. In our study, we focus on one of these measures, the Risk Ratio and introduced several estimators, valid in RCT or observational studies. Using dedicated mathematical tools (influence function theory, M-estimation), we establish their asymptotic normality, limiting variance and derive asymptotic confidence intervals. Empirical evaluations show that RR-N and RR-IPW have poor performances.  Either Linear or Forest RR-AIPW (or RR-OS) show similar (good) behaviors to estimate the Risk Ratio, with the best theoretical guarantees among all studied estimators. Since RR-AIPW requires fewer assumption and is simpler to compute, we would recommend to use RR-AIPW. As for the doubly-robust approaches, G-formula is competitive, with performances that depend on the setting and the estimation method used for the nuisance components. Identifying guidelines establishing when linear nuisance components should be used instead of non-parametric ones still remains an open problem. In practice, observational studies may be used to generalize the treatment effect from a RCT population to the general population of interest. Our work is a first step toward proposing procedures to generalize the Risk Ratio to general populations.

\section*{Acknowledgments}
This work has been done in the frame of the PEPR SN SMATCH project and has benefited from a governmental grant managed by the Agence Nationale de la Recherche under the France 2030 programme, reference ANR-22-PESN-0003. 

\section*{Impact Statement}
This paper proposes and studies estimators for the Risk Ratio in observational studies, addressing a methodological gap in causal inference. The Risk Ratio is widely used in clinical and epidemiological research to assess relative treatment effects. However, it is important to consider the Risk Ratio alongside other effect measures, as each provides a different perspective on the data and may have different implications for practice. As these methods rely on standard causal assumptions such as ignorability—which can be difficult to satisfy in practice—caution is needed when interpreting the results. Violations of these assumptions may lead to biased estimates and incorrect conclusions about treatment effects.

\bibliography{example_paper}
\bibliographystyle{icml2025}
\newpage
\onecolumn

\section{Appendix}
\subsection{Preliminary results}

Since we are studying the asymptotic properties of the risk ratio, we cannot directly apply a central limit theorem as in \cite{wager2020stats}. We will therefore rely on \Cref{Th1} to prove most of our asymptotic results.

\begin{theorem}[\textbf{Asymptotic normality of the ratio of two estimators}]\label{Th1}
    Let \((Z_1, \hdots , Z_n)\) be \(n\) i.i.d. random variables, \(g_0\) and \(g_1\) two functions square integrable such that \(\mathbb{E}\left[g_0(Z_i)\right] = \tau_0\) and \(\mathbb{E}\left[g_1(Z_i)\right] = \tau_1\), where \(\tau_0 \neq 0\). Then, we have that
    \[\sqrt{n}\left(\frac{\sum_{i=1}^n g_1(Z_i)}{\sum_{i=1}^n g_0(Z_i)} - \frac{\tau_1}{\tau_0} \right) \stackrel{d}{\rightarrow} \mathcal{N}\left(0, V_{\textrm{\tiny RR}}^{\star} \right),\]
    where \[V_{\textrm{\tiny RR}}^{\star} = \left(\frac{\tau_1}{\tau_0}\right)^2\Var\left(\frac{g_1(Z)}{\tau_1} - \frac{g_0(Z)}{\tau_0}\right).\]
\end{theorem}
    
\begin{proof}
    We rely on M-estimation theory to prove \Cref{Th1}. Let  
    \begin{align}
    \hat{\boldsymbol{\theta}}_n & =\left(\begin{array}{c}
    \frac{1}{n}\sum_{i=1}^n g_0(Z_i) \\
    \frac{1}{n}\sum_{i=1}^n g_1(Z_i) \\
    \frac{\sum_{i=1}^n g_1(Z_i)}{\sum_{i=1}^n g_0(Z_i)}
    \end{array}\right) \quad \textrm{and} \quad 
    \psi(Z, \boldsymbol{\theta}) =\left(\begin{array}{c}
    g_0(Z)-\theta_{0} \\
    g_1(Z)-\theta_{1} \\
    \theta_1 - \theta_2\theta_0
    \end{array}\right),
    \end{align}
    where \(\boldsymbol{\theta} = (\theta_0, \theta_1, \theta_2)\). 
   We have that
    \begin{align*}
        \sum_{i=1}^n \left(g_0(Z_i) - \frac{1}{n}\sum_{j=1}^n g_0(Z_j)\right) &= \sum_{i=1}^n g_0(Z_i) - \sum_{j=1}^n g_0(Z_j) = 0,
    \end{align*}
    and similarly 
    \begin{align*}
        \sum_{i=1}^n \left(g_1(Z_i) - \frac{1}{n}\sum_{j=1}^n g_1(Z_j)\right) &= \sum_{i=1}^n g_1(Z_i) - \sum_{j=1}^n g_1(Z_j) = 0.
    \end{align*}
    Besides,     
    \begin{align*}
    \sum_{i=1}^n \left(\frac{1}{n}\sum_{j=1}^n g_1(Z_j) - \frac{\sum_{j=1}^n g_1(Z_j)}{\sum_{j=1}^n g_0(Z_j)}\frac{1}{n}\sum_{j=1}^n g_0(Z_j)\right) = 0.
    \end{align*}
    Gathering the three previous equalities, we obtain
    \begin{align}
    \sum_{i=1}^n \psi(Z_i, \boldsymbol{\theta}_n) = 0,    
    \end{align}
    which proves that \( \hat{\boldsymbol{\theta}}_n\) is an M-estimator of type \(\psi\) \citep[see][]{Stefanski2002Mestimation}. 
    Furthermore, letting \(\boldsymbol{\theta}_{\infty} = (\tau_0, \tau_1, \tau_1/\tau_0)\), simple calculations show that 
    \begin{align}
        \mathbb{E}\left[\psi(Z, \boldsymbol{\theta}_\infty) \right] = 0. \label{eq_proof1}
    \end{align}
    Since the first two components of \(\psi\) are linear with respect to \(\theta_0\) and \(\theta_1\) and since the third component is linear with respect to \(\theta_2\), \(\boldsymbol{\theta}_{\infty}\) defined above is the only value satisfying \eqref{eq_proof1}. Define
    \begin{align}
      A\left(\theta_\infty \right) = \mathbb{E}\left[\frac{\partial \psi}{\partial \theta}|_{\theta=\theta_\infty}\right] \quad \textrm{and} \quad B(\theta_\infty) =  \mathbb{E}\left[\psi(Z, \theta_\infty) \psi(Z, \theta_\infty)^T\right].
    \end{align}
    We now check the conditions of Theorem~7.2 in  \citet{Stefanski2002Mestimation}. First, let us compute \(A\left(\theta_\infty \right) \) and \(B\left(\theta_\infty \right) \). Since
     \begin{align}
     \frac{\partial \psi}{\partial \theta} (Z, \theta)=\left(\begin{array}{ccc}
    -1 & 0 & 0 \\
    0 & -1 & 0 \\
    -\theta_2 & 1 & -\theta_0
    \end{array}\right),  
    \end{align}
    we obtain
    \begin{align} 
    A\left(\boldsymbol{\theta_\infty} \right)= \left(\begin{array}{ccc}
    -1 & 0 & 0 \\
    0 & -1 & 0 \\
    -\frac{\tau_1}{\tau_0} & 1 & -\tau_0
    \end{array}\right),
    \end{align}
    which leads to 
    \begin{align}
    A^{-1}\left(\boldsymbol{\theta_\infty} \right)= \left(\begin{array}{ccc}
    -1 & 0 & 0 \\
    0 & -1 & 0 \\
    \frac{\tau_1}{\tau_0^2} & -\frac{1}{\tau_0} & -\frac{1}{\tau_0}
    \end{array}\right).
    \end{align}
    Regarding \(B\left(\theta_\infty \right) \), elementary calculations show that 
    \begin{align*}
            \psi(Z, \theta_\infty) \psi(Z, \theta_\infty)^T
            & = \left(\begin{array}{ccccc}
        \left(g_0(Z) -\tau_0\right)^2 & \left(g_0(Z) -\tau_0\right)\left(g_1(Z) -\tau_1\right) & 0 \\
        \left(g_0(Z) -\tau_0\right)\left(g_1(Z) -\tau_1\right) & 	\left(g_1(Z) -\tau_1\right)^2 & 0 \\
        0 & 0 & 0
        \end{array}\right),
        \end{align*}
    which leads to 
           \begin{align*}
            B(\theta_\infty)
            & = \left(\begin{array}{ccccc}
        \Var\left[g_0(Z)\right] & \Cov\left(g_0(Z), g_1(Z)\right) & 0 \\
        \Cov\left(g_0(Z), g_1(Z)\right) & 	\Var\left[g_1(Z)\right] & 0 \\
        0 & 0 & 0
        \end{array}\right).
        \end{align*}
    Based on the previous calculations, we have
    \begin{itemize}
        \item  \(\psi(z,\boldsymbol{\theta})\)  and its first two partial derivatives with respect to \(\boldsymbol{\theta}\) exist for all \(z\) and for all \(\boldsymbol{\theta}\) in the neighborhood of \(\boldsymbol{\theta_\infty}\).
        \item For each \(\boldsymbol{\theta}\) in the neighborhood of \(\boldsymbol{\theta_\infty}\), we have for all \(i,j,k \in \{0, 2\}\):
        \[\left| \frac{\partial^2}{\partial \theta_i \partial \theta_j} \psi_k(z, \boldsymbol{\theta}) \right| \leq 1\]
        and 1 is integrable.
        \item \(A(\theta_\infty)\) exists and is nonsingular.
        \item \(B(\theta_\infty)\) exists and is finite.
    \end{itemize}
    
    Since we have: 
    \[\sum_{i=1}^n \psi(T_i, Y_i, \hat{\boldsymbol{\theta}}_n) = 0 \quad \text{and} \quad \hat{\boldsymbol{\theta}}_n \stackrel{p}{\rightarrow} \theta_\infty .\]
    Then the conditions of Theorem~7.2 in  \citet{Stefanski2002Mestimation} are satisfied, we have
        \[\sqrt{n}\left(  \hat{\boldsymbol{\theta}}_n - \theta_\infty \right) \stackrel{d}{\rightarrow} \mathcal{N}\left(0, A(\theta_\infty)^{-1}B(\theta_\infty)(A(\theta_\infty)^{-1})^{\top} \right),\]
    where 
    \begin{align}
    &
    A(\theta_\infty)^{-1}B(\theta_\infty)(A(\theta_\infty)^{-1})^{\top} \nonumber \\
        & = \scriptstyle \left[\begin{matrix}\Var\left[g_0(Z)\right] & \Cov\left(g_0(Z), g_1(Z)\right) & \frac{\Cov\left(g_0(Z), g_1(Z)\right)}{\tau_{0}} - \frac{\tau_{1} \Var\left[g_0(Z)\right]}{\tau_{0}^{2}}\\\Cov\left(g_0(Z), g_1(Z)\right) & \Var\left[g_1(Z)\right] & - \frac{\Cov\left(g_0(Z), g_1(Z)\right) \tau_{1}}{\tau_{0}^{2}} + \frac{\Var\left[g_1(Z)\right]}{\tau_{0}}\\ \frac{\Cov\left(g_0(Z), g_1(Z)\right)}{\tau_{0}} - \frac{\tau_{1} \Var\left[g_0(Z)\right]}{\tau_{0}^{2}} & - \frac{\Cov\left(g_0(Z), g_1(Z)\right) \tau_{1}}{\tau_{0}^{2}} + \frac{\Var\left[g_1(Z)\right]}{\tau_{0}} & V_{\textrm{\tiny RR}}^{\star} \end{matrix}\right],
    \end{align}
    with 
    \begin{align}
    V_{\textrm{\tiny RR}}^{\star} = \left(\frac{\tau_1}{\tau_0}\right)^2\Var\left(\frac{g_1(Z)}{\tau_1} - \frac{g_0(Z)}{\tau_0}\right).    
    \end{align}
    In particular, 
    \begin{align}
    \sqrt{n}\left(\frac{\sum_{i=1}^n g_1(Z_i)}{\sum_{i=1}^n g_0(Z_i)} - \frac{\tau_1}{\tau_0} \right) \stackrel{d}{\rightarrow} \mathcal{N}\left(0, V_{\textrm{\tiny RR}}^{\star} \right).    
    \end{align} 
    
\end{proof}

\begin{assumption}\label{a:Outcome_positivity1}
    We suppose that both \(Y^{(0)}\) and \(Y^{(1)}\) are squared integrable and that $ \mathbb{E}\left[Y^{(0)} | X\right], \mathbb{E}\left[Y^{(1)} | X\right] > 0.$
\end{assumption}

\begin{theorem}[\textbf{Finite sample bias and variance of the ratio of two estimators}]\label{Th2}    Let \(T_1(\boldsymbol{Z})\) and \(T_0(\boldsymbol{Z})\) be two unbiased estimators of \(\tau_1\) and \(\tau_0 > 0\) where \(\boldsymbol{Z} = (Z_1, \hdots , Z_n)\) be \(n\) i.i.d. random variables. We assume that \(M_0 \geq T_0(\boldsymbol{Z}) \geq m_0 > 0\), \(|T_1(\boldsymbol{Z})| \leq M_1\). We also assume that \(\Var({T_1(\boldsymbol{Z})})= O_p\left(\frac{1}{n}\right)\) and  \(\Var({T_0(\boldsymbol{Z})})= O_p\left(\frac{1}{n}\right)\) Then, we have that
    \[\operatorname{Bias}\left(\frac{T_1(\boldsymbol{Z})}{T_0(\boldsymbol{Z})}, \frac{\tau_1}{\tau_0}\right) = \left|\mathbb{E}\left[\frac{T_1(\boldsymbol{Z})}{T_0(\boldsymbol{Z})}\right] - \frac{\tau_1}{\tau_0} \right| \leq   \frac{M_1M_0}{nm_0^2}\left(\frac{M_0}{m_0} + 1\right),\]
    and
    \[\left|\Var\left(\frac{T_1(\boldsymbol{Z})}{T_0(\boldsymbol{Z})}\right) - \left(\frac{\tau_1}{\tau_0}\right)^2\Var\left(\frac{T_1(\boldsymbol{Z})}{\tau_1}- \frac{T_0(\boldsymbol{Z})}{\tau_0}\right)\right|\leq\frac{2M_0M_1}{nm_0^4}\left(\frac{M_0M_1}{m_0^2}+1\right)\].
\end{theorem}

\begin{proof}
    We rely on the multivariate version of Taylor's theorem to prove \Cref{Th2}. We first introduce the multi-index notation:
    \begin{align*}
        |\alpha| = \alpha_1+\cdots+\alpha_n, \quad \alpha!=\alpha_1!\cdots\alpha_n!, \quad \boldsymbol{x}^\alpha=x_1^{\alpha_1}\cdots x_n^{\alpha_n}
    \end{align*}
    and
    \begin{align*}
        D^\alpha f = \frac{\partial^{|\alpha|}f}{\partial x_1^{\alpha_1}\cdots \partial x_n^{\alpha_n}}, \qquad |\alpha|\leq k
    \end{align*}
    Let \(f\) be the ratio function
    \begin{align*}
        f \colon \mathbb{R}_+^* \times \mathbb{R}_+^* & \longrightarrow \mathbb{R} \\
        (x_1, x_2) & \longmapsto x_1/x_2.
    \end{align*}
    Since \(f\) is two times continuously differentiable then one can derive an exact formula for the remainder in terms of 2nd order partial derivatives of \(f\). Namely, if we define \(\boldsymbol{x} = (x_1, x_2)\) and for \(\boldsymbol{a} \in \mathbb{R}_+^* \times \mathbb{R}_+^*\)
    \begin{align}
        f(\boldsymbol{x}) = \sum_{|\alpha|\leq 1} \frac{D^{\alpha} f(\boldsymbol{a})}{\alpha!} (\boldsymbol{x} - \boldsymbol{a})^{\alpha} + R_{k+1}(\boldsymbol{x}), \label{taylor}
    \end{align}
    with 
    \[R_{k+1}(\boldsymbol{x}) =  \sum_{|\beta|=k+1} (\boldsymbol{x} - \boldsymbol{a})^{\beta} \frac{|\beta|}{\beta!} \int_0^1 (1 - t)^{|\beta|-1} D^{\beta}f(\boldsymbol{a} + t(\boldsymbol{x} - \boldsymbol{a}))\, dt .\]
    \textbf{Bias:} \\
    Computing \ref{taylor} for the ratio function with \(\boldsymbol{x} = (T_1(\boldsymbol{Z}), T_0(\boldsymbol{Z}))\), \(\boldsymbol{a} = (\tau_1,\tau_0)\) and taking the expectation gives us:
    \begin{align*}
        & \mathbb{E}\left[f(T_1(\boldsymbol{Z}), T_0(\boldsymbol{Z}))\right] \\
        &= \mathbb{E}\left[f(\tau_1, \tau_0) + \frac{\partial f(\tau_1, \tau_0)}{\partial T_1(\boldsymbol{Z})}(T_1(\boldsymbol{Z}) - \tau_1) +\frac{\partial f(\tau_1, \tau_0)}{\partial T_0(\boldsymbol{Z})}(T_0(\boldsymbol{Z}) - \tau_0) + R_2(T_1(\boldsymbol{Z}), T_0(\boldsymbol{Z}))\right]\\
        &=\mathbb{E}\left[f(\tau_1, \tau_0)\right] + \frac{\partial f(\tau_1, \tau_0)}{\partial T_1(\boldsymbol{Z})}\mathbb{E}\left[(T_1(\boldsymbol{Z}) - \tau_1)\right] \\
        &+ \frac{\partial f(\tau_1, \tau_0)}{\partial T_0(\boldsymbol{Z})}\mathbb{E}\left[(T_0(\boldsymbol{Z}) - \tau_0)\right] + \mathbb{E}\left[R_2(T_1(\boldsymbol{Z}), T_0(\boldsymbol{Z}))\right]\\
        &= \frac{\tau_1}{\tau_0} + \mathbb{E}\left[R_2(T_1(\boldsymbol{Z}), T_0(\boldsymbol{Z}))\right]
    \end{align*}
    In order to produce \Cref{Th2}, we just need to show that \(\mathbb{E}\left[R_2(T_1(\boldsymbol{Z}), T_0(\boldsymbol{Z}))\right] = O_p\left(\frac{1}{n}\right)\). To do so, we first compute \(R_2(T_1(\boldsymbol{Z}), T_0(\boldsymbol{Z}))\)
    \begin{align*}
        R_2(T_1(\boldsymbol{Z}), T_0(\boldsymbol{Z})) &= \underbrace{2 (T_0(\boldsymbol{Z})-\tau_0)^2\int_0^1 \frac{(1 - t)(\tau_1 + t(T_1(\boldsymbol{Z})-\tau_1))}{(\tau_0 + t(T_0(\boldsymbol{Z})-\tau_0))^3}\, dt}_{R_2^1(T_1(\boldsymbol{Z}), T_0(\boldsymbol{Z}))} \\
        &- \underbrace{2 (T_0(\boldsymbol{Z})-\tau_0)(T_1(\boldsymbol{Z})-\tau_1)\int_0^1 \frac{1 - t}{(\tau_0 + t(T_0(\boldsymbol{Z})-\tau_0))^2}\, dt}_{R_2^2(T_1(\boldsymbol{Z}), T_0(\boldsymbol{Z}))}
    \end{align*}
    Since we assume that \(T_0(\boldsymbol{Z}) \geq m_0 > 0\) and that \(|T_1(\boldsymbol{Z})| \leq M_1\) we have:
    \begin{align*}
        \left|R_2^1(T_1(\boldsymbol{Z}), T_0(\boldsymbol{Z}))\right| &=\left|2 (T_0(\boldsymbol{Z})-\tau_0)^2\int_0^1 \frac{(1 - t)(\tau_1 + t(T_1(\boldsymbol{Z})-\tau_1))}{(\tau_0 + t(T_0(\boldsymbol{Z})-\tau_0))^3}\, dt\right| \\
        &\leq 2 (T_0(\boldsymbol{Z})-\tau_0)^2 \int_0^1 \left|\frac{(1 - t)\max(\tau_1, M_1 )}{\min(m_0,\tau_0) ^3}\right|\, dt\\
        &= (T_0(\boldsymbol{Z})-\tau_0)^2 \underbrace{\frac{M_1}{m_0^3}}_{C_{1}}
    \end{align*}
    Similarly, we have:
    \begin{align*}
        \left|R_2^2(T_1(\boldsymbol{Z}), T_0(\boldsymbol{Z}))\right| &= \left| 2 (T_0(\boldsymbol{Z})-\tau_0)(T_1(\boldsymbol{Z})-\tau_1)\int_0^1 \frac{1 - t}{(\tau_0 + t(T_0(\boldsymbol{Z})-\tau_0))^2}\, dt\right|\\
        &= 2\left| (T_0(\boldsymbol{Z})-\tau_0)(T_1(\boldsymbol{Z})-\tau_1)\right|\left|\int_0^1 \frac{1 - t}{(\tau_0 + t(T_0(\boldsymbol{Z})-\tau_0))^2}\, dt \right|\\
        &\leq \left|(T_0(\boldsymbol{Z})-\tau_0)(T_1(\boldsymbol{Z})-\tau_1)\right|\underbrace{\frac{1}{m_0 ^2}}_{C_{2}}\\
    \end{align*}
    Finally we get that:
    \begin{align*}
    \left|\mathbb{E}\left[R_2(T_1(\boldsymbol{Z}), T_0(\boldsymbol{Z}))\right] \right| &\leq \mathbb{E}\left[|R_2^1(T_1(\boldsymbol{Z}), T_0(\boldsymbol{Z}))| + |R_2^2(T_1(\boldsymbol{Z}), T_0(\boldsymbol{Z}))|\right] \\
    &\leq C_{1} \Var({T_0(\boldsymbol{Z})}) + C_{2}\mathbb{E}\left[\left|(T_0(\boldsymbol{Z})-\tau_0)(T_1(\boldsymbol{Z})-\tau_1)\right|\right]\\
    &\leq C_{1} \Var({T_0(\boldsymbol{Z})}) + C_{2}\sqrt{\Var({T_0(\boldsymbol{Z})})\Var({T_1(\boldsymbol{Z})})}\\
    &\leq C_{1}M_0^2 +C_{2}M_0M_1
    \end{align*}
    Since we have that \(\Var({T_1(\boldsymbol{Z})})= O_p\left(\frac{1}{n}\right)\) and  \(\Var({T_0(\boldsymbol{Z})})= O_p\left(\frac{1}{n}\right)\), we can conclude by:
    \begin{align*}
        \operatorname{Bias}\left(\frac{T_1(\boldsymbol{Z})}{T_0(\boldsymbol{Z})}, \frac{\tau_1}{\tau_0}\right) = \left|\mathbb{E}\left[\frac{T_1(\boldsymbol{Z})}{T_0(\boldsymbol{Z})}\right] - \frac{\tau_1}{\tau_0}\right| \lesssim   \frac{M_1M_0}{nm_0^2}\left(\frac{M_0}{m_0} + 1\right)
    \end{align*}
    \textbf{Variance:} \\
    \text{Let us begin by expanding the variance of the function \( f \):}
\begin{align*}
    \Var{f( T_1(\boldsymbol{Z}),  T_0(\boldsymbol{Z}))} &= \mathbb{E}\left[(f( T_1(\boldsymbol{Z}),  T_0(\boldsymbol{Z})) - \mathbb{E}\left[f( T_1(\boldsymbol{Z}),  T_0(\boldsymbol{Z}))\right])^2\right]
\end{align*}

\text{Next, apply Taylor's expansion around the means \( \tau_1 \) and \( \tau_0 \):}
\begin{align*}
    &= \mathbb{E}\left[(f(\tau_1, \tau_0) + \frac{\partial f(\tau_1, \tau_0)}{\partial T_1(\boldsymbol{Z})}( T_1(\boldsymbol{Z}) - \tau_1) + \frac{\partial f(\tau_1, \tau_0)}{\partial T_0(\boldsymbol{Z})}( T_0(\boldsymbol{Z}) - \tau_0) \right. \\
    &\left. \quad + R_2( T_1(\boldsymbol{Z}),  T_0(\boldsymbol{Z})) - \mathbb{E}\left[f( T_1(\boldsymbol{Z}),  T_0(\boldsymbol{Z}))\right])^2\right]
\end{align*}

\text{Simplify by focusing on the first-order derivatives and residual terms:}
\begin{align*}
    &= \mathbb{E}\left[( \frac{\partial f(\tau_1, \tau_0)}{\partial T_1(\boldsymbol{Z})}( T_1(\boldsymbol{Z}) - \tau_1) + \frac{\partial f(\tau_1, \tau_0)}{\partial T_0(\boldsymbol{Z})}( T_0(\boldsymbol{Z}) - \tau_0) \right. \\
    &\left. \quad + R_2( T_1(\boldsymbol{Z}),  T_0(\boldsymbol{Z})) - \mathbb{E}\left[R_2( T_1(\boldsymbol{Z}),  T_0(\boldsymbol{Z}))\right])^2\right]
\end{align*}

\text{Decompose the variance into linear, cross-term, and residual contributions:}
\begin{align*}
    &= \mathbb{E}\left[( \frac{\partial f(\tau_1, \tau_0)}{\partial T_1(\boldsymbol{Z})}( T_1(\boldsymbol{Z}) - \tau_1) + \frac{\partial f(\tau_1, \tau_0)}{\partial T_0(\boldsymbol{Z})}( T_0(\boldsymbol{Z}) - \tau_0) )^2\right] \\
    &+ 2 \mathbb{E}\left[ \frac{\partial f(\tau_1, \tau_0)}{\partial T_1(\boldsymbol{Z})}( T_1(\boldsymbol{Z}) - \tau_1) \right.\\
    &\left. \quad +\frac{\partial f(\tau_1, \tau_0)}{\partial T_0(\boldsymbol{Z})}( T_0(\boldsymbol{Z}) - \tau_0) \right]\mathbb{E}\left[R_2( T_1(\boldsymbol{Z}),  T_0(\boldsymbol{Z})) - \mathbb{E}\left[R_2( T_1(\boldsymbol{Z}),  T_0(\boldsymbol{Z}))\right]\right]\\
    &+ \Var (R_2( T_1(\boldsymbol{Z}),  T_0(\boldsymbol{Z})))
\end{align*}

\text{Finally, re-express the result using a simplified ratio of variances:}
\begin{align*}
    &= \left(\frac{\tau_1}{\tau_0}\right)^2 \Var{\left(\frac{ T_1(\boldsymbol{Z})}{\tau_1}-\frac{ T_0(\boldsymbol{Z})}{\tau_0}\right)} + \Var (R_2( T_1(\boldsymbol{Z}),  T_0(\boldsymbol{Z}))
\end{align*}

    We now focus on \(\Var (R_2( T_1(\boldsymbol{Z}),  T_0(\boldsymbol{Z})))\):
    \begin{align*}
    \Var (R_2( T_1(\boldsymbol{Z}),  T_0(\boldsymbol{Z}))) &\leq \mathbb{E}\left[(R_2( T_1(\boldsymbol{Z}),  T_0(\boldsymbol{Z})))^2\right]\\
    &\leq 2 \mathbb{E}\left[|R_2^1(T_1(\boldsymbol{Z}), T_0(\boldsymbol{Z}))|^2 \right] + 2 \mathbb{E}\left[|R_2^2(T_1(\boldsymbol{Z}), T_0(\boldsymbol{Z}))|^2\right]
    \end{align*}
    We first focus on the first term:
    \begin{align*}
    \mathbb{E}\left[|R_2^1(T_1(\boldsymbol{Z}), T_0(\boldsymbol{Z}))|^2 \right] &\leq \mathbb{E}\left[\left(C_{02}(T_0(\boldsymbol{Z})-\tau_0)^2\right)^2 \right]\\
    &\leq C_{1}^2 \mathbb{E}\left[(T_0(\boldsymbol{Z})-\tau_0)^4\right]\\
    &\leq C_{1}^2 M_0^2\mathbb{E}\left[(T_0(\boldsymbol{Z})-\tau_0)^2\right] \quad  T_0(\boldsymbol{Z}) \leq M_0\\
    &\leq C_{1}^2 M_0^2\Var({T_0(\boldsymbol{Z})})
    \end{align*}
    For the second term we have:
    \begin{align*}
        \mathbb{E}\left[|R_2^2(T_1(\boldsymbol{Z}), T_0(\boldsymbol{Z}))|^2\right] &= C_{2}^2\mathbb{E}\left[(T_0(\boldsymbol{Z})-\tau_0)^2(T_1(\boldsymbol{Z})-\tau_1)^2\right]\\
        &\leq C_{2}^2 \sqrt{\mathbb{E}\left[(T_0(\boldsymbol{Z})-\tau_0)^4\right]\mathbb{E}\left[(T_1(\boldsymbol{Z})-\tau_1)^4\right]}\\
        T_1(\boldsymbol{Z}),T_0(\boldsymbol{Z}) \text{ bounded} \quad &\leq C_{2}^2 M_0M_1\sqrt{\mathbb{E}\left[(T_0(\boldsymbol{Z})-\tau_0)^2\right]\mathbb{E}\left[(T_1(\boldsymbol{Z})-\tau_1)^2\right]}\\
        &\leq C_{2}^2 M_0M_1 \sqrt{\Var({T_0(\boldsymbol{Z})})\Var({T_1(\boldsymbol{Z})})}
    \end{align*}
    Hence we get that:
    \begin{align*}
        \Var (R_2( T_1(\boldsymbol{Z}),  T_0(\boldsymbol{Z}))) \lesssim\frac{2M_0M_1}{nm_0^4}\left(\frac{M_0M_1}{m_0^2}+1\right)
    \end{align*}
\end{proof}

%This theorem is an adaptation of Theorem 4.1 of \cite{derumigny2019construction}.

\subsection{Proofs of \Cref{risk_ratio_RCTs}}

\subsubsection{Risk Ratio Neyman estimator}
\label{proof:N}
\begin{proof}[Proof of \Cref{prop:N}]
\ \\
\textbf{Asymptotic Bias and Variance:} we proceed with M-estimations to prove asymptotic bias and variance of the Ratio Neyman estimator, we first define the following:
\begin{align}
    \hat{\boldsymbol{\theta}}_n & =\left(\begin{array}{c}
    \frac{1}{n_0}\sum_{T_i=0}Y_i \\
    \frac{1}{n_1}\sum_{T_i=1}Y_i \\
    \hat{\tau}_{R-N,n}
    \end{array}\right) \quad \textrm{and} \quad 
    \psi(T, Y, \boldsymbol{\theta}) =\left(\begin{array}{c}
    \psi_0(\boldsymbol{\theta})  \\
    \psi_1(\boldsymbol{\theta})\\
    \psi_2(\boldsymbol{\theta})
    \end{array}\right)
    =:\left(\begin{array}{c}
    (1-T) \left(Y - \theta_0\right)  \\
    T\left(Y - \theta_1\right)\\
    \theta_1 - \theta_2\theta_0
    \end{array}\right),
    \end{align}
where \(\boldsymbol{\theta} = (\theta_0, \theta_1, \theta_2)\).

Next, we verify that for \(\hat{\boldsymbol{\theta}}_n = (\frac{1}{n_0}\sum_{T_i=0}Y_i, \frac{1}{n_1}\sum_{T_i=1}Y_i, \hat{\tau}_{R-N,n})\), we have:
\[\sum_{i=1}^n \psi(T_i, Y_i, \hat{\boldsymbol{\theta}}_n) = 0.\]

We begin by demonstrating this for \(\psi_1\):
\begin{align*}
    \sum_{i=1}^n \psi_1(T_i, Y_i, \hat{\boldsymbol{\theta}}_n) &=\sum_{i=1}^n T_i\left(Y_i -\frac{1}{n_1}\sum_{T_j=1}Y_j \right) \\
    &= \sum_{i=1}^n T_i\left(Y_i -\frac{1}{n_1}\sum_{j=1}^n T_jY_j \right)\\
    &= \sum_{i=1}^n T_iY_i - \frac{1}{n_1}\sum_{i=1}^n T_i \sum_{j=1}^n T_jY_j\\
    &= \sum_{i=1}^n T_iY_i - \sum_{j=1}^n T_jY_j\\
    &= 0.
\end{align*}

Similarly, we can show:
\[\sum_{i=1}^n \psi_0(T_i, Y_i, \hat{\boldsymbol{\theta}}_n)  = 0.\]

Moreover, by construction:
\[\sum_{i=1}^n \psi_2(T_i, Y_i, \hat{\boldsymbol{\theta}}_n)  = 0.\]

Thus, we have established that \( \hat{\boldsymbol{\theta}}_n\) is an M-estimator of type \(\psi\) \citep[see][]{Stefanski2002Mestimation}. Given that we are in a \hyperref[a:bernoulli_trial]{Bernoulli Trial}, we now demonstrate that \(\mathbb{E}\left[\psi(T, Y, \theta_\infty)\right]=0\) where \(\theta_\infty= (\mathbb{E}[Y^{(0)}], \mathbb{E}[Y^{(1)}], \tau_{RR})\). Therefore, we have:
\begin{align*}
    \mathbb{E}\left[\psi_1(\theta_\infty)\right] &=\mathbb{E}\left[T\left(Y - \mathbb{E}[Y^{(1)}]\right)\right] \\
    &= \mathbb{E}\left[T\left(Y^{(1)} - \mathbb{E}[Y^{(1)}]\right)\right] && \text{(by \hyperref[a:SUTVA]{SUTVA})}\\
    &=\mathbb{E}\left[T\right]\mathbb{E}\left[Y^{(1)} - \mathbb{E}[Y^{(1)}]\right] && \text{(by \hyperref[a:ignorability]{ignorability})}\\
    &= 0.
\end{align*}

Similarly, we can show:
\[\mathbb{E}\left[\psi_0(\theta_\infty)\right] = 0.\]

Furthermore, we have:
\[\mathbb{E}\left[\psi_2(\theta_\infty)\right] = \mathbb{E}[Y^{(1)}] - \tau_{RR} \mathbb{E}[Y^{(0)}] = 0.\]

At this point, we note that \(\theta_\infty\) is the only value of \(\boldsymbol{\theta}\) such that \(\mathbb{E}\left[\psi(T, Y, \boldsymbol{\theta})\right]=0\). We proceed by defining:
\begin{align*}
  A\left(\theta_\infty \right) = \mathbb{E}\left[\frac{\partial \psi}{\partial \theta}|_{\theta=\theta_\infty}\right] \quad \textrm{and} \quad B(\theta_\infty) =  \mathbb{E}\left[\psi(Z, \theta_\infty) \psi(Z, \theta_\infty)^T\right].
\end{align*}

Next, we check the conditions of Theorem~7.2 in \citet{Stefanski2002Mestimation}. First, we compute \(A\left(\theta_\infty \right) \) and \(B\left(\theta_\infty \right)\). Since:
\begin{align*}
 \frac{\partial \psi}{\partial \theta} (Z, \theta)=\left(\begin{array}{ccc}
-(1-T) & 0 & 0 \\
0 & -T & 0 \\
-\theta_2 & 1 & -\theta_0
\end{array}\right),  
\end{align*}
we obtain:
\begin{align*} 
A\left(\boldsymbol{\theta_\infty} \right)= \left(\begin{array}{ccc}
-(1-e) & 0 & 0 \\
0 & -e & 0 \\
-\tau_{RR} & 1 & -\mathbb{E}[Y^{(0)}]
\end{array}\right),
\end{align*}
which leads to:
\begin{align*}
A^{-1}\left(\boldsymbol{\theta_\infty} \right)= \left(\begin{array}{ccc}
\frac{1}{e-1} & 0 & 0 \\
0 & -\frac{1}{e} & 0 \\
\tau_{RR}\frac{1}{\mathbb{E}[Y^{(0)}](1-e)} & -\frac{1}{e\mathbb{E}[Y^{(0)}]} & -\frac{1}{\mathbb{E}[Y^{(0)}]}
\end{array}\right).
\end{align*}

Regarding \(B\left(\theta_\infty \right)\), elementary calculations show that:
\begin{align*}
        & \psi(Z, \theta_\infty) \psi(Z, \theta_\infty)^T \\
        & = \left(\begin{array}{ccccc}
    \left((1-T)(Y-\mathbb{E}[Y^{(0)}])\right)^2 & (1-T)(Y-\mathbb{E}[Y^{(0)}])T(Y-\mathbb{E}[Y^{(1)}]) & 0 \\
    (1-T)(Y-\mathbb{E}[Y^{(0)}])T(Y-\mathbb{E}[Y^{(1)}]) & 	\left(T(Y-\mathbb{E}[Y^{(1)}])\right)^2 & 0 \\
    0 & 0 & 0
    \end{array}\right),
\end{align*}
which leads to:
\begin{align*}
        B(\theta_\infty)
        & = \left(\begin{array}{ccccc}
    (1-e)\Var\left[Y^{(0)}\right] & 0 & 0 \\
    0 & e\Var\left[Y^{(1)}\right] & 0 \\
    0 & 0 & 0
    \end{array}\right).
\end{align*}

Based on the previous calculations, we have:
\begin{itemize}
    \item  \(\psi(z,\boldsymbol{\theta})\) and its first two partial derivatives with respect to \(\boldsymbol{\theta}\) exist for all \(z\) and for all \(\boldsymbol{\theta}\) in the neighborhood of \(\boldsymbol{\theta_\infty}\).
    \item For each \(\boldsymbol{\theta}\) in the neighborhood of \(\boldsymbol{\theta_\infty}\), we have for all \(i,j,k \in \{0, 2\}\):
    \[\left| \frac{\partial^2}{\partial \theta_i \partial \theta_j} \psi_k(z, \boldsymbol{\theta}) \right| \leq 1\]
    and 1 is integrable.
    \item \(A(\theta_\infty)\) exists and is nonsingular.
    \item \(B(\theta_\infty)\) exists and is finite.
\end{itemize}

Since we have: 
\[\sum_{i=1}^n \psi(T_i, Y_i, \hat{\boldsymbol{\theta}}_n) = 0 \quad \text{and} \quad \hat{\boldsymbol{\theta}}_n \stackrel{p}{\rightarrow} \theta_\infty .\]

Then, the conditions of Theorem~7.2 in \citet{Stefanski2002Mestimation} are satisfied, we have:
\[\sqrt{n}\left(  \hat{\boldsymbol{\theta}}_n - \theta_\infty \right) \stackrel{d}{\rightarrow} \mathcal{N}\left(0, A(\theta_\infty)^{-1}B(\theta_\infty)(A(\theta_\infty)^{-1})^{\top} \right),\]
where:
\begin{align*}
    A(\theta_\infty)^{-1}B(\theta_\infty)(A(\theta_\infty)^{-1})^{\top} \nonumber  = \left[\begin{matrix}\frac{\Var\left[Y^{(0)}\right]}{\left(1 - e\right)} & 0 & -\frac{\tau \Var\left[Y^{(0)}\right]}{\tau_{0}\left(1 - e\right)}\\0 & \frac{\Var\left[Y^{(1)}\right]}{e} & \frac{\Var\left[Y^{(1)}\right]}{e \tau_{0}}\\- \frac{\tau \Var\left[Y^{(0)}\right] }{\tau_{0}\left(1 - e\right)} & \frac{\Var\left[Y^{(1)}\right]}{e \tau_{0}} &V_{R-N}\end{matrix}\right],
\end{align*}
with:
\begin{align*}
V_{R-N} =  \tau_{RR}^2\left(\frac{\Var\left(Y^{(1)}\right)}{e\mathbb{E}[Y^{(1)}]^2} + \frac{\Var\left(Y^{(0)}\right)}{(1-e)\mathbb{E}[Y^{(0)}]^2}\right).    
\end{align*}

In particular, we obtain:
\[\sqrt{n}\left(\hat{\tau}_{\textrm{\tiny RR,N,n}} - \tau_{\textrm{\tiny RR}} \right) \stackrel{d}{\rightarrow} \mathcal{N}\left(0, V_{\textrm{\tiny RR,N}} \right).\]

Finally, note that:
\begin{align*}
    V_{R-N} &=\tau_{RR}^2\left(\frac{\Var\left(Y^{(1)}\right)}{e\mathbb{E}[Y^{(1)}]^2} + \frac{\Var\left(Y^{(0)}\right)}{(1-e)\mathbb{E}[Y^{(0)}]^2}\right) \\
    &= \tau_{RR}^2\left(\frac{\mathbb{E}[(Y^{(1)})^2]- \mathbb{E}[Y^{(1)}]^2}{e\mathbb{E}[Y^{(1)}]^2} + \frac{\mathbb{E}[(Y^{(0)})^2]- \mathbb{E}[Y^{(0)}]^2}{(1-e)\mathbb{E}[Y^{(0)}]^2}\right) \\
    &= V_{R-HT} - \frac{\tau_{RR}^2}{e(1-e)}.
\end{align*}
As a consequence an estimator \(\hat V_{R-N}\) can be derived :

\begin{align}\label{V_R-N}
    \hat V_{R-N} = \hat{\tau}_{\textrm{\tiny RR,N,n}}^2 \left(\frac{\frac{1}{n}\sum_{T_i=1}\left(Y_i - \frac{1}{n}\sum_{T_i=1}Y_i\right)^2}{\hat e\left(\frac{1}{n}\sum_{T_i=1}Y_i\right)^2} + \frac{\frac{1}{n}\sum_{T_i=0}\left(Y_i - \frac{1}{n}\sum_{T_i=0}Y_i\right)^2}{(1-\hat e)\left(\frac{1}{n}\sum_{T_i=0}Y_i\right)^2}\right)
\end{align}

\textbf{Optimal choice of \(e\):}
the optimal value of \(e_{opt}\) is the one that minimizes the variance of the Ratio Neyman estimator. Therefore, we need to solve:
\[\inf_{e \in (0,1)} \tau_{\textrm{\tiny RR}}^2 \left(\frac{\Var\left(Y^{(1)}\right)}{e\mathbb{E}\left[Y^{(1)}\right]^2}+ \frac{\Var\left(Y^{(0)}\right)}{(1-e)\mathbb{E}\left[Y^{(0)}\right]^2}\right)\]
Noting that the variance we want to minimize is convex in \(e\), we can derive the variance and set it to \(0\) to find \(e_{opt}\). We have:
\[\frac{C_1}{e_{opt}^2} = \frac{C_0}{(1-e_{opt})^2}\]
where \(C_1 = \frac{\Var\left(Y^{(1)}\right)}{\mathbb{E}\left[Y^{(1)}\right]^2}\) and \(C_0 = \frac{\Var\left(Y^{(0)}\right)}{\mathbb{E}\left[Y^{(0)}\right]^2}\). 
\begin{itemize}
    \item If \(\frac{\Var\left(Y^{(1)}\right)}{\mathbb{E}[Y^{(1)}]^2} = \frac{\Var\left(Y^{(0)}\right)}{\mathbb{E}[Y^{(0)}]^2}\):
    \[e_{opt} = 0.5\]
    \item otherwise:
    \[e_{opt} = \frac{C_1-\sqrt{C_1C_0}}{C_1-C_0}\in (0,1)\]
\end{itemize}
\end{proof}

\subsubsection{Risk Ratio Horvitz-Thomson estimator}
\label{proof:HT}

\begin{definition}[\textbf{Risk Ratio Horvitz-Thomson estimator}]
Grant \Cref{a:bernoulli_trial} and \Cref{a:Outcome_positivity}.
    The Risk Ratio Horvitz-Thomson estimator denoted \(\hat{\tau}_{{\textrm{\tiny RR,HT,n}}}\) is defined as,
    \begin{align}
     \hat{\tau}_{\textrm{\tiny RR,HT,n}} = \frac{\sum_{i=1}^n \frac{T_iY_i}{e}}{\sum_{i=1}^n \frac{(1-T_i)Y_i}{1-e}}
     \end{align}
     if \(\sum_{i=1}^n T_i < n\) and \(0\) otherwise.
    
\end{definition}

Within the context of a Bernoulli trial, \Cref{prop:HT} proves that the Risk Ratio Horvitz-Thompson estimator is asymptotically unbiased and normally distributed.
\begin{proposition}[\textbf{Asymptotic normality of \(\hat{\tau}_{\textrm{\tiny RR,HT,n}}\)}]\label{prop:HT}
Under \Cref{a:bernoulli_trial} and \Cref{a:Outcome_positivity}, the Risk Ratio Horvitz-Thompson estimator is asymptotically unbiased and satisfies
\begin{align}\label{Var_HT}
\sqrt{n}\left(\hat{\tau}_{\textrm{\tiny RR,HT,n}} - \tau_{\textrm{\tiny RR}} \right) \stackrel{d}{\rightarrow} \mathcal{N}\left(0, V_{\textrm{\tiny RR,HT}} \right)
\end{align} 
where   \(V_{\textrm{\tiny RR,HT}} = \tau_{\textrm{\tiny RR}}^2 \left(\frac{\mathbb{E}\left[\left(Y^{(1)}\right)^2\right]}{e\mathbb{E}\left[Y^{(1)}\right]^2}+ \frac{\mathbb{E}\left[\left(Y^{(0)}\right)^2\right]}{(1-e)\mathbb{E}\left[Y^{(0)}\right]^2}\right)\).

If we assume that for all \(i\), \( M \geq Y_i \geq m > 0\) and \(0<\sum_{i=1}^n T_i < n\), we also have:
\[\left|Bias(\hat{\tau}_{\textrm{\tiny RR, HT, n}}\right| \leq \frac{2M^3(1-e)^3}{nm^3e^3}\] 
\[ \left|\Var(\hat{\tau}_{\textrm{\tiny RR, HT, n}})\right| \leq \frac{4M^4(1-e)^6}{nm^6e^4}\]
\end{proposition}

\begin{proof}[Proof of \Cref{prop:HT}]

\ 

\textbf{Asymptotic Bias and Variance.} Let \(Z_i := (T_i, Y_i)\) and define \(g_0(Z_i) = \frac{(1-T_i)Y_i}{1-e}\) and \(g_1(Z_i) = \frac{T_iY_i}{e}\). First, we evaluate the expectation of \(g_1(Z_i)\):

\begin{align*}
    \mathbb{E} \left[g_1(Z_i)\right] &= \mathbb{E} \left[\frac{T_iY_i}{e}\right] && \text{(by \hyperref[a:i.i.d.]{i.i.d})}\\
    &= \mathbb{E} \left[\frac{T_iY_i^{(1)}}{e}\right] && \text{(by \hyperref[a:SUTVA]{SUTVA})} \\
    &= \mathbb{E} \left[\frac{T_i}{e}\right] \mathbb{E} \left[Y_i^{(1)}\right] && \text{(by \hyperref[a:ignorability]{ignorability})}\\
    &= \mathbb{E} \left[Y_i^{(1)}\right] && \text{(by \hyperref[a:Trial_positivity]{Trial positivity})}
\end{align*}

Similarly, we can find the expectation of \(g_0(Z_i)\):

\[
\mathbb{E} \left[g_0(Z_i)\right] = \mathbb{E}\left[Y^{(0)}\right] > 0.
\]

Thus, according to \Cref{Th1}, we have $
\sqrt{n}\left(\hat{\tau}_{\textrm{\tiny RR-HT, n}} - \tau_{\textrm{\tiny RR}} \right) \stackrel{d}{\rightarrow} \mathcal{N}\left(0, V_{\textrm{\tiny RR-HT}} \right)$, with
\begin{align*}
V_{\textrm{\tiny RR-HT}} & = \left(\frac{\tau_1}{\tau_0}\right)^2 \Var\left(\frac{g_1(Z)}{\tau_1} - \frac{g_0(Z)}{\tau_0}\right) \\
 & = \tau_{\textrm{\tiny RR}}^2 \Var\left(\frac{TY}{e \mathbb{E}\left[Y^{(1)}\right]} - \frac{(1-T)Y}{(1-e) \mathbb{E}\left[Y^{(0)}\right]}\right).
\end{align*}

Next, we evaluate the variance terms separately:

\begin{align*}
    \Var\left(\frac{TY}{e \mathbb{E}\left[Y^{(1)}\right]}\right) &= \frac{1}{\mathbb{E}\left[Y^{(1)}\right]^2 e^2} \Var\left(TY\right) \\
    &= \frac{1}{\mathbb{E}\left[Y^{(1)}\right]^2 e^2} \left(\mathbb{E}\left[\left(TY\right)^2\right] - \mathbb{E}\left[TY\right]^2\right) \\
    &= \frac{1}{\mathbb{E}\left[Y^{(1)}\right]^2 e^2} \left(\mathbb{E}\left[T\left(Y\right)^2\right] - \mathbb{E}\left[TY\right]^2\right) && \text{(\(T\) is binary)} \\
    &= \frac{1}{\mathbb{E}\left[Y^{(1)}\right]^2 e^2} \left(\mathbb{E}\left[T\left(Y^{(1)}\right)^2\right] - \mathbb{E}\left[TY^{(1)}\right]^2\right) && \text{(by \hyperref[a:SUTVA]{SUTVA})} \\
    &= \frac{1}{\mathbb{E}\left[Y^{(1)}\right]^2 e^2} \left(e \mathbb{E}\left[\left(Y^{(1)}\right)^2\right] - e^2 \mathbb{E}\left[Y^{(1)}\right]^2\right) && \text{(by \hyperref[a:ignorability]{ignorability})} \\
    &= \frac{\mathbb{E}\left[\left(Y^{(1)}\right)^2\right]}{e \mathbb{E}\left[Y^{(1)}\right]^2} - 1.
\end{align*}

Similarly, we find the variance of the second term:

\[
\Var\left(\frac{(1-T)Y}{(1-e) \mathbb{E}\left[Y^{(0)}\right]}\right) = \frac{\mathbb{E}\left[\left(Y^{(0)}\right)^2\right]}{(1-e) \mathbb{E}\left[Y^{(0)}\right]^2} - 1.
\]

Finally, we compute the covariance between the two terms:

\begin{align*}
    \Cov\left(\frac{TY}{e \mathbb{E}\left[Y^{(1)}\right]}, \frac{(1-T)Y}{(1-e) \mathbb{E}\left[Y^{(0)}\right]}\right) &= \frac{\Cov(TY, (1-T)Y)}{e \mathbb{E}\left[Y^{(1)}\right] (1-e) \mathbb{E}\left[Y^{(0)}\right]}  \\
    &= \frac{\left(\mathbb{E}\left[T(1-T)Y^2\right] - \mathbb{E}\left[TY\right] \mathbb{E}\left[(1-T)Y\right]\right)}{e \mathbb{E}\left[Y^{(1)}\right] (1-e) \mathbb{E}\left[Y^{(0)}\right]}  \\
    &= \frac{- \mathbb{E}\left[TY\right] \mathbb{E}\left[(1-T)Y\right]}{e \mathbb{E}\left[Y^{(1)}\right] (1-e) \mathbb{E}\left[Y^{(0)}\right]} \\
    &= -1.
\end{align*}

Using Bienayme's identity, we finally obtain:

\[
V_{\textrm{\tiny RR-HT}} = \tau_{\textrm{\tiny RR}}^2 \left(\frac{\mathbb{E}\left[\left(Y^{(1)}\right)^2\right]}{e \mathbb{E}\left[Y^{(1)}\right]^2} + \frac{\mathbb{E}\left[\left(Y^{(0)}\right)^2\right]}{(1-e) \mathbb{E}\left[Y^{(0)}\right]^2}\right).
\]

As a consequence an estimator \(\hat V_{RR-HT}\) can be derived:

\begin{align}\label{V_RR-HT}
    \hat V_{RR-HT} = \hat{\tau}_{\textrm{\tiny RR,HT,n}}^2 \left(\frac{\frac{1}{n}\sum_{T_i=1}Y_i^2 }{\hat e\left(\frac{1}{n}\sum_{T_i=1}Y_i\right)^2} + \frac{\frac{1}{n}\sum_{T_i=0}Y_i^2}{(1-\hat e)\left(\frac{1}{n}\sum_{T_i=0}Y_i\right)^2}\right)
\end{align}

\textbf{Finite sample Bias and Variance.} Let \(T_1(\boldsymbol{Z}) = \frac{1}{n}\sum_{i=1}^n \frac{T_i Y_i}{e}\) and \(T_0(\boldsymbol{Z}) = \frac{1}{n}\sum_{i=1}^n \frac{(1-T_i) Y_i}{1-e}\) where \(\boldsymbol{Z} = (Z_1, \hdots, Z_n)\). First, consider the variance of \(T_1(\boldsymbol{Z})\):

\begin{align*}
    \Var(T_1(\boldsymbol{Z})) &= \frac{1}{n e^2} \Var\left(T_i Y_i\right) && && \text{(by \hyperref[a:i.i.d.]{i.i.d})}\\ 
    &= \frac{1}{n e^2} \left(\mathbb{E}\left[(T_i Y_i)^2\right] - \mathbb{E}\left[T_i Y_i\right]^2\right) \\
    &= \frac{\mathbb{E}\left[\left(Y^{(1)}\right)^2\right] - e \mathbb{E}\left[Y^{(1)}\right]^2}{n e}.
\end{align*}
Thus $\Var(T_1(\boldsymbol{Z})) = O_p\left(1/n\right)$ and similarly  $\Var(T_0(\boldsymbol{Z})) = O_p\left(1/n\right)$. Next, we show that \(T_0(\boldsymbol{Z})\) is bounded:

\begin{align*}
    T_0(\boldsymbol{Z}) &= \frac{1}{n} \sum_{i=1}^n \frac{(1-T_i) Y_i}{1-e} \\
    &= \frac{1}{n(1-e)} \sum_{i=1}^n (1-T_i) Y_i \\
    &\geq \frac{m}{(1-e)} \sum_{i=1}^n (1-T_i) && \text{(since \(Y_i \geq m > 0\))} \\
    &\geq \frac{m}{(1-e)} && \text{(as \(\sum_{i=1}^n T_i < n\))}.
\end{align*}

Similarly, we also have the upper bound
\begin{align*}
    T_0(\boldsymbol{Z}) &= \frac{1}{n} \sum_{i=1}^n \frac{(1-T_i) Y_i}{1-e} \\
    &= \frac{1}{n e} \sum_{i=1}^n (1-T_i) Y_i  \\
    &\leq \frac{1}{n e} \sum_{i=1}^n Y_i && \text{(since \(T\) is binary)} \\
    &\leq \frac{M}{e} && \text{(since \(Y_i \leq M\))}.
\end{align*}

Similarly,  we have $T_1(\boldsymbol{Z})\leq \frac{M}{e}$. Therefore, we have shown that \(T_1(\boldsymbol{Z})\) and \(T_0(\boldsymbol{Z})\) are unbiased estimators of \(\mathbb{E}\left[Y^{(1)}\right]\) and \(\mathbb{E}\left[Y^{(0)}\right] > 0\), respectively. We also established that \(M/e \geq T_0(\boldsymbol{Z}) \geq m/(1-e) > 0\) and \(|T_1(\boldsymbol{Z})| \leq M/e\).
Furthermore, we pointed out that \(\Var(T_1(\boldsymbol{Z})) = O_p\left(\frac{1}{n}\right)\) and \(\Var(T_0(\boldsymbol{Z})) = O_p\left(\frac{1}{n}\right)\). Applying \Cref{Th2}, we obtain:
\[\left|\mathbb{E}\left[\hat{\tau}_{\textrm{\tiny RR, HT, n}}\right] - \tau_{\textrm{\tiny RR}}\right| \leq \frac{M^2(1-e)^2}{ne^2m^2}\left(\frac{M(1-e)}{me} + 1\right) \leq \frac{2M^3(1-e)^3}{nm^3e^3},\]
and
\begin{align*}
\left|\Var(\hat{\tau}_{\textrm{\tiny RR, HT, n}}) - V_{\textrm{\tiny RR, HT}}\right| \leq \frac{2M^2(1-e)^4}{nm^4e^2}\left(\frac{M^2(1-e)^2}{m^2e^2}+1\right)\leq \frac{4M^4(1-e)^6}{nm^6e^4}.
\end{align*}

\textbf{Optimal choice of \(e\)}
The optimal value of \(e_{opt}\) is the one that minimizes the variance of the Ratio Horvitz-Thomson estimator. Therefore, we need to solve:
\[\inf_{e \in (0,1)} \tau_{\textrm{\tiny RR}}^2 \left(\frac{\mathbb{E}\left[\left(Y^{(1)}\right)^2\right]}{e\mathbb{E}\left[Y^{(1)}\right]^2}+ \frac{\mathbb{E}\left[\left(Y^{(0)}\right)^2\right]}{(1-e)\mathbb{E}\left[Y^{(0)}\right]^2}\right)\]
Noting that the variance we want to minimize is convex in \(e\), we can derive the variance and set it to \(0\) to find \(e_{opt}\). We have:
\[\frac{C_1}{e_{opt}^2} = \frac{C_0}{(1-e_{opt})^2}\]
where \(C_1 = \frac{\mathbb{E}\left[\left(Y^{(1)}\right)^2\right]}{\mathbb{E}\left[Y^{(1)}\right]^2}\) and \(C_0 = \frac{\mathbb{E}\left[\left(Y^{(0)}\right)^2\right]}{\mathbb{E}\left[Y^{(0)}\right]^2}\). 
\begin{itemize}
    \item If \(\frac{\Var\left(Y^{(1)}\right)}{\mathbb{E}[Y^{(1)}]^2} = \frac{\Var\left(Y^{(0)}\right)}{\mathbb{E}[Y^{(0)}]^2}\):
    \[e_{opt} = 0.5\]
    \item otherwise:
    \[e_{opt} = \frac{\mathbb{E}\left[\left(Y^{(1)}\right)^2\right]\mathbb{E}\left[Y^{(0)}\right]^2-\sqrt{\mathbb{E}\left[\left(Y^{(1)}\right)^2\right]\mathbb{E}\left[\left(Y^{(0)}\right)^2\right]}\mathbb{E}\left[Y^{(1)}\right]\mathbb{E}\left[Y^{(0)}\right]}{\mathbb{E}\left[\left(Y^{(1)}\right)^2\right]\mathbb{E}\left[Y^{(0)}\right]^2-\mathbb{E}\left[\left(Y^{(0)}\right)^2\right]\mathbb{E}\left[Y^{(1)}\right]^2} \in (0,1)\]
\end{itemize}
\end{proof}

\subsubsection{Link with existing asymptotic confidence intervals}
\label{app_link_existing_CI}

According to \Cref{prop:N}, a $(1-\alpha)$ asymptotic confidence interval for $\tau_{\textrm{\tiny RR}}$ is given by 
\begin{align}
\left[ \hat{\tau}_{\textrm{\tiny RR,N,n}} \pm \frac{\sqrt{\widehat{V_{\textrm{\tiny RR,N}}}}z_{1-\alpha/2}}{n} \right]
\end{align}
with $\widehat{V_{\textrm{\tiny RR,N}}}$ an estimator of 
\begin{align*}
V_{\textrm{\tiny RR,N}} &=\tau_{\textrm{\tiny RR}}^2\left(\frac{\Var\left(Y^{(1)}\right)}{e\mathbb{E}[Y^{(1)}]^2} + \frac{\Var\left(Y^{(0)}\right)}{(1-e)\mathbb{E}[Y^{(0)}]^2}\right).
 \end{align*}
Now, assume that $Y^{(0)}, Y^{(1)} \in \{0,1\}$ with associated probabilities $\mathbb{P}[Y^{(0)} = 1]= p_0$ and $\mathbb{P}[Y^{(1)} = 1]= p_1$. In this setting, the variance $V_{\textrm{\tiny RR,N}}$ takes the form 
\begin{align*}
\frac{V_{\textrm{\tiny RR,N}}}{n} &=\frac{\tau_{\textrm{\tiny RR}}^2}{n}\left(\frac{\Var\left(Y^{(1)}\right)}{e\mathbb{E}[Y^{(1)}]^2} + \frac{\Var\left(Y^{(0)}\right)}{(1-e)\mathbb{E}[Y^{(0)}]^2}\right)\\
& = \tau_{\textrm{\tiny RR}}^2\left(\frac{p_1 (1-p_1)}{N_1 p_1^2} + \frac{p_0 (1-p_0)}{N_0 p_0^2}\right)\\
& = \tau_{\textrm{\tiny RR}}^2\left(\frac{1-p_1}{N_1 p_1} + \frac{1-p_0}{N_0 p_0}\right)\\
& = \tau_{\textrm{\tiny RR}}^2\left(\frac{1}{N_1 p_1} -\frac{1}{N_1} + \frac{1}{N_0 p_0} -\frac{1}{N_0}\right)\\
& = \tau_{\textrm{\tiny RR}}^2\left(\frac{1}{N_1 p_1} -\frac{1}{N_1} + \frac{1}{N_0 p_0} -\frac{1}{N_0}\right).
 \end{align*}
An estimation of such a quantity can be constructed by replacing $p_1$ (resp. $p_0$) by $(1/N_1)\sum_{i=1}^n T_i Y_i$ (resp. $(1/N_0)\sum_{i=1}^n (1 - T_i) Y_i$), which leads to 
\begin{align}
\frac{\widehat{V_{\textrm{\tiny RR,N}}}}{n} =  \hat{\tau}_{\textrm{\tiny RR}}^2\left(\frac{1}{\sum_{i=1}^n T_i Y_i} -\frac{1}{N_1} + \frac{1}{\sum_{i=1}^n (1 - T_i) Y_i} -\frac{1}{N_0}\right).   
\end{align}
Thus, a $(1-\alpha)$ asymptotic confidence interval for $\tau_{\textrm{\tiny RR}}$ is given by 
\begin{align}
& \left[ \hat{\tau}_{\textrm{\tiny RR,N,n}} \pm z_{1-\alpha/2} \hat{\tau}_{\textrm{\tiny RR,N,n}} \sqrt{\left(\frac{1}{\sum_{i=1}^n T_i Y_i} -\frac{1}{N_1} + \frac{1}{\sum_{i=1}^n (1 - T_i) Y_i} -\frac{1}{N_0}\right)} \right] \\
& = \left[ \hat{\tau}_{\textrm{\tiny RR,N,n}} \left( 1  \pm z_{1-\alpha/2}  \sqrt{\left(\frac{1}{\sum_{i=1}^n T_i Y_i} -\frac{1}{N_1} + \frac{1}{\sum_{i=1}^n (1 - T_i) Y_i} -\frac{1}{N_0}\right)}\right)  \right].
\end{align}
Finally, since $e^{x}$ is equivalent to $1+x$ near $x=0$, the above interval is equivalent to  that given by \eqref{eq_CI_existing_literature}, which concludes the proof. 

\subsubsection{Delta method with $\log$ function}
\label{app_delta_method_log}

According to \Cref{prop:N}, we know that 
\begin{align}
 \sqrt{n}\left(\hat{\tau}_{\textrm{\tiny RR,N,n}} - \tau_{\textrm{\tiny RR}} \right) \stackrel{d}{\rightarrow} \mathcal{N}\left(0, V_{\textrm{\tiny RR,N}} \right),
 \end{align}
where 
\begin{align*}
V_{\textrm{\tiny RR,N}} &=\tau_{\textrm{\tiny RR}}^2\left(\frac{\Var\left(Y^{(1)}\right)}{e\mathbb{E}[Y^{(1)}]^2} + \frac{\Var\left(Y^{(0)}\right)}{(1-e)\mathbb{E}[Y^{(0)}]^2}\right).
%\\ &= V_{\textrm{\tiny RR,HT}} - \frac{\tau_{\textrm{\tiny RR}}^2}{e(1-e)}.
 \end{align*}
Using the Delta method, with the function $\theta \mapsto \log(\theta)$, we obtain
\begin{align}
\sqrt{n}\left(\log(\hat{\tau}_{\textrm{\tiny RR,N,n}}) - \log(\tau_{\textrm{\tiny RR}}) \right) \stackrel{d}{\rightarrow} \mathcal{N}\left(0, (1/\tau_{\textrm{\tiny RR}})^2 V_{\textrm{\tiny RR,N}} \right).
\end{align} 
Thus, a $(1-\alpha)$ asymptotic confidence interval for $\log(\tau_{\textrm{\tiny RR}})$ is given by 
\begin{align}
    \left[ \log(\hat{\tau}_{\textrm{\tiny RR,N,n}}) \pm z_{1-\alpha/2} \sqrt{\frac{V_{\textrm{\tiny RR,N}}}{n \tau_{\textrm{\tiny RR}}^2 }}\right].
\end{align}
Letting $V_{\textrm{\tiny log RR,N}} = V_{\textrm{\tiny RR,N}}/\tau_{\textrm{\tiny RR}}^2 $, a $(1-\alpha)$ asymptotic confidence interval for $\tau_{\textrm{\tiny RR}}$ is 
\begin{align}
    \left[ \hat{\tau}_{\textrm{\tiny RR,N,n}} 
    \exp\left( \pm z_{1-\alpha/2} \sqrt{\frac{V_{\textrm{\tiny log RR,N}}}{n  }} \right) \right].\label{eq_proof_log_CI1}
\end{align}
Now, note that, if $Y^{(0)}, Y^{(1)} \in \{0,1\}$ with $\mathbb{P}[Y^{(t)} = 1] = p_t$, we have
\begin{align}
V_{\textrm{\tiny log RR,N}} & = \frac{\Var\left(Y^{(1)}\right)}{e\mathbb{E}[Y^{(1)}]^2} + \frac{\Var\left(Y^{(0)}\right)}{(1-e)\mathbb{E}[Y^{(0)}]^2} \\
& = \frac{p_1 ( 1 - p_1)}{e p_1^2} + \frac{p_0 ( 1 - p_0)}{(1-e) p_0^2} \\
& = \frac{1}{ep_1} - \frac{1}{e} + \frac{1}{ep_0} -  \frac{1}{1-e}.
\end{align}
Hence, 
\begin{align}
\frac{V_{\textrm{\tiny log RR,N}}}{n}
& = \frac{1}{e n p_1} - \frac{1}{e n } + \frac{1}{e n p_0} -  \frac{1}{n(1-e)},
\end{align}
which can be estimated replacing $ne$ (resp. $n(1-e)$) by $N_1 = \sum_{i=1}^n T_i$ (resp. $N_0 = n-N_1$ and $enp_1$ (resp. $enp_0$) by $\sum_{i=1}^n Y_i T_i$ (resp. $\sum_{i=1}^n Y_i (1-T_i)$). Replacing $V_{\textrm{\tiny log RR,N}}/n$ by such an estimate in the asymptotic confidence interval \eqref{eq_proof_log_CI1} leads to the well-known formula presented in Equation~\eqref{eq_CI_existing_literature}.

\subsection{Proofs of \Cref{risk_ratio_OBS}}

\subsubsection{Risk Ratio Inverse Propensity Weighting}

\begin{proof}[Proof of \Cref{prop:ipw_normality_obs}]\label{proof:ipw_normality_obs}
\ \\
\textbf{Asymptotic bias and variance of the oracle Risk Ratio IPW estimator} Recall that the oracle Risk Ratio IPW is defined as
\begin{align*}
\tau_{\textrm{\tiny RR,IPW}}^\star = \left( \sum_{i=1}^n  \frac{T_iY_i}{e(X_i)} \right) \big/ \left( \sum_{i=1}^n 
\frac{(1-T_i)Y_i}{1-e(X_i)} \right),
\end{align*}
where the propensity score \(e\) is assumed to be known. Let us define \(g_1(Z) = TY/e(X)\) and \(g_0(Z) = (1-T)Y/(1-e(X))\) with \(Z= (X,T,Y)\). Since 
\begin{align*}
    \frac{m}{1-\eta}\leq g_1(Z) \leq \frac{M}{\eta} \quad \text{and} \quad g_0(Z) \leq \frac{M}{\eta}, 
\end{align*}
the function $g_0$ and $g_1$ are bounded from above and below and thus square integrable. Besides, \(\mathbb{E}\left[g_0(Z_i)\right] = \mathbb{E}\left[Y^{(0)} \right]\) and \(\mathbb{E}\left[g_1(Z_i)\right] = \mathbb{E}\left[Y^{(1)} \right]\). We can therefore apply \Cref{Th2} and conclude that
\begin{align*}
    \sqrt{n} (\tau_{\textrm{\tiny RR,IPW}}^\star - \tau_{\textrm{\tiny RR}}) \to \mathcal{N}(0, V_{\textrm{\tiny RR,IPW}}), 
\end{align*}
where 
\begin{align}
V_{\textrm{\tiny RR,IPW}} =\tau_{\textrm{\tiny RR}}^2\Var\left(\frac{\frac{T_iY_i}{e(X_i)}}{\mathbb{E}\left[Y^{(1)}\right]} - \frac{\frac{(1-T_i)Y_i}{1-e(X_i)}}{\mathbb{E}\left[Y^{(1)}\right]}\right).    
\end{align}   
Moreover, 
\begin{align*}
    \Var\left(\frac{TY}{e(X)}\right) &= \mathbb{E}\left[\left(\frac{TY}{e(X)}\right)^2\right] - \mathbb{E}\left[\frac{TY}{e(X)}\right]^2 \\
    &= \mathbb{E}\left[\frac{TY^2}{e(X)^2}\right] - \mathbb{E}\left[Y^{(1)} \right]^2 \\
    &= \mathbb{E}\left[\frac{1}{e(X)^2}\mathbb{E}\left[T(Y^{(1)})^2 | X\right]\right] - \mathbb{E}\left[Y^{(1)} \right]^2 \\
    &= \mathbb{E}\left[\frac{1}{e(X)}\mathbb{E}\left[(Y^{(1)})^2 | X\right]\right] - \mathbb{E}\left[Y^{(1)} \right]^2 \\
    &= \mathbb{E}\left[\frac{1}{e(X)}\mathbb{E}\left[(Y^{(1)})^2 | X\right]\right] - \mathbb{E}\left[Y^{(1)} \right]^2 \\
    &= \mathbb{E}\left[\frac{(Y^{(1)})^2}{e(X)}\right] - \mathbb{E}\left[Y^{(1)}\right]^2. 
\end{align*}
Similarly
\begin{align*}
    \Var\left(\frac{(1-T)Y}{1-e(X)}\right) &=  \mathbb{E}\left[\frac{(Y^{(0)})^2}{1-e(X)}\right] - \mathbb{E}\left[Y^{(0)}\right]^2. 
\end{align*}
Additionally, the covariance satisfies 
\begin{align*}
    \Cov\left(\frac{TY}{e(X)}, \frac{(1-T)Y}{1-e(X)}\right) &= \mathbb{E}\left[\left(\frac{TY}{e(X)} -\mathbb{E}\left[Y^{(1)} \right]\right)\left(\frac{(1-T)Y}{1-e(X)} -\mathbb{E}\left[Y^{(0)} \right]\right)\right] \\
    &= \mathbb{E}\left[\frac{TY}{e(X)}\frac{(1-T)Y}{1-e(X)}\right] - \mathbb{E}\left[Y^{(1)} \right] \mathbb{E}\left[\frac{(1-T)Y}{1-e(X)} \right]\\
    & \quad - \mathbb{E}\left[Y^{(0)} \right] \mathbb{E}\left[\frac{TY}{e(X)} \right] + \mathbb{E}\left[Y^{(1)} \right]  \mathbb{E}\left[Y^{(0)} \right] \\
    &= -\mathbb{E}\left[Y^{(1)} \right]  \mathbb{E}\left[Y^{(0)} \right].
\end{align*}
Therefore, we get that
\begin{align*}
    V_{\textrm{\tiny RR,IPW}} = \tau_{\textrm{\tiny RR}}^2\left(\frac{\mathbb{E}\left[\frac{(Y^{(1)})^2}{e(X)}\right]}{\mathbb{E}\left[Y^{(1)}\right]^2} + \frac{\mathbb{E}\left[\frac{(Y^{(0)})^2}{1-e(X)}\right]}{\mathbb{E}\left[Y^{(0)}\right]^2}\right).
\end{align*}

As a consequence an estimator \(\hat V_{RR-IPW}\) can be derived:

\begin{align}\label{V_RR-IPW}
    \hat V_{RR-IPW} = \hat{\tau}_{\textrm{\tiny RR,IPW,n}}^2 \left(\frac{\frac{1}{n}\sum_{T_i=1}\left(\frac{Y_i}{\hat e(x_i)}\right)^2}{\left(\frac{1}{n}\sum_{T_i=1}Y_i\right)^2} + \frac{\frac{1}{n}\sum_{T_i=0}\left(\frac{Y_i}{1- \hat e(x_i)}\right)^2}{\left(\frac{1}{n}\sum_{T_i=0}Y_i\right)^2}\right)
\end{align}

Since we have \(\mathbb{E}\left[\left(\frac{TY}{e(X)}\right)^2\right] = \mathbb{E}\left[\frac{(Y^{(1)})^2}{e(X)}\right]\).

\textbf{Finite sample bias and variance of the oracle Risk Ratio IPW estimator} Let \(T_1(\boldsymbol{Z}) = \frac{1}{n}\sum_{i=1}^n  \frac{T_iY_i}{e(X_i)}\) and \(T_0(\boldsymbol{Z}) = \frac{1}{n}\sum_{i=1}^n  \frac{(1-T_i)Y_i}{1-e(X_i)}\) where \(\boldsymbol{Z} = (Z_1, \hdots, Z_n)\). We first show that \(\Var(T_1(\boldsymbol{Z})) = O_p\left(\frac{1}{n}\right)\) and \(\Var(T_0(\boldsymbol{Z})) = O_p\left(\frac{1}{n}\right)\):

\begin{align*}
    \Var(T_1(\boldsymbol{Z})) &= \frac{1}{n^2}\Var\left(\sum_{i=1}^n \frac{T_iY_i}{e(X_i)}\right) \\
    &= \frac{1}{n^2} \sum_{i=1}^n \Var\left(\frac{T_iY_i}{e(X_i)}\right) && \text{(by \hyperref[a:i.i.d.]{i.i.d.})}\\
    &= \frac{1}{n} \left(\mathbb{E}\left[\left(\frac{T_iY_i}{e(X_i)}\right)^2\right] - \mathbb{E}\left[\frac{T_iY_i}{e(X_i)} \right]^2\right) && \text{(by law of total expectation)} \\
    &= \frac{\mathbb{E}\left[\frac{(Y^{(1)})^2}{e(X_i)}\right]-\mathbb{E}\left[Y^{(1)} \right]^2}{n} \\
    &= O_p\left(\frac{1}{n}\right)
\end{align*}
Similarly, $\Var(T_0(\boldsymbol{Z})) = O_p\left(\frac{1}{n}\right)$. And we also have:

\begin{align*}
    \mathbb{E}\left[T_1(\boldsymbol{Z})\right] = \mathbb{E}\left[\frac{T_iY_i}{e(X_i)}\right] = \mathbb{E}\left[Y^{(1)}\right]
\end{align*}

\begin{align*}
    \mathbb{E}\left[T_0(\boldsymbol{Z})\right] = \mathbb{E}\left[\frac{(1-T_i)Y_i}{1-e(X_i)}\right] = \mathbb{E}\left[Y^{(0)}\right]
\end{align*}

Therefore, we showed that \(T_1(\boldsymbol{Z})\) and \(T_0(\boldsymbol{Z})\) are respectively unbiased estimators of \(\mathbb{E}\left[Y^{(1)}\right]\) and \(\mathbb{E}\left[Y^{(0)}\right] > 0\) such that $\Var(T_1(\boldsymbol{Z})) = O_p\left(\frac{1}{n}\right)$ and $\Var(T_0(\boldsymbol{Z})) = O_p\left(\frac{1}{n}\right)$. 
By assumption, 
\begin{align*}
    \frac{m}{1-\eta}\leq T_0(\boldsymbol{Z}) \leq \frac{M}{\eta} \quad \text{and} \quad T_1(\boldsymbol{Z}) \leq \frac{M}{\eta},
\end{align*}
thus \(T_0(\boldsymbol{Z})\) and \(T_1(\boldsymbol{Z})\) are bounded. Applying \Cref{Th2}, we obtain
\[\left|\mathbb{E}\left[\hat{\tau}_{\textrm{\tiny RR, IPW, n}}\right] - \tau_{\textrm{\tiny RR}}\right| \leq \frac{2M^3(1-\eta)^3}{nm^3\eta^3},\]
and
\begin{align*}
\left|\Var(\hat{\tau}_{\textrm{\tiny RR, HT, n}}) - V_{\textrm{\tiny RR, HT}}\right| \leq  \frac{4M^4(1-\eta)^6}{nm^6\eta^4}.
\end{align*}
\end{proof}

\subsubsection{Risk Ratio Inverse Propensity Weighting in logistic models}

\begin{proof}[Proof of \Cref{prop:IPW_MLE}]
The likelihood function $L\left(\boldsymbol{\beta}\right)$ is
$$
L\left(\boldsymbol{\beta}\right) = \prod_{i=1}^{n} \operatorname{P}\left(T=t_{i} \mid X=x_{i}\right)=\prod_{i=1}^{n} e\left(x_{i} ; \boldsymbol{\beta}\right)^{t_{i}}\left(1-e\left(x_{i} ;\boldsymbol{\boldsymbol{\beta}}\right)\right)^{1-t_{i}}
$$
where we define $e\left(X ;\boldsymbol{\boldsymbol{\beta}}\right) = \{1+\exp(-X^\top\beta_1 -\beta_0)\}^{-1}$. Now, taking minus the logarithm of this expression, the log likelihood function, denoted $\ln L\left(\boldsymbol{\beta}\right)$, we obtain,

$$
- \ln L\left(\boldsymbol{\beta}\right) = - \sum_{i=1}^{n} t_i \log(e(X_i ; \boldsymbol{\beta})) + (1-t_i)\log(1-e(X_i ; \boldsymbol{\beta})).
$$

The minimization of this quantity can be obtained when looking for the root of the derivative, so that we obtain the following expression,

\begin{equation*}
  - \frac{\partial}{\partial \boldsymbol{\beta_0}} \ln L\left(\boldsymbol{\beta}\right) =  - \sum_{i=1}^{n} \frac{T_{i}-e\left(X_i; \boldsymbol{\beta}\right)}{e\left(X_i; \boldsymbol{\beta}\right)\left(1-e\left(X_i;\boldsymbol{\beta}\right)\right)}\frac{\partial}{\partial \boldsymbol{\beta_0}}e\left(X_i; \boldsymbol{\beta}\right).
\end{equation*}

\begin{equation*}
  - \frac{\partial}{\partial \boldsymbol{\beta_1}} \ln L\left(\boldsymbol{\beta}\right) =  - \sum_{i=1}^{n} \frac{T_{i}-e\left(X_i; \boldsymbol{\beta}\right)}{e\left(X_i; \boldsymbol{\beta}\right)\left(1-e\left(X_i;\boldsymbol{\beta}\right)\right)}\frac{\partial}{\partial \boldsymbol{\beta_1}}e\left(X_i; \boldsymbol{\beta}\right).
\end{equation*}

Therefore, 

\begin{equation}\label{likelihood_cond}
   \frac{\partial}{\partial \boldsymbol{\beta_0}} \ln L\left(\boldsymbol{\beta}\right)  = - \sum_{i=1}^{n} \left(T_{i}-e\left(X_i; \boldsymbol{\beta}\right)\right) \quad \text{and} \quad \frac{\partial}{\partial \boldsymbol{\beta_1}} \ln L\left(\boldsymbol{\beta}\right)  = - \sum_{i=1}^{n} X_i \left(T_{i}-e\left(X_i; \boldsymbol{\beta}\right)\right) 
\end{equation}

In particular, if we apply \ref{likelihood_cond} for the maximum likelihood estimator $\boldsymbol{\hat\beta}_n$, and if we define $\Tilde{X} := (1, X)$, we have:

\begin{align*} 
\frac{\partial}{\partial \boldsymbol{\beta}} \ln L\left(\boldsymbol{\beta}\right)\bigg|_{\boldsymbol{\beta} = \boldsymbol{\hat\beta}_n} = 0 \quad \Longleftrightarrow \quad 
\sum\limits_{i=1}^{n} \Tilde{X_i} \left(T_{i} - e\left(X_i; \boldsymbol{\hat \beta}_n\right)\right) = 0.
\end{align*}

Let $Z = (X,T,Y)$ and $\boldsymbol{\theta} = (\boldsymbol{\beta}, \theta_{0}, \theta_{1}, \theta_{2})$, we define $\psi$ and $\hat{\boldsymbol{\theta}}_n$ as:

\begin{align*}
    \psi(Z, \boldsymbol{\theta}) =\left(\begin{array}{c}
    \Tilde{X}(T-e(X, \boldsymbol{\beta}))\\
    \frac{(1-T)Y}{1-e(X, \boldsymbol{\beta})}-\theta_{0}\\
    \frac{TY}{e(X, \boldsymbol{\beta})}-\theta_{1} \\
    \theta_1 - \theta_2 \theta_0
    \end{array}\right) \quad \textrm{and} \quad 
    \hat{\boldsymbol{\theta}}_n & =\left(\begin{array}{c}
    \boldsymbol{\hat \beta}_n \\
    \hat \tau_{\textrm{IPW},0} \\
    \hat \tau_{\textrm{IPW},1} \\
    \hat \tau_{\textrm{IPW},1}/\hat \tau_{\textrm{IPW},0}
    \end{array}\right) 
\end{align*}
where $\hat \tau_{\textrm{IPW},1} =  \frac{1}{n}\sum_{i=1}^n \frac{TY}{e(X, \boldsymbol{\hat \beta}_n)}$ and $\hat \tau_{\textrm{IPW},0} =  \frac{1}{n}\sum_{i=1}^n \frac{(1-T)Y}{1-e(X, \boldsymbol{\hat \beta}_n)}$. One can note that:

\begin{align*}
    \sum_{i=1}^n \psi_1(Z_i, \hat{\boldsymbol{\theta}}_n) = \sum_{i=1}^{n} \Tilde{X_i} \left(T_{i}-e\left(X_i; \boldsymbol{\hat \beta}_n\right)\right) = 0.
\end{align*}

\begin{align*}
    \sum_{i=1}^n \psi_2(Z_i, \hat{\boldsymbol{\theta}}_n) = \sum_{i=1}^{n} \left(\frac{(1-T_i)Y_i}{1-e(X_i, \boldsymbol{\hat \beta}_n)} - \frac{1}{n}\sum_{j=1}^n \frac{(1-T_j)Y_j}{1-e(X_j, \boldsymbol{\hat \beta}_n)}\right) = 0.
\end{align*}

\begin{align*}
    \sum_{i=1}^n \psi_3(Z_i, \hat{\boldsymbol{\theta}}_n) = \sum_{i=1}^{n} \left(\frac{T_iY_i}{e(X_i, \boldsymbol{\hat \beta}_n)} - \frac{1}{n}\sum_{j=1}^n \frac{T_jY_j}{e(X_j, \boldsymbol{\hat \beta}_n)}\right) = 0.
\end{align*}

\begin{align*}
    \sum_{i=1}^n \psi_4(Z_i, \hat{\boldsymbol{\theta}}_n) = \sum_{i=1}^{n} \hat \tau_{\textrm{IPW},1} - \underbrace{\frac{\hat \tau_{\textrm{IPW},1}}{\hat \tau_{\textrm{IPW},0}}\hat \tau_{\textrm{IPW},0}}_{ \hat \tau_{\textrm{IPW},1}} = 0.
\end{align*}

Gathering the three previous equalities, we obtain
\begin{align}
\sum_{i=1}^n \psi(Z_i, \boldsymbol{ \hat \theta}_n) = 0,
\end{align}

which proves that \( \hat{\boldsymbol{\theta}}_n\) is an M-estimator of type \(\psi\) \citep[see][]{Stefanski2002Mestimation}.  Furthermore, letting \(\boldsymbol{\theta}_{\infty} = (\boldsymbol{\beta}_\infty, \mathbb{E}\left[Y^{(0)}\right], \mathbb{E}\left[Y^{(1)}\right], \mathbb{E}\left[Y^{(1)}\right]/ \mathbb{E}\left[Y^{(0)}\right])\), we can compute the following quantities:

\[
\begin{aligned}
    \mathbb{E}\left[\psi_1(Z, \boldsymbol{\theta}_\infty) \right] &= \mathbb{E}\left[\Tilde{X}(T - e(X))\right] \\
    &= \mathbb{E}\left[\Tilde{X} \cdot \mathbb{E}\left[T - e(X) \mid X\right]\right] && \text{(Law of Total Probability)} \\
    &= \mathbb{E}\left[\Tilde{X} \cdot \left(\mathbb{E}\left[T \mid X\right] - e(X)\right)\right] && \text{($e(X)$ is a function of $X$)} \\
    &= 0 && \text{(Definition of $e(X)$)} \\
\end{aligned}
\]

Furthermore, note that:

\[
\mathbb{E}\left[\psi_2(Z, \boldsymbol{\theta}_\infty) \right] = \mathbb{E}\left[\frac{1}{n} \sum_{i=1}^{n}\left(\frac{T_{i} Y_{i}}{e(X_i)}\right)\right] - \mathbb{E}[Y^{(1)}],
\]

and the following holds:

\[
\begin{aligned}
    \mathbb{E}\left[\frac{1}{n} \sum_{i=1}^{n}\left(\frac{T_{i} Y_{i}}{e(X_i)}\right)\right] 
    &= \frac{1}{n} \sum_{i=1}^{n} \mathbb{E}\left[\frac{T_{i} Y_{i}}{e(X_i)}\right] && \text{(Linearity of expectation)} \\
    &= \mathbb{E}\left[\frac{T Y^{(1)}}{e(X)}\right] && \text{(Independence and consistency)} \\
    &= \mathbb{E}\left[\mathbb{E}\left[\frac{T Y^{(1)}}{e(X)} \mid X\right]\right] && \text{(Law of Total Probability)} \\
    &= \mathbb{E}\left[\frac{1}{e(X)} \mathbb{E}\left[T Y^{(1)} \mid X\right]\right] && \text{($e(X)$ depends on $X$)} \\
    &= \mathbb{E}\left[\frac{1}{e(X)} \mathbb{E}\left[Y^{(1)} \mid X\right] \mathbb{E}\left[T \mid X\right]\right] && \text{(No confounding assumption)} \\
    &= \mathbb{E}\left[\mathbb{E}\left[Y^{(1)} \mid X\right]\right] && \text{(Definition of $e(X)$)} \\
    &= \mathbb{E}[Y^{(1)}].
\end{aligned}
\]

This shows that \(\mathbb{E}\left[\psi_2(Z, \boldsymbol{\theta}_\infty) \right] = 0\). Similarly, one can show that \(\mathbb{E}\left[\psi_3(Z, \boldsymbol{\theta}_\infty) \right] = 0\). Finally, we also have \(\psi_4(Z, \boldsymbol{\theta}_\infty) = 0\). Therefore, 
\begin{align}\label{M_esti}
    \mathbb{E}\left[\psi(Z, \boldsymbol{\theta}_\infty) \right] = 0. 
\end{align}

We now show that $\boldsymbol{\theta}_\infty$ defined above is the unique value that satisfies \ref{M_esti}. Let 
\[
L(\beta)=-\,\mathbb{E}\Bigl[T \,\ln\bigl(e(X,\beta)\bigr) \;+\;\bigl(1 - T\bigr)\,\ln\bigl(1 - e(X,\beta)\bigr)\Bigr].
\]
A direct calculation shows that 
\[
\nabla_{\beta} L(\beta)=\mathbb{E}\Bigl[\Tilde{X} \,\bigl(e(X,\beta) \;-\; T \bigr)\Bigr] \quad \text{and} \quad \nabla^2_{\beta} L(\beta)=\mathbb{E}\Bigl[\Tilde{X}\,\Tilde{X}^\top \, e(X,\beta)\,\bigl(1 - e(X,\beta)\bigr)\Bigr]
\]
Since $\mathbb{E}[T \mid X] = e(X,\beta_\infty)$, so at $\beta=\beta_\infty$,
\[
\nabla_{\beta} L(\beta_\infty)=\mathbb{E}\Bigl[\Tilde{X} \,\bigl(e(X,\beta_\infty) - T\bigr)\Bigr] = 0 
\]
making $\beta_\infty$ a stationary point. Furthermore, using overlap we have $e(X,\beta)\bigl(1 - e(X,\beta)\bigr) \geq  \eta^2$ therefore $\forall v \in \mathbb{R}^{p+1}$:
\begin{align*}
    v^\top \nabla^2_{\beta} L(\beta)v&=\mathbb{E}\Bigl[||\Tilde{X}^\top v||_2^2 \, e(X,\beta)\,\bigl(1 - e(X,\beta)\bigr)\Bigr]\\
    &\geq \eta^2 \mathbb{E}\Bigl[||\Tilde{X}^\top v||_2^2\Bigr]\\
     &\geq \eta^2 v^\top\mathbb{E}\Bigl[ \Tilde{X}\,\Tilde{X}^\top \,\Bigr]v.
\end{align*}
Since we assumed that $\mathbb{E}\Bigl[
X\,X^\top\Bigr]$ is positive definite, the Hessian $\nabla^2_{\beta} L(\beta)$ is positive definite, so $L(\beta)$ is strictly convex. Hence there is a unique global minimizer of $L(\beta)$; since $\beta_\infty$ is a critical point, it must be that unique minimizer. Consequently, any solution to 
\[
\mathbb{E}\Bigl[\Tilde{X}\,\bigl(e(X,\beta) - T\bigr)\Bigr] = 0
\]
must equal $\beta_\infty$. Since the second and third two components of \(\psi\) are linear with respect to \(\theta_0\) and \(\theta_1\) and since the forth component is linear with respect to \(\theta_2\), \(\boldsymbol{\theta}_{\infty}\) is the only value satisfying \eqref{M_esti}.

We want to show that for every $\boldsymbol{\theta}$ in a neighborhood of $\boldsymbol{\theta_\infty}$, all the components of the second derivatives
\[
\left| \frac{\partial^2}{\partial^2 \boldsymbol{\theta}} \psi_k(z, \boldsymbol{\theta}) \right|
\]
are integrable for all $k \in \{1,4\}$. Since $\boldsymbol{\theta} = (\boldsymbol{\beta}, \theta_{0}, \theta_{1}, \theta_{2})$, we need to show that for $k \in \{1,4\}$ and $i,j \in \{0,2\}$ the following quantities are integrable
\[
\left| \frac{\partial^2}{\partial \theta_i \partial \theta_j} \psi_k(z, \boldsymbol{\theta}) \right| \qquad \left| \frac{\partial^2}{\partial \theta_i \partial \boldsymbol{\beta}} \psi_k(z, \boldsymbol{\theta}) \right| \qquad \left| \frac{\partial^2}{\partial \boldsymbol{\beta} \partial \theta_i} \psi_k(z, \boldsymbol{\theta}) \right| \qquad \left| \frac{\partial^2}{\partial \boldsymbol{\beta} \partial \boldsymbol{\beta}} \psi_k(z, \boldsymbol{\theta}) \right|
\]

One can note that the first three quantities are bounded by 1 and therefore integrable. Hence, it suffices to consider
\[
\left| \frac{\partial^2}{\partial^2 \boldsymbol{\beta}} \psi_k(z, \boldsymbol{\theta}) \right|,
\]
where $k \in \{1,2,3\}$, since $\psi_4$ does not depend on $\theta_1$. For $k=1$, a direct calculation yields
\[
\left|\frac{\partial^2}{\partial^2 \boldsymbol{\beta}} \psi_1(z, \boldsymbol{\theta})\right|
= \bigl|- \Tilde{X_k} \,\Tilde{X_l}\,\Tilde{X_m} \,e(X, \boldsymbol{\beta})\bigl(1-e(X, \boldsymbol{\beta})\bigr)\bigl(1-2e(X, \boldsymbol{\beta})\bigr)\bigr|
\leq \bigl|\Tilde{X_k} \,\Tilde{X_l}\,\Tilde{X_m}\bigr|.
\]
By Cauchy--Schwarz or Hölder’s inequality,
\[
\mathbb{E}\bigl[\bigl|\Tilde{X_k}\,\Tilde{X_l}\,\Tilde{X_m}\bigr|\bigr] 
\;\le\; \mathbb{E}\bigl[(\Tilde{X_k})^2\bigr]^{1/2}\,\mathbb{E}\bigl[(\Tilde{X_l}\,\Tilde{X_m})^2\bigr]^{1/2}
\;\le\; \mathbb{E}\bigl[(\Tilde{X_k})^2\bigr]^{1/2}\,\mathbb{E}\bigl[(\Tilde{X_l})^4\bigr]^{1/4}\,\mathbb{E}\bigl[(\Tilde{X_m})^4\bigr]^{1/4}.
\]
Since $\Tilde{X}$ is sub-Gaussian, it has finite moments of all orders, implying integrability of $\bigl|\Tilde{X_k}\,\Tilde{X_l}\,\Tilde{X_m}\bigr|$.

For $k=2$ or $k=3$, we similarly get
\[
\left|\frac{\partial^2}{\partial^2 \boldsymbol{\beta}} \psi_2(z, \boldsymbol{\theta})\right|
\leq
\bigl|Y\,\exp(2\,\Tilde{X}^\top \boldsymbol{\beta})\,\Tilde{X}\,\Tilde{X}^\top\bigr|,
\]
and
\[
\left|\frac{\partial^2}{\partial^2 \boldsymbol{\theta}} \psi_3(z, \boldsymbol{\theta})\right|
\leq
\bigl|Y\,\exp(2\,\Tilde{X}^\top \boldsymbol{\beta})\,\Tilde{X}\,\Tilde{X}^\top\bigr|.
\]
Hence, it remains to show that
\[
\mathbb{E}\Bigl[\bigl|Y\,\exp(2\,\Tilde{X}^\top \boldsymbol{\beta})\,\Tilde{X}_k\,\Tilde{X}_l\bigr|\Bigr]
\]
is finite. By Cauchy--Schwarz,
\[
\mathbb{E}\Bigl[\bigl|Y\,\exp(2\,\Tilde{X}^\top \boldsymbol{\beta})\,\Tilde{X}_k\,\Tilde{X}_l\bigr|\Bigr]
\;\le\;
\sqrt{\mathbb{E}\bigl[Y^2\bigr]}\,\sqrt{\mathbb{E}\Bigl[\exp\bigl(4\,\Tilde{X}^\top \boldsymbol{\beta}\bigr)\,(\Tilde{X}_k\,\Tilde{X}_l)^2\Bigr]}.
\]
Since $\Tilde{X}$ is sub-Gaussian, its exponential moments are finite. Specifically, for some $\sigma>0$,
\[
\mathbb{E}\bigl[\exp\bigl(\lambda\,\boldsymbol{v}^\top \Tilde{X}\bigr)\bigr] \le \exp\Bigl(\tfrac{\lambda^2 \|\boldsymbol{v}\|_2^2 \sigma^2}{2}\Bigr)\quad \forall \lambda\in\mathbb{R},\ \boldsymbol{v}\in\mathbb{R}^d,
\]
which guarantees $\mathbb{E}[\exp(8\,\Tilde{X}^\top \boldsymbol{\beta})]$ is finite. Moreover, sub-Gaussian random variables have finite polynomial moments, so $\mathbb{E}[(\Tilde{X}_k\,\Tilde{X}_l)^4]$ is also finite. Therefore,
\[
\bigl|Y\,\exp(2\,\Tilde{X}^\top \boldsymbol{\beta})\,\Tilde{X}_k\,\Tilde{X}_l\bigr|
\]
is integrable.

Collecting these results, we conclude that every second derivative
\(\left| \frac{\partial^2}{\partial^2 \boldsymbol{\theta}} \psi_k(z, \boldsymbol{\theta}) \right|\)
is integrable for all $k \in \{1,4\}$ in a neighborhood of $\boldsymbol{\theta_\infty}$.

Define
\[
\begin{aligned}
    A\left(\theta_\infty \right) &= \mathbb{E}\left[\frac{\partial \psi}{\partial \theta}\bigg|_{\theta=\theta_\infty}\right] \quad \text{and} \quad
    B(\theta_\infty) = \mathbb{E}\left[\psi(Z, \theta_\infty) \psi(Z, \theta_\infty)^T\right].
\end{aligned}
\]

Next, we verify the conditions of Theorem~7.2 in \citet{Stefanski2002Mestimation}. To do so, we compute \(A\left(\theta_\infty \right)\) and \(B\left(\theta_\infty \right)\). Since

\begin{align}
     \frac{\partial \psi}{\partial \theta} (Z, \theta)=\left(\begin{array}{cccc}
    -e(X, \boldsymbol{\beta})\left(1-e(X, \boldsymbol{\beta})\right)\Tilde{X}\Tilde{X}^\top & 0 & 0 & 0 \\
    \frac{(1-T)Ye(X, \boldsymbol{\beta})}{1-e(X, \boldsymbol{\beta})}\Tilde{X}^\top &-1 & 0 & 0 \\
    -\frac{TY(1-e(X, \boldsymbol{\beta}))}{e(X, \boldsymbol{\beta})}\Tilde{X}^\top & 0 & -1 & 0 \\
    0 & -\theta_2 & 1 & -\theta_0
    \end{array}\right),  
\end{align}

We obtain

\[
A\left(\boldsymbol{\theta_\infty} \right) = 
\begin{pmatrix}
    -Q & 0 & 0 & 0 \\
    c_{10}^\top & -1 & 0 & 0 \\
    -c_{01}^\top & 0 & -1 & 0 \\
    0 & -\frac{\mathbb{E}\left[Y^{(1)}\right]}{\mathbb{E}\left[Y^{(0)}\right]} & 1 & -\mathbb{E}\left[Y^{(0)}\right]
\end{pmatrix},
\]

where:  
\begin{itemize}
    \item \( Q = \mathbb{E}\left[e(X)(1 - e(X))\Tilde{X}\Tilde{X}^\top\right] \),
    \item \( c_{10} = \mathbb{E}[\Tilde{X}\,e(X)Y^{(0)}] \) and \( c_{01} = \mathbb{E}[\Tilde{X}(1 - e(X))Y^{(1)}] \),
\end{itemize}

which using Schur complement leads to:
\begin{align*}
A^{-1}\left(\boldsymbol{\theta_\infty} \right)= \left(\begin{array}{cccc}
-Q^{-1} & 0 & 0 & 0 \\
-c_{10}^\top Q^{-1} & -1 & 0 & 0\\
c_{01}^\top Q^{-1} & 0 & -1 & 0\\
\left(\frac{\mathbb{E}[Y(1)]}{\mathbb{E}[Y(0)]^2}c_{10}^\top + \frac{1}{\mathbb{E}\left[Y^{(0)}\right]}c_{01}^\top\right)Q^{-1} &  \frac{\mathbb{E}[Y(1)]}{\mathbb{E}[Y(0)]^2} & \frac{-1}{\mathbb{E}\left[Y^{(0)}\right]} &  \frac{-1}{\mathbb{E}\left[Y^{(0)}\right]}\\
\end{array}\right).
\end{align*}

Regarding \(B\left(\theta_\infty \right) \), elementary calculations show that 
       \begin{align*}
        B(\theta_\infty)
        & = \left(\begin{array}{ccccc}
        Q & - c_{10} & c_{01} & 0\\
    - c_{10}^\top& \Var\left(\frac{(1-T)Y}{1-e(X)}\right) & -\mathbb{E}\left[Y^{(1)} \right]  \mathbb{E}\left[Y^{(0)} \right] & 0\\
    c_{01}^\top & 	-\mathbb{E}\left[Y^{(1)} \right]  \mathbb{E}\left[Y^{(0)} \right] & \Var\left(\frac{TY}{e(X)}\right) & 0\\
    0 & 0 & 0 & 0
    \end{array}\right),
    \end{align*}

    Based on the previous calculations, we have
    \begin{itemize}
        \item  \(\psi(z,\boldsymbol{\theta})\)  and its first two partial derivatives with respect to \(\boldsymbol{\theta}\) exist for all \(z\) and for all \(\boldsymbol{\theta}\) in the neighborhood of \(\boldsymbol{\theta_\infty}\).
        \item For each \(\boldsymbol{\theta}\) in the neighborhood of \(\boldsymbol{\theta_\infty}\), we have for all \(k \in \{1, 4\}\)\( \left| \frac{\partial^2}{\partial^2 \boldsymbol{\theta}} \psi_k(z, \boldsymbol{\theta}) \right| \)is integrable.
        \item \(A(\theta_\infty)\) exists and is nonsingular.
        \item \(B(\theta_\infty)\) exists and is finite.
    \end{itemize}
    
    Since we have: 
    \[\sum_{i=1}^n \psi(T_i, Y_i, \hat{\boldsymbol{\theta}}_n) = 0 \quad \text{and} \quad \hat{\boldsymbol{\theta}}_n \stackrel{p}{\rightarrow} \theta_\infty .\]
    Then the conditions of Theorem~7.2 in  \citet{Stefanski2002Mestimation} are satisfied, we have
        \[\sqrt{n}\left(  \hat{\boldsymbol{\theta}}_n - \theta_\infty \right) \stackrel{d}{\rightarrow} \mathcal{N}\left(0, A(\theta_\infty)^{-1}B(\theta_\infty)(A(\theta_\infty)^{-1})^{\top} \right),\]
Since we are only interested in the bottom right term of the sandwich term , we only need to compute \(u_{d+3}^TA(\theta_\infty)^{-1}B(\theta_\infty)(A(\theta_\infty)^{-1})^{\top} u_{d+3}\) where \(u_{d+3}\) is the last vector canonical basis vector of \(\mathbb{R}^{d+3}\). Hence,

\begin{align*}
    \left[A(\theta_\infty)^{-1}B(\theta_\infty)(A(\theta_\infty)^{-1})^{\top}\right]_{d+3, d+3} &= u_{d+3}^TA(\theta_\infty)^{-1}B(\theta_\infty)(A(\theta_\infty)^{-1})^{\top} u_{d+3}\\
    &= u_{d+3}^TA(\theta_\infty)^{-1} B(\theta_\infty)(u_{d+3}^TA(\theta_\infty)^{-1})^{\top}\\
\end{align*}

Noting that \( u_{d+3}^\top A(\theta_\infty)^{-1} = \left(
\left(\frac{\mathbb{E}[Y(1)]}{\mathbb{E}[Y(0)]^2}c_{10}^\top + \frac{1}{\mathbb{E}\left[Y^{(0)}\right]}c_{01}^\top\right)Q^{-1},  \frac{\mathbb{E}[Y(1)]}{\mathbb{E}[Y(0)]^2}, \frac{-1}{\mathbb{E}[Y(0)]},  \frac{-1}{\mathbb{E}[Y(0)]}\right) \) where we used that \((Q^{-1})^\top = Q^{-1}\) since \(Q\) is symmetric. We defining 

\[u_{d+3}^\top A(\theta_\infty)^{-1}B(\theta_\infty) := W = (W_1, W_2, W_3, W_4)\]

where:
\[
   W_1 = \left[\frac{\mathbb{E}[Y(1)]}{\mathbb{E}[Y(0)]^2}c_{10}^\top + \frac{1}{\mathbb{E}[Y(0)]}c_{01}^\top\right]\underbrace{Q^{-1}Q}_{I_d} \;-\; c_{10}^\top \frac{\mathbb{E}[Y(1)]}{\mathbb{E}[Y(0)]^2} \;-\; c_{01}^\top\,\frac{1}{\mathbb{E}[Y(0)]}= 0,
\]

\[
   W_2 = -\left[\frac{\mathbb{E}[Y(1)]}{\mathbb{E}[Y(0)]^2}c_{10}^\top + \frac{1}{\mathbb{E}[Y(0)]}c_{01}^\top\right]Q^{-1}c_{10}  \;+\; \Var\left(\frac{(1-T)Y}{1-e(X)}\right)\frac{\mathbb{E}[Y(1)]}{\mathbb{E}[Y(0)]^2} \;+\; \,\mathbb{E}\left[Y^{(1)} \right]  ,
\]

\[
   W_3 = \left[\frac{\mathbb{E}[Y(1)]}{\mathbb{E}[Y(0)]^2}c_{10}^\top + \frac{1}{\mathbb{E}[Y(0)]}c_{01}^\top\right]Q^{-1}c_{01} \;-\;   \frac{\mathbb{E}[Y(1)]^2}{\mathbb{E}[Y(0)]} \;-\; \,\frac{\Var\left(\frac{TY}{e(X)}\right)}{\mathbb{E}[Y(0)]},
\]

and \(W_4 = 0\). Finally, gathering all the terms we have:

\begin{align*}
    \left[A(\theta_\infty)^{-1}B(\theta_\infty)(A(\theta_\infty)^{-1})^{\top}\right]_{d+3, d+3} &= W_2\frac{\mathbb{E}[Y(1)]}{\mathbb{E}[Y(0)]^2} - W_3\frac{1}{\mathbb{E}[Y(0)]}\\
    &= -\left[\frac{\mathbb{E}[Y(1)]}{\mathbb{E}[Y(0)]^2}c_{10}^\top + \frac{1}{\mathbb{E}[Y(0)]}c_{01}^\top\right]Q^{-1}c_{10}\frac{\mathbb{E}[Y(1)]}{\mathbb{E}[Y(0)]^2}\\
    &-\left[\frac{\mathbb{E}[Y(1)]}{\mathbb{E}[Y(0)]^2}c_{10}^\top + \frac{1}{\mathbb{E}[Y(0)]}c_{01}^\top\right]Q^{-1}c_{01}\frac{1}{\mathbb{E}[Y(0)]}\\
    &+ \Var\left(\frac{(1-T)Y}{1-e(X)}\right)\left(\frac{\mathbb{E}[Y(1)]}{\mathbb{E}[Y(0)]^2}\right)^2 + \frac{\Var\left(\frac{TY}{e(X)}\right)}{\left(\mathbb{E}[Y(0)]\right)^2}\\
    &+2\frac{\mathbb{E}[Y(1)]^2}{\mathbb{E}[Y(0)]^2}.
\end{align*}

One can note that:
\begin{align*}
 \Var\left(\frac{(1-T)Y}{1-e(X)}\right)\left(\frac{\mathbb{E}[Y(1)]}{\mathbb{E}[Y(0)]^2}\right)^2 + \frac{\Var\left(\frac{TY}{e(X)}\right)}{\left(\mathbb{E}[Y(0)]\right)^2} +2\frac{\mathbb{E}[Y(1)]^2}{\mathbb{E}[Y(0)]^2} &= \underbrace{\left(\frac{\mathbb{E}[Y(1)]}{\mathbb{E}[Y(0)]}\right)^2}_{\tau_{\textrm{\tiny RR}}^2}\left(\frac{\Var\left(\frac{(1-T)Y}{1-e(X)}\right)}{\mathbb{E}[Y(0)]^2} + \frac{\Var\left(\frac{TY}{e(X)}\right)}{\mathbb{E}[Y(1)]^2} + 2\right)\\
 &= \tau_{\textrm{\tiny RR}}^2 \left(\frac{\mathbb{E}\left[\frac{(Y^{(1)})^2}{e(X)}\right]}{\mathbb{E}\left[Y^{(1)}\right]^2} + \frac{\mathbb{E}\left[\frac{(Y^{(0)})^2}{1-e(X)}\right]}{\mathbb{E}\left[Y^{(0)}\right]^2}\right)
\end{align*}

where for the last equality we used that:

\begin{align*}
    \Var\left(\frac{TY}{e(X)}\right) = \mathbb{E}\left[\frac{(Y^{(1)})^2}{e(X)}\right] - \mathbb{E}\left[Y^{(1)}\right]^2 \quad \text{and} \quad \Var\left(\frac{(1-T)Y}{1-e(X)}\right) =  \mathbb{E}\left[\frac{(Y^{(0)})^2}{1-e(X)}\right] - \mathbb{E}\left[Y^{(0)}\right]^2
\end{align*}
Finally using calculation we did previously we get that:

\[V_{\textrm{RR-MLE}} = \underbrace{\tau_{\textrm{\tiny RR}}^2\left(\frac{\mathbb{E}\left[\frac{(Y^{(1)})^2}{e(X)}\right]}{\mathbb{E}\left[Y^{(1)}\right]^2} + \frac{\mathbb{E}\left[\frac{(Y^{(0)})^2}{1-e(X)}\right]}{\mathbb{E}\left[Y^{(0)}\right]^2}\right)}_{V_{\textrm{RR-IPW}}} - \tau_{\textrm{\tiny RR}}^2\bigg\| \frac{c_{10}}{\mathbb{E}[Y(0)]} + \frac{c_{01}}{\mathbb{E}[Y(1)]}\bigg\| _{Q^{-1}}^2\]

\end{proof} 

\subsubsection{Risk Ratio G formula estimator}\label{proof:G}

\begin{proof}[Proof of \Cref{prop:g_formula}]
\ \\
\textbf{Asymptotic bias and variance of the oracle risk ratio G formula estimator} Recall that the oracle risk ratio G formula is defined as
\begin{align*}
\tau_{\textrm{\tiny RR,G,n}}^{\star} = \frac{\sum_{i=1}^n  \mu_{(1)}(X_i)}{\sum_{i=1}^n 
\mu_{(0)}(X_i)},
\end{align*}
where the response surfaces \(\mu_{(0)}\) and \(\mu_{(1)}\) are assumed to be known. Let us define \(g_1(Z) = \mu_{(1)}(X_i)\) and \(g_0(Z) = \mu_{(0)}(X_i)\) with \(Z= X\). Since  \(g_1(Z)\) and \(g_0(Z)\) are bounded, they are square integrable. We also have that \(\mathbb{E}\left[g_0(Z_i)\right] = \mathbb{E}\left[Y^{(0)} \right]\) and \(\mathbb{E}\left[g_1(Z_i)\right] = \mathbb{E}\left[Y^{(1)} \right]\). We can therefore apply \Cref{Th2} and conclude that
\begin{align*}
    \sqrt{n} (\tau_{\textrm{\tiny RR,G,n}}^{\star} - \tau_{\textrm{\tiny RR}}) \to \mathcal{N}(0, V_{\textrm{\tiny RR,G}}),
\end{align*}
where \(V_{\textrm{\tiny RR,G}} = \tau_{\textrm{\tiny RR}}^2\Var\left(\frac{\mu_1^{\star}(X)}{\mathbb{E}\left[Y^{(1)} \right]}- \frac{\mu_0^{\star}(X)}{\mathbb{E}\left[Y^{(0)} \right]}\right).\) 
As a consequence an estimator \(\hat V_{RR,G}\) can be derived:

\begin{align}\label{V_RR,G}
    \hat V_{RR,G} =  \frac{\hat{\tau}_{\textrm{\tiny RR,G,n}}^2}{n} \sum_{i=1}^n \left(\frac{\hat \mu_1(X_i)}{\frac{1}{n}\sum_{T_i=1}Y_i} - \frac{\hat \mu_0(X_i)}{\frac{1}{n}\sum_{T_i=0}Y_i } - \frac{1}{n}\sum_{i=1}^n\frac{\hat \mu_1(X_i)}{\frac{1}{n}\sum_{T_i=1}Y_i} - \frac{\hat \mu_0(X_i)}{\frac{1}{n}\sum_{T_i=0}Y_i }\right)^2
\end{align}

\textbf{Finite sample bias and variance of the oracle ratio G formula estimator} Let \(T_1(\boldsymbol{Z}) = \frac{1}{n}\sum_{i=1}^n  \mu_{(1)}(X_i)\) and \(T_0(\boldsymbol{Z}) = \frac{1}{n}\sum_{i=1}^n  \mu_{(0)}(X_i)\) where \(\boldsymbol{Z} = (X_1, \hdots, X_n)\). We first show that \(\Var(T_1(\boldsymbol{Z})) = O_p\left(\frac{1}{n}\right)\) and \(\Var(T_0(\boldsymbol{Z})) = O_p\left(\frac{1}{n}\right)\):
\begin{align*}
    \Var(T_1(\boldsymbol{Z})) &= \frac{1}{n^2}\Var \left(\sum_{i=1}^n \mu_{(1)}(X_i) \right) \\
    &= \frac{1}{n^2} \sum_{i=1}^n \Var (\mu_{(1)}(X_i)) && \text{(by \hyperref[a:i.i.d.]{i.i.d.})}\\
    &= \frac{1}{n} \left(\mathbb{E}\left[(\mu_{(1)}(X_i))^2\right] - \mathbb{E}\left[Y^{(1)} \right]^2\right) && \text{(by law of total expectation)} \\
    &\leq \frac{M^2-\mathbb{E}\left[Y^{(1)} \right]^2}{n} && (\mu_{(1)}(X_i))\leq M)\\
    &= O_p\left(\frac{1}{n}\right).
\end{align*}
Similarly, $\Var(T_0(\boldsymbol{Z}))= O_p\left(1/n\right)$. Since we also have that

\[\mathbb{E}\left[T_1(\boldsymbol{Z})\right]= \mathbb{E}\left[Y^{(1)}\right] \quad \mathbb{E}\left[T_0(\boldsymbol{Z})\right]= \mathbb{E}\left[Y^{(0)}\right] \]

Therefore, we showed that \(T_1(\boldsymbol{Z})\) and \(T_0(\boldsymbol{Z})\) are unbiased estimators of \(\mathbb{E}\left[Y^{(1)}\right]\) and \(\mathbb{E}\left[Y^{(0)}\right] > 0\) such that $\Var(T_1(\boldsymbol{Z}))= O_p\left(1/n\right)$ and $\Var(T_0(\boldsymbol{Z}))= O_p\left(1/n\right)$. We also have that \(T_0(\boldsymbol{Z})\) and \(T_1(\boldsymbol{Z})\) are bounded:
\[m_0\leq T_0(\boldsymbol{Z}) \leq M_0  \quad \text{and} \quad T_1(\boldsymbol{Z}) \leq M_1\]
Applying \Cref{Th2}, under \Cref{a:Outcome_positivity1} we obtain:
\[\left|\mathbb{E}\left[\hat{\tau}_{\textrm{\tiny RR, HT, n}}\right] - \tau_{\textrm{\tiny RR}}\right| \leq \frac{2M_1M_0^2}{nm_0^3} \quad \text{and} \quad \left|\Var(\hat{\tau}_{\textrm{\tiny RR, HT, n}}) - V_{\textrm{\tiny RR, HT}}\right| \leq \frac{2M_0^2M_1(M_1+M_0)}{m_0^6}\]
\end{proof}

\subsubsection{Risk Ratio G-formula in linear models}
\label{app:sec_linear_comparison}

\begin{lemma}[see, e.g. \cite{seber2012linear}]
Grant  \Cref{a:linear_model}
{linear model}. Let \(\gamma_{(t)} = (c_{(t)},\beta_{(t)}) \in \mathbb{R}^{d+1}\) and \(Z = (1, X)\). We rearrange the \(Y_i\) and \(Z_i\) so that the first \(n_1\) observations correspond to \(T = 1\). We then define \(\mathbf{Y}_1 = (Y_1, \ldots, Y_{n_1})^\top\) and \(\mathbf{Y}_0 = (Y_{n_1+1}, \ldots, Y_n)^\top\), as well as \(\mathbf{Z}_1 = (Z_1, \ldots, Z_{n_1})^\top\) and \(\mathbf{Z}_0 = (Z_{n_1+1}, \ldots, Z_n)^\top\). Then for \(t \in \{0,1\}\), the linear model can be formulated as:
\[Y^{(t)} = Z^{\top} \gamma_{(t)}  + \varepsilon_{(t)}, \quad \mathbb{E}[\varepsilon_{(t)} | Z] = 0, \quad \text{Var}[\varepsilon_{(t)} | Z] = \sigma^2,\]
and the least square estimator is given as
\begin{align*}
    \hat \gamma_{(t)} =  \left(\frac{1}{n_t} \mathbf{Z}_t^{\top}\mathbf{Z}_t\right)^{-1} \frac{1}{n_t} \mathbf{Z}_t^{\top}\mathbf{Y}_t
\end{align*}
\end{lemma}

\begin{proposition}
\label{prop:psi_gaus_ber}
 Grant \Cref{a:linear_model}. Let $\hat e  = (\sum_{i=1}^n T_i)/n$ and for all $t \in \{0,1\}$,  
 \begin{align}
   \Bar{Z}_{(t)} & = \frac{1}{\sum_{i=1}^n \mathds{1}_{T_i = t}} \sum_{i=1}^n \mathds{1}_{T_i=t} Z_i.
 \end{align}
Defining \(\nu_t = \mathbb{E}[X|T = t]\) and \(\Sigma_t = \text{Var}(X|T = t)\), we have
\[
\sqrt{n} (\hat{\boldsymbol{\theta}}_n - \boldsymbol{\theta}_\infty)\stackrel{d}{\rightarrow} \mathcal{N}
\left(0, \Sigma\right),\] 
where
\[\boldsymbol{\theta}_n = \begin{pmatrix}
\Bar{Z}_{(0)}\\
\Bar{Z}_{(1)}  \\
\hat \gamma_{(0)}\\
\hat \gamma_{(1)} \\
\hat e
\end{pmatrix} ,  \quad \boldsymbol{\theta}_\infty = \begin{pmatrix}
E[Z|T=0]\\
E[Z|T=1]\\
\gamma_{(0)}\\
\gamma_{(1)}\\
e
\end{pmatrix},  \quad \Sigma = \left(\begin{matrix}\frac{\Var\left[Z|T=0\right]}{\left(1 - e\right)} & 0 & 0 & 0 & 0\\
0 & \frac{\Var\left[Z|T=1\right]}{e} & 0 & 0 & 0\\
0 & 0 & \frac{\sigma^2Q^{-1}_0}{1-e}& 0 & 0\\
0 & 0 & 0 & \frac{\sigma^2Q^{-1}_1}{e} & 0\\
0 & 0 & 0 & 0 & e(1-e)\end{matrix}\right),
\]
with \(Q_t^{-1} =\begin{bmatrix}
1 + \nu_t^T \Sigma_t^{-1} \nu_t & - \nu_t^T \Sigma_t^{-1} \\
- \Sigma_t^{-1} \nu_t& \Sigma_t^{-1}
\end{bmatrix}.
\) 
\end{proposition}
\begin{proof}
Using M-estimation theory to prove asymptotic normality of the \(\theta_n\), we first define the following:
\[
\psi(T, Z, \boldsymbol{\theta}) =\left(\begin{array}{c}
    \psi_0(T, Z, \boldsymbol{\theta})  \\
    \psi_1(T, Z, \boldsymbol{\theta})\\
    \psi_2(T, Z, \boldsymbol{\theta}) \\
    \psi_3(T, Z, \boldsymbol{\theta}) \\
    \psi_4(T, Z, \boldsymbol{\theta})
    \end{array}\right) :=
\begin{pmatrix}
(1-T)(Z-\theta_0) \\
T(Z-\theta_1)  \\
(1-T)\left(Z\epsilon(0) - ZZ^\top\left(\theta_2 -\gamma_{(0)}\right)\right)\\
T\left(Z\epsilon(1) - ZZ^\top\left(\theta_3 -\gamma(1)\right)\right)\\
T-\theta_4
\end{pmatrix}
\]
where \(\boldsymbol{\theta} = (\theta_0, \theta_1, \theta_2, \theta_3, \theta_4)\).
We still have that \( \hat{\boldsymbol{\theta}}_n = (\Bar{Z}_{(0)}, \Bar{Z}_{(1)}, \hat \gamma_{(0)}, \hat \gamma_{(1)}, \hat e)\) is an M-estimator of type \(\psi\) \citep[see][]{Stefanski2002Mestimation} since
\[\sum_{i=1}^n \psi(T_i, Z_i, \hat{\boldsymbol{\theta}}_n) = 0.\]

We now demonstrate that \(\mathbb{E}\left[\psi(T, Y, \boldsymbol{\theta}_\infty)\right]=0\). We directly have that \(\mathbb{E}\left[\psi_4(T, Y, \boldsymbol{\theta}_\infty)\right]=0\). For the other terms we have:
\begin{align*}
    \mathbb{E}\left[\psi_1(T, Z, \boldsymbol{\theta}_\infty)\right] &=\mathbb{E}\left[T\left(Z - \mathbb{E}[Z|T=1]\right)\right] \\
    &=\mathbb{E}\left[\mathbb{E}\left[T\left(Z - \mathbb{E}[Z|T=1]\right)| T\right]\right] \\
    &=\mathbb{E}\left[T\left(\mathbb{E}\left[Z| T\right] - \mathbb{E}[Z|T=1]\right)\right] \\
    &=\mathbb{E}\left[T\left(\mathbb{E}\left[Z| T\right] - \mathbb{E}[Z|T=1]\right)\right] \\
    &= \mathbb{P}\left[T=1\right]\left(\mathbb{E}\left[Z| T=1\right] - \mathbb{E}[Z|T=1]\right) \\
    &= 0
\end{align*}
We also have that:
\begin{align*}
    \mathbb{E}\left[\psi_3(T, Z, \boldsymbol{\theta}_\infty)\right] 
    &=\mathbb{E}\left[TZ\epsilon_{(1)}\right] \\
    &=\mathbb{E}\left[Z\mathbb{E}\left[T\epsilon_{(1)}|Z\right]\right] \\
    &=\mathbb{E}\left[Z\mathbb{E}\left[\epsilon_{(1)}|Z, T=1\right]\right] \\
    &= 0.
\end{align*}
Similarly, we can show:
\[\mathbb{E}\left[\psi_0(T, Z, \boldsymbol{\theta}_\infty)\right] = 0 \quad \text{and} \quad \mathbb{E}\left[\psi_2(T, Z, \boldsymbol{\theta}_\infty)\right] = 0.\]
At this point, we note that since \(\psi(T, Z, \boldsymbol{\theta})\) is a linear function of \(\boldsymbol{\theta}\), \(\theta_\infty\) is the only value of \(\boldsymbol{\theta}\) such that \(\mathbb{E}\left[\psi(T, Z, \boldsymbol{\theta})\right]=0\) We proceed by defining:
\begin{align*}
  A\left(\theta_\infty \right) = \mathbb{E}\left[\frac{\partial \psi}{\partial \theta}\Big\vert_{\theta=\theta_\infty}\right] \quad \textrm{and} \quad B(\theta_\infty) =  \mathbb{E}\left[\psi(T, Z, \theta_\infty) \psi(T, Z, \theta_\infty)^T\right].
\end{align*}

Next, we check the conditions of Theorem~7.2 in \citet{Stefanski2002Mestimation}. First, we compute \(A\left(\theta_\infty \right) \) and \(B\left(\theta_\infty \right)\). Since:
\begin{align*}
 \frac{\partial \psi}{\partial \theta} (T, Z, \theta)=\left(\begin{array}{ccccc}
-(1-T) & 0 & 0 & 0 & 0\\
0 & -T & 0 & 0 & 0\\
0 & 0 & -(1-T)ZZ^\top & 0 & 0\\
0 & 0 & 0 & -TZZ^\top & 0\\
0 & 0 & 0 & 0 & -1
\end{array}\right),  
\end{align*}
we obtain:
\begin{align*} 
A\left(\boldsymbol{\theta_\infty} \right)= \left(\begin{array}{ccccc}
-(1-e) & 0 & 0 & 0 & 0 \\
0 & -e & 0 & 0 & 0 \\
0 & 0 & -(1-e)Q_0 & 0 & 0 \\
0 & 0 & 0 & -eQ_1 & 0 \\
0 & 0 & 0 & 0 & -1
\end{array}\right), \quad \text{where}  \quad Q_t = \mathbb{E}\left[ZZ^\top|T=t\right].
\end{align*}
which leads to:
\begin{align*}
A^{-1}\left(\boldsymbol{\theta_\infty} \right)= \left(\begin{array}{ccccc}
-\frac{1}{1-e} & 0 & 0 & 0 & 0\\
0 & -\frac{1}{e} & 0 & 0 & 0 \\
0 & 0 & -\frac{Q^{-1}_0}{1-e} & 0 & 0\\
0 & 0 & 0 & -\frac{Q^{-1}_1}{e} & 0\\
0 & 0 & 0 & 0 & -1\\
\end{array}\right).
\end{align*}

Regarding \(B(\boldsymbol{\theta_\infty})\), since we have \(T(1-T) = 0\), elementary calculations show that:

\[\begin{array}{cc}B(\boldsymbol{\theta_\infty})_{1,2} = B(\boldsymbol{\theta_\infty})_{2,1} = 0\\
B(\boldsymbol{\theta_\infty})_{3,4} = B(\boldsymbol{\theta_\infty})_{4,3} = 0\end{array}\quad \text{and} \quad \begin{array}{cc}B(\boldsymbol{\theta_\infty})_{1,4} = B(\boldsymbol{\theta_\infty})_{4,1} = 0 \\ B(\boldsymbol{\theta_\infty})_{2,3} = B(\boldsymbol{\theta_\infty})_{3,2} = 0.\end{array}\]

Besides
\begin{align*}
B(\boldsymbol{\theta_\infty})_{2,2} &= E\left[T^2(Z-E\left[Z|T=1\right])(Z-E\left[Z|T=1\right])^\top\right] \\
&= E\left[T(Z-E\left[Z|T=1\right])(Z-E\left[Z|T=1\right])^\top\right] \\
&= E\left[TE\left[(Z-E\left[Z|T=1\right])(Z-E\left[Z|T=1\right])^\top|T\right]\right] \\
&=\mathbb{P}\left[T=1\right]E\left[(Z-E\left[Z|T=1\right])(Z-E\left[Z|T=1\right])^\top|T=1\right] \\
&= e\Var\left[Z|T=1\right],
\end{align*}
 and similarly,
 \[B(\boldsymbol{\theta_\infty})_{1,1} = (1-e)\Var\left[Z|T=0\right].\]
We can also note that:
\begin{align*}
B(\boldsymbol{\theta_\infty})_{4,4} = E\left[T^2ZZ^\top\epsilon_{(1)}^2\right] &= E\left[TZZ^\top\epsilon_{(1)}^2\right] \\
&= E\left[TE\left[ZZ^\top\epsilon_{(1)}^2|T\right]\right] \\
&= \mathbb{P}\left[T=1\right]E\left[ZZ^\top\epsilon_{(1)}^2|T=1\right] \\
&= eE\left[ZZ^\top E\left[\epsilon_{(1)}^2|T=1,Z\right]|T=1\right] \\
&= e\sigma^2E\left[ZZ^\top|T=1\right] :=e\sigma^2Q_1,
\end{align*}
and similarly,
\[B(\boldsymbol{\theta_\infty})_{3,3} = (1-e)\sigma^2Q_0.\]
Finally, 
\begin{align*}
B(\boldsymbol{\theta_\infty})_{2,4} = B(\boldsymbol{\theta_\infty})_{4,2} &= E\left[T^2(Z-E\left[Z|T=1\right])Z^\top\epsilon_{(1)}\right]\\
&= E\left[T(Z-E\left[Z|T=1\right])Z^\top\epsilon_{(1)}\right] \\
&= \mathbb{P}\left[T=1\right]E\left[(Z-E\left[Z|T=1\right])Z^\top\epsilon_{(1)}|T=1\right] \\
&= eE\left[(Z-E\left[Z|T=1\right])Z^\top E\left[\epsilon_{(1)}|T=1, Z\right]|T=1\right] \\
&= 0,
\end{align*}
and similarly,
\[B(\boldsymbol{\theta_\infty})_{1,3} = B(\boldsymbol{\theta_\infty})_{3,1} = 0.\]
We also have that:
\begin{align*}
B(\boldsymbol{\theta_\infty})_{2,5} = B(\boldsymbol{\theta_\infty})_{5,2} &= E\left[T(Z-E\left[Z|T=1\right])(T-e)\right]\\
&= E\left[T^2Z - T^2E\left[Z|T=1\right] -eTZ + eTE\left[Z|T=1\right]\right]\\
&= E\left[T^2Z - T^2E\left[Z|T=1\right] -eTZ + eTE\left[Z|T=1\right]\right]\\
&= eE\left[Z|T=1\right] - eE\left[Z|T=1\right] - e^2E\left[Z|T=1\right] + e^2E\left[Z|T=1\right]\\
&=0
\end{align*}
and similarly,
\[B(\boldsymbol{\theta_\infty})_{1,5} = B(\boldsymbol{\theta_\infty})_{5,1} = 0.\]
We also have that :
\begin{align*}
    B(\boldsymbol{\theta_\infty})_{4,5} = B(\boldsymbol{\theta_\infty})_{5,4} &= E\left[(T-e)TZ\epsilon_{(0)}\right]\\
    &= E\left[(TZ\epsilon_{(0)}\right] - eE\left[(TZ\epsilon_{(0)}\right]\\
    &= (1-e)E\left[TZ\epsilon_{(0)}\right]\\
    &=(1-e)\mathbb{E}\left[Z\mathbb{E}\left[T\epsilon_{(0)}|Z\right]\right] \\
    &=(1-e)\mathbb{E}\left[Z\mathbb{E}\left[\epsilon_{(0)}|Z, T=1\right]\right] \\
    &= 0
\end{align*}
and similarly,
\[B(\boldsymbol{\theta_\infty})_{3,5} = B(\boldsymbol{\theta_\infty})_{5,3} = 0.\]
Gathering all calculations, and since  $B(\boldsymbol{\theta_\infty})_{5,5} =e(1-e)$, we have
\begin{align*}
B(\boldsymbol{\theta_\infty}) = \left(\begin{array}{ccccc}
    (1-e)\Var\left[Z|T=0\right] & 0 & 0 & 0 & 0\\
    0 & e\Var\left[Z|T=1\right] & 0 & 0 & 0 \\
    0  & 0 & (1-e)\sigma^2Q_0 & 0 & 0 \\
    0  & 0 & 0 & e\sigma^2Q_1 & 0\\
    0 & 0 & 0 & 0 & e(1-e)
    \end{array}\right),
\end{align*}

Based on the previous calculations, we have:
\begin{itemize}
    \item  \(\psi(z,\boldsymbol{\theta})\) and its first two partial derivatives with respect to \(\boldsymbol{\theta}\) exist for all \(z\) and for all \(\boldsymbol{\theta}\) in the neighborhood of \(\boldsymbol{\theta_\infty}\).
    \item For each \(\boldsymbol{\theta}\) in the neighborhood of \(\boldsymbol{\theta_\infty}\), we have for all \(i,j,k \in \{0, 2\}\):
    \[\left| \frac{\partial^2}{\partial \theta_i \partial \theta_j} \psi_k(z, \boldsymbol{\theta}) \right| \leq 1\]
    and 1 is integrable.
    \item \(A(\theta_\infty)\) exists and is nonsingular.
    \item \(B(\theta_\infty)\) exists and is finite.
\end{itemize}

Since we have: 
\[\sum_{i=1}^n \psi(T_i, Z_i, \hat{\boldsymbol{\theta}}_n) = 0 \quad \text{and} \quad \hat{\boldsymbol{\theta}}_n \stackrel{p}{\rightarrow} \theta_\infty .\]

Then, the conditions of Theorem~7.2 in \citet{Stefanski2002Mestimation} are satisfied, we have:
\[\sqrt{n}\left(  \hat{\boldsymbol{\theta}}_n - \theta_\infty \right) \stackrel{d}{\rightarrow} \mathcal{N}\left(0, A(\theta_\infty)^{-1}B(\theta_\infty)(A(\theta_\infty)^{-1})^{\top} \right),\]
where:
\begin{align*}
    A(\theta_\infty)^{-1}B(\theta_\infty)(A(\theta_\infty)^{-1})^{\top} \nonumber  = \left(\begin{matrix}\frac{\Var\left[Z|T=0\right]}{\left(1 - e\right)} & 0 & 0 & 0 & 0\\
    0 & \frac{\Var\left[Z|T=1\right]}{e} & 0 & 0 & 0\\
    0 & 0 & \frac{\sigma^2Q^{-1}_0}{1-e}& 0 & 0\\
    0 & 0 & 0 & \frac{\sigma^2Q^{-1}_1}{e} & 0\\
    0 & 0 & 0 & 0 & e(1-e)
    \end{matrix}\right),
\end{align*}
\end{proof}

\begin{proposition}[asymptotical normality of \(\hat{\tau}_{\textrm{RR,OLS}}\)]
 Assume we have \hyperref[a:linear_model]{linear model} then we have:
\[\sqrt{n}(\hat{\tau}_{\textrm{RR,OLS}}-\tau_{\textrm{RR}}) \stackrel{d}{\rightarrow} \mathcal{N}\left(0, V_{\textrm{RR-OLS}}\right)\]
with 
\[\frac{V_{\textrm{RR-OLS}}}{\tau_{\textrm{\tiny RR}}^2} =  \left\Vert\frac{\beta_{(1)}}{\mathbb{E}\left[Y^{(1)}\right]} - \frac{\beta_{(0)}}{\mathbb{E}\left[Y^{(0)}\right]}\right\Vert_{\Sigma}^2 + \sigma^2 \left(\frac{1+(1-e)^2\|\nu_1 - \nu_0\|^2_{\Sigma_1^{-1}}}{e\mathbb{E}\left[Y^{(1)}\right]^2} + \frac{1+e^2\|\nu_1 - \nu_0\|^2_{\Sigma_0^{-1}}}{(1-e)\mathbb{E}\left[Y^{(0)}\right]^2}\right).
\]
\end{proposition}

\begin{proof}
    Let $\hat{\beta}_{(1)}$ and $\hat{c}_{(1)}$ be the parameters obtained via fitting an ordinary least square method on the treated individuals only, that is 
    \begin{align}
        (\hat{\beta}_{(1)}, \hat{c}_{(1)}) \in \arg \min_{c_{(1)}, \beta_{(1)}} \sum_{i=1}^n (Y_i^{(1)} - c_{(1)} - \beta_{(1)} X_i)^2 \mathds{1}_{T_i=1}.
    \end{align}
    Similarly, let $\hat{\beta}_{(0)}$ and $\hat{c}_{(0)}$ be the parameters obtained via fitting an ordinary least square method on the control individuals only, that is 
    \begin{align}
        (\hat{\beta}_{(0)}, \hat{c}_{(0)}) \in \arg \min_{c_{(0)}, \beta_{(0)}} \sum_{i=1}^n (Y_i^{(0)} - c_{(0)} - \beta_{(0)} X_i)^2 \mathds{1}_{T_i=0}.
    \end{align}
    An estimator of the RR using the G-formula approach is thus given by 
    \begin{align}
    \hat{\tau}_{\textrm{RR,OLS}}   &= \frac{\sum_{i=1}^n \left( \hat{c}_{(1)} + X_i^\top \hat{\beta}_{(1)} \right)}{\sum_{i=1}^n \left( \hat{c}_{(0)} + X_i^\top \hat{\beta}_{(0)} \right)} \\
     &= \frac{  \hat{c}_{(1)} + \bar{X}^\top \hat{\beta}_{(1)}   }{  \hat{c}_{(0)} + \bar{X}^\top \hat{\beta}_{(0)} }.
    \end{align}
    Besides, note that assuming a linear model implies that 
    \begin{align}
    \hat{\tau}_{\textrm{RR,OLS}}   &=     \frac{  c_{(1)} + \mathds{E}[X] ^\top \beta_{(1)}   }{  c_{(0)} + \mathds{E}[X]^\top  \beta_{(0)} }.
    \end{align}
    Let, for all $i$, $Z_i = (1, X_i)$ and $\gamma_{(j)} = (c_{(j)}, \beta_{(j)})$ for all $j \in \{0,1\}$. Expanding the following difference, we have:
\begin{align}
    \sqrt{n}(\tau_{\textrm{RR,OLS}}-\tau_{\textrm{RR}}) &= \sqrt{n}\left(\frac{\hat c_{(1)} + \Bar{X}^{\top} \hat \beta_{(1)}}{\hat c_{(0)} + \Bar{X}^{\top} \hat \beta_{(0)}} - \frac{c_{(1)} + \mathbb{E}[X]^{\top} \beta_{(1)}}{ c_{(0)} + \mathbb{E}[X]^{\top} \beta_{(0)}}\right)\\
    &= \sqrt{n}\left(\hat c_{(1)} + \Bar{X}^{\top} \hat \beta_{(1)}\right)\left(\frac{1}{\hat c_{(0)} + \Bar{X}^{\top} \hat \beta_{(0)}} - \frac{1}{c_{(0)} + \mathbb{E}[X]^{\top} \beta_{(0)}}\right)\\
    & \quad + \frac{\sqrt{n}}{ c_{(0)} + \mathbb{E}[X]^{\top}  \beta_{(0)}}\left(\hat c_{(1)} + \Bar{X}^{\top} \hat \beta_{(1)} - c_{(1)} - \mathbb{E}[X]^{\top} \beta_{(1)}\right)\\
    &= \sqrt{n}\left(\Bar{Z}^{\top} \hat \gamma_{(1)}\right)\left(\frac{1}{\Bar{Z}^{\top} \hat \gamma_{(0)}} - \frac{1}{\mathbb{E}[Z]^{\top} \gamma_{(0)}}\right)\\
    & \quad + \frac{\sqrt{n}}{\mathbb{E}[Z]^{\top}  \gamma_{(0)}}\left(\Bar{Z}^{\top} \hat \gamma_{(1)} - \mathbb{E}[Z]^{\top} \gamma_{(1)}\right)\\
    &= \sqrt{n} \frac{\Bar{Z}^{\top} \hat \gamma_{(1)}}{\Bar{Z}^{\top} \hat \gamma_{(0)}\mathbb{E}[Z]^{\top} \gamma_{(0)}}\left(\mathbb{E}[Z]^{\top} \gamma_{(0)} -\Bar{Z}^{\top} \hat \gamma_{(0)}\right)\\
    & \quad + \frac{\sqrt{n}}{\mathbb{E}[Z]^{\top}  \gamma_{(0)}}\left(\Bar{Z}^{\top} \hat \gamma_{(1)} - \mathbb{E}[Z]^{\top} \gamma_{(1)}\right).
\end{align}

Besides, we have
\begin{align*}
    \bar{Z} - \mathbb{E}[Z] &= \hat e \bar{Z}_{(1)} + (1 - \hat e)\bar{Z}_{(0)} - e \mathbb{E}[Z | T = 1] - (1 - e) \mathbb{E}[Z | T = 0] \\
    &= (1 - e)\left(\bar{Z}_{(0)} - \mathbb{E}[Z | T = 0]\right) + e\left(\bar{Z}_{(1)} - \mathbb{E}[Z | T = 1]\right)  + \left( \bar{Z}_{(1)} - \bar{Z}_{(0)}\right) (\hat e -  e)\\
    &= \zeta (\theta_n-\theta_{\infty}), 
\end{align*}
where $\zeta = \left[ (1 - e)I_{d+1}, \; eI_{d+1}, \; 0_{d+1}, \; 0_{d+1}, \; (\bar{Z}_{(1)} - \bar{Z}_{(0)}) \right] \in \mathds{R}^{(d+1)\times 4(d+1)+1}$ and 
\[\boldsymbol{\theta}_n = \begin{pmatrix}
\Bar{Z}_{(0)}\\
\Bar{Z}_{(1)}  \\
\hat \gamma_{(0)}\\
\hat \gamma_{(1)} \\
\hat e
\end{pmatrix} ,  \quad \boldsymbol{\theta}_\infty = \begin{pmatrix}
E[Z|T=0]\\
E[Z|T=1]\\
\gamma_{(0)}\\
\gamma_{(1)}\\
e
\end{pmatrix}.
\]

Note that for all \(t \in \{0,1\}\), 
\begin{align*}
    \Bar{Z}^{\top} \hat \gamma_{(t)} - \mathbb{E}[Z]^{\top} \gamma_{(t)} &= \hat \gamma_{(t)}^{\top}\left(\Bar{Z}- \mathbb{E}[Z]\right)  + \mathbb{E}[Z]^{\top}\left(\hat \gamma_{(t)} - \gamma_{(t)}\right)\\
    &= \hat \gamma_{(t)}^{\top}\zeta (\theta_n-\theta_{\infty}) + \mathbb{E}[Z]^{\top}\left(\hat \gamma_{(t)} - \gamma_{(t)}\right)\\
    &= \hat{\alpha}_{(t)}^{\top} (\theta_n-\theta_{\infty}),
\end{align*}
with 
\[\hat{\alpha}_{(t)} = 
\begin{pmatrix}
(1-e)\hat \gamma_{(t)}\\
e \hat \gamma_{(t)}\\
\mathds{1}_{t=0}\mathbb{E}[Z] \\
\mathds{1}_{t=1}\mathbb{E}[Z]\\
(\bar{Z}_{(1)}-\bar{Z}_{(0)})^\top\hat \gamma_{(t)}.
\end{pmatrix}\].
Therefore
\begin{align*}
    \sqrt{n}(\tau_{\textrm{RR,OLS}}-\tau_{\textrm{RR}}) 
     &= \sqrt{n} \frac{\Bar{Z}^{\top} \hat \gamma_{(1)}}{\Bar{Z}^{\top} \hat \gamma_{(0)}\mathbb{E}[Z]^{\top} \gamma_{(0)}}\left(\mathbb{E}[Z]^{\top} \gamma_{(0)} -\Bar{Z}^{\top} \hat \gamma_{(0)}\right)\\
    &+ \frac{\sqrt{n}}{\mathbb{E}[Z]^{\top} \gamma_{(0)}}\left(\Bar{Z}^{\top} \hat \gamma_{(1)} - \mathbb{E}[Z]^{\top} \gamma_{(1)}\right)\\
    &=\frac{\sqrt{n}}{\mathbb{E}[Z]^{\top} \gamma_{(0)}} \hat{\alpha}_{(1)}^{\top}(\theta_n-\theta_{\infty})\\
    &-\sqrt{n} \frac{\Bar{Z}^{\top} \hat \gamma_{(1)}}{\Bar{Z}^{\top} \hat \gamma_{(0)}\mathbb{E}[Z]^{\top} \gamma_{(0)}} \hat{\alpha}_{(0)}^{\top}(\theta_n-\theta_{\infty})
\end{align*}

Therefore, we get that
\begin{align*}
    \sqrt{n}(\tau_{\textrm{RR,OLS}}-\tau_{\textrm{RR}})   &= \sqrt{n} \left(\frac{1}{\mathbb{E}[Z]^{\top} \gamma_{(0)}} \hat{\alpha}_{(1)}- \frac{\Bar{Z}^{\top} \hat \gamma_{(1)}}{\Bar{Z}^{\top} \hat \gamma_{(0)}\mathbb{E}[Z]^{\top} \gamma_{(0)}}\hat{\alpha}_{(0)}\right)^\top(\theta_n - \theta_\infty).
\end{align*}
According to the Law of Large Numbers,  
\begin{align*}
\frac{1}{\mathbb{E}[Z]^{\top} \gamma_{(0)}} \hat{\alpha}_{(1)}- \frac{\Bar{Z}^{\top} \hat \gamma_{(1)}}{\Bar{Z}^{\top} \hat \gamma_{(0)}\mathbb{E}[Z]^{\top} \gamma_{(0)}}\hat{\alpha}_{(0)}
&\stackrel{p}{\rightarrow} \frac{\mathbb{E}[Z]^{\top} \gamma_{(1)}}{\mathbb{E}[Z]^{\top} \gamma_{(0)}} \left(\frac{\alpha_{(1)}}{\mathbb{E}[Z]^{\top} \gamma_{(1)}}- \frac{\alpha_{(0)}}{\mathbb{E}[Z]^{\top} \gamma_{(0)}}\right):= \alpha_\infty,
\end{align*}
with, for all $t \in \{0,1\}$, 
\[\alpha_{(t)} = \begin{pmatrix}
(1-e)\gamma_{(t)}\\
e\gamma_{(t)}\\
\mathds{1}_{t=0}\mathbb{E}[Z] \\
\mathds{1}_{t=1}\mathbb{E}[Z]\\
(\mathbb{E}[Z|T=1]-\mathbb{E}[Z|T=0])^\top\gamma_{(t)}
\end{pmatrix} 
\]
and
\begin{align}
\alpha_\infty = \frac{\mathbb{E}[Z]^{\top} \gamma_{(1)}}{\mathbb{E}[Z]^{\top} \gamma_{(0)}}\begin{pmatrix}
\frac{(1-e)\gamma_{(1)}}{\mathbb{E}[Z]^{\top} \gamma_{(1)}} - \frac{(1-e)\gamma_{(0)}}{\mathbb{E}[Z]^{\top} \gamma_{(0)}}\\
\frac{e\gamma_{(1)}}{\mathbb{E}[Z]^{\top} \gamma_{(1)}} - \frac{e\gamma_{(0)}}{\mathbb{E}[Z]^{\top} \gamma_{(0)}}\\
-\frac{\mathbb{E}[Z]}{\mathbb{E}[Z]^{\top} \gamma_{(0)}} \\
\frac{\mathbb{E}[Z]}{\mathbb{E}[Z]^{\top} \gamma_{(1)}}\\
\frac{\gamma_{(0)}^\top (\mathbb{E}[Z|T=1]-\mathbb{E}[Z|T=0])}{\mathbb{E}[Z]^{\top} \gamma_{(0)}} - \frac{\gamma_{(1)}^\top (\mathbb{E}[Z|T=1]-\mathbb{E}[Z|T=0])}{\mathbb{E}[Z]^{\top} \gamma_{(1)}}
\end{pmatrix}.
\end{align}
According to \Cref{prop:psi_gaus_ber}, letting $Q_t = \mathbb{E}\left[ZZ^\top|T=t\right]$ for all $t \in \{0,1\}$, we have
\[\sqrt{n}(\theta_n - \theta_\infty) \stackrel{d}{\rightarrow}\mathcal{N}\left(0, \Sigma\right) \quad \text{where} \quad  \Sigma = \left(\begin{matrix}\frac{\Var\left[Z|T=0\right]}{\left(1 - e\right)} & 0 & 0 & 0 & 0\\
    0 & \frac{\Var\left[Z|T=1\right]}{e} & 0 & 0 & 0\\
    0 & 0 & \frac{\sigma^2Q^{-1}_0}{1-e}& 0 & 0\\
    0 & 0 & 0 & \frac{\sigma^2Q^{-1}_1}{e} & 0\\
    0 & 0 & 0 & 0 & e(1-e)
    \end{matrix}\right).\]
By Slutsky's theorem, 
\begin{align}
    \sqrt{n} \left(\frac{1}{\mathbb{E}[Z]^{\top} \gamma_{(0)}} \hat{\alpha}_{(1)}- \frac{\Bar{Z}^{\top} \hat \gamma_{(1)}}{\Bar{Z}^{\top} \hat \gamma_{(0)}\mathbb{E}[Z]^{\top} \gamma_{(0)}}\hat{\alpha}_{(0)}\right)^\top(\theta_n - \theta_\infty) \stackrel{d}{\rightarrow}\mathcal{N}\left(0, \alpha_\infty^\top \Sigma \alpha_\infty \right).
\end{align}
We now compute the covariance matrix
%\begin{align*}
 %   \frac{\alpha_\infty^\top \Sigma \alpha_\infty}{\left(\frac{\mathbb{E}[Z]^{\top} \gamma_{(1)}}{\mathbb{E}[Z]^{\top} \gamma_{(0)}}\right)^2} &= (1-e)\left\Vert\frac{\gamma_{(1)}}{\mathbb{E}[Z]^{\top} \gamma_{(1)}} - \frac{\gamma_{(0)}}{\mathbb{E}[Z]^{\top} \gamma_{(0)}}\right\Vert_{\Var\left[Z|T=0\right]}^2\\
 %   &+e\left\Vert\frac{\gamma_{(1)}}{\mathbb{E}[Z]^{\top} \gamma_{(1)}} - \frac{\gamma_{(0)}}{\mathbb{E}[Z]^{\top} \gamma_{(0)}}\right\Vert_{\Var\left[Z|T=1\right]}^2\\
  %  &+\frac{\sigma^2}{1-e}\left\Vert\frac{\mathbb{E}[Z]}{\mathbb{E}[Z]^{\top} \gamma_{(0)}}\right\Vert_{Q_0^{-1}}^2 + \frac{\sigma^2}{e}\left\Vert\frac{\mathbb{E}[Z]}{\mathbb{E}[Z]^{\top} \gamma_{(1)}}\right\Vert_{Q_1^{-1}}^2 \\
   % &+ e(1-e)\left(\frac{\gamma_{(0)}^\top (\mathbb{E}[Z|T=1]-\mathbb{E}[Z|T=0])}{\mathbb{E}[Z]^{\top} \gamma_{(0)}} - \frac{\gamma_{(1)}^\top (\mathbb{E}[Z|T=1]-\mathbb{E}[Z|T=0])}{\mathbb{E}[Z]^{\top} \gamma_{(1)}}\right)^2
%\end{align*}
\begin{align}
    \frac{\alpha_\infty^\top \Sigma \alpha_\infty}{\left(\frac{\mathbb{E}[Z]^{\top} \gamma_{(1)}}{\mathbb{E}[Z]^{\top} \gamma_{(0)}}\right)^2} &= (1-e)\left\Vert\frac{\gamma_{(1)}}{\mathbb{E}[Z]^{\top} \gamma_{(1)}} - \frac{\gamma_{(0)}}{\mathbb{E}[Z]^{\top} \gamma_{(0)}}\right\Vert_{\Var\left[Z|T=0\right]}^2 +e\left\Vert\frac{\gamma_{(1)}}{\mathbb{E}[Z]^{\top} \gamma_{(1)}} - \frac{\gamma_{(0)}}{\mathbb{E}[Z]^{\top} \gamma_{(0)}}\right\Vert_{\Var\left[Z|T=1\right]}^2\\
    & \quad +\frac{\sigma^2}{1-e}\left\Vert\frac{\mathbb{E}[Z]}{\mathbb{E}[Z]^{\top} \gamma_{(0)}}\right\Vert_{Q_0^{-1}}^2 + \frac{\sigma^2}{e}\left\Vert\frac{\mathbb{E}[Z]}{\mathbb{E}[Z]^{\top} \gamma_{(1)}}\right\Vert_{Q_1^{-1}}^2 + e(1-e)\left\Vert\frac{\gamma_{(1)}}{\mathbb{E}[Z]^{\top} \gamma_{(1)}} - \frac{\gamma_{(0)}}{\mathbb{E}[Z]^{\top} \gamma_{(0)}}\right\Vert_{
    \Delta \Delta^\top}^2, \label{eq_proof_linear_model_RR_G_formula}
\end{align}
where $\Delta = \mathbb{E}[Z \mid T = 1] - \mathbb{E}[Z \mid T = 0].$ This variance can be rewritten as follows. Summing the first two terms and the last term in \eqref{eq_proof_linear_model_RR_G_formula} leads to 
\[\left\Vert\frac{\gamma_{(1)}}{\mathbb{E}[Z]^{\top} \gamma_{(1)}} - \frac{\gamma_{(0)}}{\mathbb{E}[Z]^{\top} \gamma_{(0)}}\right\Vert_{J}^2,\]

where \(J = (1 - e) \operatorname{Var}(Z \mid T = 0) + e \operatorname{Var}(Z \mid T = 1) + e(1 - e) \Delta \Delta^\top\).
Let us prove that $J = \Var(Z)$. Letting $Z_i$ the components of $Z$ for all $1 \leq i \leq d+1$, by the law of total covariance, we have
\begin{align}
    \Cov [Z_i, Z_j] = \E[\Cov[Z_i, Z_j |T]] + \Cov[\E[Z_i|T], \E[Z_j|T]], \label{proof_total_covariance}
\end{align}
with, since $T \in \{0,1\}$, 
\begin{align}
    \E[\Cov[Z_i, Z_j |T]] = e \Cov[Z_i, Z_j |T =1] + (1-e) \Cov[Z_i, Z_j |T =0].
\end{align}
%Note that J can is in fact \(\Var(J)\). Hence,if we recall the Law of Total Variance for a random vector \( Z \) and a random variable \( T \):
%\[
%\operatorname{Var}(Z) = \mathbb{E}[\operatorname{Var}(Z \mid T)] + \operatorname{Var}(\mathbb{E}[Z \mid T]).
%\]

%Since \( T \) is binary,
%\[
%\mathbb{E}[\operatorname{Var}(Z \mid T)] = (1 - e) \operatorname{Var}(Z \mid T = 0) + e \operatorname{Var}(Z \mid T = 1).
%\]
Besides, since $
\mathbb{E}[Z] = (1 - e) \mathbb{E}[Z \mid T = 0] + e \mathbb{E}[Z \mid T = 1]$, we can compute the deviations from the unconditional mean:
\[
\begin{aligned}
\mathbb{E}[Z \mid T = 0] - \mathbb{E}[Z] &= \mathbb{E}[Z \mid T = 0] - \left( (1 - e) \mathbb{E}[Z \mid T = 0] + e \mathbb{E}[Z \mid T = 1] \right) \\
&= (1 - (1 - e)) \mathbb{E}[Z \mid T = 0] - e \mathbb{E}[Z \mid T = 1] \\
&= -e \left( \mathbb{E}[Z \mid T = 1] - \mathbb{E}[Z \mid T = 0] \right) \\
&= -e \Delta.
\end{aligned}
\]
and 
\[
\begin{aligned}
\mathbb{E}[Z \mid T = 1] - \mathbb{E}[Z] &= \mathbb{E}[Z \mid T = 1] - \left( (1 - e) \mathbb{E}[Z \mid T = 0] + e \mathbb{E}[Z \mid T = 1] \right) \\
&= (1 - e) \left( \mathbb{E}[Z \mid T = 1] - \mathbb{E}[Z \mid T = 0] \right) \\
&= (1 - e) \Delta.
\end{aligned}
\]
Now, we can compute the second term in \eqref{proof_total_covariance}
\begin{align}
  \Cov[\E[Z_i|T], \E[Z_j|T]] & = \E[(\E[Z_i|T] - \E[Z_i])(\E[Z_j|T] - \E[Z_j])] \\
  & = e (\E[Z_i|T=1] - \E[Z_i])(\E[Z_j|T=1] - \E[Z_j]) \\
  & \quad + (1-e) (\E[Z_i|T=0] - \E[Z_i])(\E[Z_j|T=0] - \E[Z_j])\\
  & = e(1-e)^2 \Delta_i \Delta_j + e^2(1-e) \Delta_i \Delta_j\\
  & = e(1-e) \Delta_i \Delta_j.
\end{align}
Consequently, according to \eqref{proof_total_covariance},
\begin{align}
    \Cov [Z_i, Z_j] = e \Cov[Z_i, Z_j |T =1] + (1-e) \Cov[Z_i, Z_j |T =0] + e(1-e) \Delta_i \Delta_j,
\end{align}
which leads to 
\begin{align}
    \Var[Z] = (1 - e) \operatorname{Var}(Z \mid T = 0) + e \operatorname{Var}(Z \mid T = 1) + e(1 - e) \Delta \Delta^\top.
\end{align}

%Now, we compute the variance of the conditional expectations:
%\[
%\begin{aligned}
%\operatorname{Var}(\mathbb{E}[Z \mid T]) &= \mathbb{E}\left[ \left( \mathbb{E}[Z \mid T] - \mathbb{E}[Z] \right) \left( \mathbb{E}[Z \mid T] - \mathbb{E}[Z] \right)^\top \right] \\
%&= (1 - e) \left( -e \Delta \right) \left( -e \Delta \right)^\top + e \left( (1 - e) \Delta \right) \left( (1 - e) \Delta \right)^\top \\
%&= (1 - e) e^2 \Delta \Delta^\top + e (1 - e)^2 \Delta \Delta^\top \\
%&= e (1 - e) \left( e + (1 - e) \right) \Delta \Delta^\top \\
%&= e (1 - e) \Delta \Delta^\top.
%\end{aligned}
%\]

%Substituting back into the Law of Total Variance,
%\[
%\begin{aligned}
%\operatorname{Var}(Z) &= \mathbb{E}[\operatorname{Var}(Z \mid T)] + \operatorname{Var}(\mathbb{E}[Z \mid T]) \\
%&= (1 - e) \operatorname{Var}(Z \mid T = 0) + e \operatorname{Var}(Z \mid T = 1) + e(1 - e) \Delta \Delta^\top\\
%\end{aligned}
%\]
Similarly the last two remaining terms we have for \(t \in \{0,1\}\):
\begin{align*}
    \left\Vert\frac{\mathbb{E}[Z]}{\mathbb{E}[Z]^{\top} \gamma_{(t)}}\right\Vert_{Q_t^{-1}}^2 &= \frac{1}{(\mathbb{E}[Z]^{\top} \gamma_{(t)})^2}\left\Vert\mathbb{E}[Z]\right\Vert_{Q_t^{-1}}^2 \\
    &= \frac{1}{(\mathbb{E}[Z]^{\top} \gamma_{(t)})^2} \mathbb{E}[Z]^\top Q_t^{-1} \mathbb{E}[Z]\\
\end{align*}
%\es{changer $\mu_t$ en $\nu_t$}
Note that we have \(\mathbb{E}[Z] = e\mathbb{E}[Z|T=1]+(1-e)\mathbb{E}[Z|T=0]\) and that for \(t \in \{0,1\}\),
\begin{align*}
    \mathbb{E}[Z | T = t]^\top Q_t^{-1} \mathbb{E}[Z | T = t] = \begin{pmatrix} 1  \\ \nu_t \end{pmatrix}^\top \begin{pmatrix} 1 + \nu_t^\top \Sigma_t^{-1}\nu_t & - \nu_t^\top \Sigma_t^{-1} \\ - \Sigma_t^{-1} \nu_t & \Sigma_t^{-1} \end{pmatrix} \begin{pmatrix} 1 \\ \nu_t \end{pmatrix} = 1,
\end{align*}
and
\begin{align*}
        \mathbb{E}[Z | T = 1-t]^\top Q_{t}^{-1} \mathbb{E}[Z | T = 1-t] &= \begin{pmatrix} 1 \\ \nu_{1-t} \end{pmatrix}^\top \begin{pmatrix} 1 + \nu_t^\top \Sigma_t^{-1}\nu_t & - \nu_t^\top \Sigma_t^{-1} \\ - \Sigma_t^{-1} \nu_t & \Sigma_t^{-1} \end{pmatrix} \begin{pmatrix} 1 \\ \nu_{1-t} \end{pmatrix} \\
        &= 1 + \|\nu_{1-t} - \nu_{t}\|^2_{\Sigma_t^{-1}},
\end{align*}
and
\begin{align*}
        \mathbb{E}[Z | T = t]^\top Q_{t}^{-1} \mathbb{E}[Z | T = 1-t] &= \begin{pmatrix} 1 \\ \nu_{t} \end{pmatrix}^\top \begin{pmatrix} 1 + \nu_t^\top \Sigma_t^{-1}\nu_t  & - \mu_t^\top \Sigma_t^{-1} \\ - \Sigma_t^{-1} \nu_t & \Sigma_t^{-1} \end{pmatrix} \begin{pmatrix} 1 \\ \nu_{1-t} \end{pmatrix} \\
        &= 1.
\end{align*}
Therefore, we have
 \begin{align*}
     \mathbb{E}[Z]^\top Q_0^{-1} \mathbb{E}[Z] &= e^2 \|\mathbb{E}[Z | T = 1]\|_{Q_0^{-1}} + (1-e)^2\|\mathbb{E}[Z | T = 0]\|_{Q_0^{-1}} + 2e(1-e)\langle \mathbb{E}[Z | T = 0], \mathbb{E}[Z | T = 1] \rangle_{Q_0^{-1}}\\
     &= e^2 \|\mathbb{E}[Z | T = 1]\|_{Q_0^{-1}} + (1-e)^2 +2e(1-e)\langle \mathbb{E}[Z | T = 0], \mathbb{E}[Z | T = 1] \rangle_{Q_0^{-1}}\\
     &= (1-e)^2 +e^2\left(1 + \|\mu_{1} - \mu_{0}\|^2_{\Sigma_0^{-1}}\right) +2e(1-e)\\
     &= 1 +e^2 \|\nu_{1} - \nu_{0}\|^2_{\Sigma_0^{-1}},
 \end{align*}
and similarly  $\mathbb{E}[Z]^\top Q_1^{-1} \mathbb{E}[Z] = 1 + (1-e)^2\|\nu_{1} - \nu_{0}\|^2_{\Sigma_1^{-1}}$. Finally, noting that  for all \(t \in \{0,1\}\)
\[\mathbb{E}[Z]^{\top} \gamma_{(t)} = \mathbb{E}\left[Y^{(t)}\right] \quad \text{and} \quad \Var\left[Z\right] = \begin{pmatrix}
0 & \cdots & 0 \\
\vdots & \Var\left[X\right]\\
0 & &
\end{pmatrix},\]
we have, letting $\Sigma = \Var\left[X\right]$
\begin{align*}
V_{\textrm{\tiny RR,G,OLS}} = \tau_{\textrm{\tiny RR}}^2 \left(\left\Vert\frac{\beta_{(1)}}{\mathbb{E}\left[Y^{(1)}\right]} - \frac{\beta_{(0)}}{\mathbb{E}\left[Y^{(0)}\right]}\right\Vert_{\Sigma}^2 + \sigma^2\left(\frac{1+(1-e)^2\|\nu_1 - \nu_0\|^2_{\Sigma_1^{-1}}}{e\mathbb{E}\left[Y^{(1)}\right]^2} + \frac{1+e^2\|\nu_1 - \nu_0\|^2_{\Sigma_0^{-1}}}{(1-e)\mathbb{E}\left[Y^{(0)}\right]^2}\right)\right).
\end{align*}
\end{proof}

\begin{lemma}[\textbf{Comparison of the asymptotic variances  of \(\hat{\tau}_{\textrm{\tiny RR,N}}\) and \(\hat{\tau}_{\textrm{\tiny RR,G}}\) under a linear model}]\label{lem-var-lin}
Grant \Cref{a:bernoulli_trial}, \Cref{a:Outcome_positivity} and \Cref{a:linear_model}. Recalling that $V_{\textrm{\tiny RR,G,OLS}}$ (resp. $V_{\textrm{\tiny RR,G,OLS}}$) is the asymptotic variance of the G-formula when oracle surface responses are used (resp. when they are estimated via OLS), we have
%, provided that either Assume that $\beta^{(0)}\neq 0$ or $\beta^{(1)}\neq 0$.
\begin{equation}\label{lin-var-n}
    V_{\textrm{\tiny RR,N}} = \tau_{\textrm{\tiny RR}}^2 \left(\frac{\left\|\beta_{(1)}\right\|_{\Sigma}^2 + \sigma^2}{e \mathbb{E}\left[Y^{(1)}\right]^2} +  \frac{\left\|\beta_{(0)}\right\|_{\Sigma}^2 + \sigma^2}{(1-e) \mathbb{E}\left[Y^{(0)}\right]^2}\right),
\end{equation}
\begin{align}\label{lin-var-g}
V_{\textrm{\tiny RR,G,OLS}} & = \tau_{\textrm{\tiny RR}}^2\left( \left\Vert\frac{\beta_{(1)}}{\mathbb{E}\left[Y^{(1)}\right]} - \frac{\beta_{(0)}}{\mathbb{E}\left[Y^{(0)}\right]}\right\Vert_{\Sigma}^2 + \sigma^2 \left(\frac{1}{e\mathbb{E}\left[Y^{(1)}\right]^2} + \frac{1}{(1-e)\mathbb{E}\left[Y^{(0)}\right]^2}\right)\right) \\
 &= V_{\textrm{\tiny RR,G}} + \tau_{\textrm{\tiny RR}}^2 \sigma^2 \left(\frac{1}{e\mathbb{E}\left[Y^{(1)}\right]^2} + \frac{1}{(1-e)\mathbb{E}\left[Y^{(0)}\right]^2}\right),
\end{align}
and
\begin{equation}\label{lin-var-dif}
V_{\textrm{\tiny RR,N}}- V_{\textrm{\tiny RR,G,OLS}} = \tau_{\textrm{\tiny RR}}^2 \left( e(1-e) \left\|   \frac{\beta_{(1)}}{e \mathbb{E}\left[Y^{(1)}\right]}-     \frac{\beta_{(0)}}{(1-e) \mathbb{E}\left[Y^{(0)}\right]}\right\|_{\Sigma}^2\right) \geq 0.
\end{equation}
\end{lemma}
\begin{proof}[Proof of \Cref{lem-var-lin}]\

\smallskip 

\textbf{First equality} 
%We first show \ref{lin-var-n}:
%Given the \hyperref[a:linear_model]{linear model}, we start by verifying outcome positivity:
%\begin{align*}
%    \mathbb{E}[Y^{(a)}|X] &=  \mathbb{E}[c(a) + X_i \beta(a) + \varepsilon_i(a)|X]\\
%    &= c(a) + X_i \beta(a)
%\end{align*}
%Since we have that  \(\|X\|_\Sigma \leq M\), then for a big enough \(c(a)\), outcome positivity is verified. 
%\ab{rajout T et notatio t et a}
%We start by computing the expectation of \( Y^{(t)} \):
%\begin{align*}
%    \mathbb{E}[Y^{(t)}] &= \mathbb{E}[c_{(t)} + X^\top \beta_{(t)} + \varepsilon_{(t)}] && \text{(by \hyperref[linear_setting]{linear model})}\\
%    &= \mathbb{E}[c_{(t)}] + \mathbb{E}[X^\top \beta_{(t)}] + \mathbb{E}[\varepsilon_{(t)}]\\
%    &= c_{(t)} + \mathbb{E}[X]\beta_{(t)} + \mathbb{E}[\mathbb{E}[\varepsilon_{(t)} | X]] && \text{(by total expectation)}\\
%    &= c_{(t)} && \text{(by \hyperref[linear_setting]{linear model})}.
%\end{align*}
The variance of \( Y^{(a)} \) satisfies
\begin{align*}
    \Var[Y^{(a)}] &= \Var[c_{(t)} + X^\top \beta_{(t)} + \varepsilon_{(t)}]\\
    &= \Var[X^\top \beta_{(t)} + \varepsilon_{(t)}] && \text{\(c_{(t)}\) is a constant}\\
    &= \Var[X^\top \beta_{(t)}] + \Var[\varepsilon_{(t)}] + 2 \Cov(X^\top \beta_{(t)}, \varepsilon_{(t)}) && \text{Bienaymé’s identity}\\
    &= ||\beta_{(t)}||_{\Sigma} + \sigma^2, &&\text{(by \hyperref[a:linear_model]{linear model})}
\end{align*}
since
\begin{align*}
    \Cov(X^\top \beta_{(t)}, \varepsilon_{(t)}) &= \mathbb{E} [X^\top \beta_{(t)} \varepsilon_{(t)}] - \mathbb{E}[X^\top \beta_{(t)}] \mathbb{E}[ \varepsilon_{(t)}]\\
    &= \mathbb{E} [ X^\top \beta_{(t)} \mathbb{E}[ \varepsilon_{(t)} | X ]] - \mathbb{E}[X^\top \beta_{(t)}] \mathbb{E}[ \mathbb{E}[\varepsilon_{(t)} | X]] && \text{(by total expectation)}\\
    &= 0, &&  \mathbb{E}[\varepsilon_{(t)} | X] = 0,
\end{align*}
and, using Eve's law, $\Var[\varepsilon_{(t)}] = \mathbb{E} [\Var [\varepsilon_{(t)} | X]] + \Var [\mathbb{E} [\varepsilon_{(t)} | X]] = \sigma^2$.
%since \( \Var[\varepsilon_{(t)} | X] = \sigma^2 \) and \( \mathbb{E}[\varepsilon_{(t)} | X] = 0 \), 
%
Thus, \( V_{\textrm{\tiny RR,N}} \) satisfies
\begin{align} V_{\textrm{\tiny RR,N}} & = \tau_{\textrm{\tiny RR}}^2 \left( \frac{\Var(Y^{(1)})}{e \mathbb{E}[Y^{(1)}]^2} + \frac{\Var(Y^{(0)})}{(1-e) \mathbb{E}[Y^{(0)}]^2} \right)\\
& = \tau_{\textrm{\tiny RR}}^2 \left( \frac{||\beta_{(1)}||_{\Sigma}^2 + \sigma^2}{e \mathbb{E}[Y^{(1)}]^2} + \frac{||\beta_{(0)}||_{\Sigma}^2 + \sigma^2}{(1-e) \mathbb{E}[Y^{(0)}]^2} \right).
\end{align}

\textbf{Second and third equality}
According to \Cref{OLS_RR_OBS} ,
\begin{align*}
\frac{V_{\textrm{\tiny RR,G,OLS}}}{\tau_{\textrm{\tiny RR}}^2} =  \left\Vert\frac{\beta_{(1)}}{\mathbb{E}\left[Y^{(1)}\right]} - \frac{\beta_{(0)}}{\mathbb{E}\left[Y^{(0)}\right]}\right\Vert_{\Sigma}^2 + \sigma^2 \left(\frac{1+(1-e)^2\|\nu_1 - \nu_0\|^2_{\Sigma_1^{-1}}}{e\mathbb{E}\left[Y^{(1)}\right]^2} + \frac{1+e^2\|\nu_1 - \nu_0\|^2_{\Sigma_0^{-1}}}{(1-e)\mathbb{E}\left[Y^{(0)}\right]^2}\right)
\end{align*}

Since we are in a RCT setting, we have that \(\nu_1 = \nu_0\)  and \(\Sigma_1 = \Sigma_0 = \Sigma\). Therefore
\begin{align*}
\frac{V_{\textrm{\tiny RR,G,OLS}}}{\tau_{\textrm{\tiny RR}}^2} =  \left\Vert\frac{\beta_{(1)}}{\mathbb{E}\left[Y^{(1)}\right]} - \frac{\beta_{(0)}}{\mathbb{E}\left[Y^{(0)}\right]}\right\Vert_{\Sigma}^2 + \sigma^2 \left(\frac{1}{e\mathbb{E}\left[Y^{(1)}\right]^2} + \frac{1}{(1-e)\mathbb{E}\left[Y^{(0)}\right]^2}\right)
\end{align*}

The first term corresponds to the Oracle variance of the G-formula. Indeed, for all $t \in \{0,1\}$, 
\begin{align*}
    \Var[\mu_{(t)}(X)] &= \Var[\mathbb{E}[Y^{(t)}|X]] \\
    &= \Var[\mathbb{E}[c_{(t)}|X] + \mathbb{E}[X^\top \beta_{(t)}|X] + \mathbb{E}[\varepsilon_{(t)}|X]] \\
    &= \Var[c_{(t)} + \mathbb{E}[X^\top \beta_{(t)}|X]] \\
    &= \Var[\mathbb{E}[X^\top \beta_{(t)}|X]] \\
    &= \Var[X^\top \beta_{(t)}] \\
    &= \|\beta_{(t)}\|_{\Sigma}^2.
\end{align*}
Besides, the covariance between \(\mu_1(X)\) and \(\mu_0(X)\) satisfies
\begin{align*}
\Cov(\mu_{(1)}(X), \mu_{(0)}(X)) =& \mathbb{E}[\mu_{(1)}(X)\mu_{(0)}(X)] -\mathbb{E}[Y^{(0)}]\mathbb{E}[Y^{(1)}]\\
=& \mathbb{E}[(c_{(1)} + X^\top \beta_{(1)})(c_{(0)} + X^\top \beta_{(0)})] -\mathbb{E}[Y^{(0)}]\mathbb{E}[Y^{(1)}]\\
=& \mathbb{E}[c_{(1)}c_{(0)}] +\mathbb{E}[c_{(1)} X^\top \beta_{(0)}] + \mathbb{E}[c_{(0)} X^\top \beta_{(1)}] \\
&+\mathbb{E}[X^\top \beta_{(0)} X^\top \beta_{(1)}] -\mathbb{E}[Y^{(0)}]\mathbb{E}[Y^{(1)}]\\
=& \mathbb{E}[X^\top \beta_{(0)} X^\top \beta_{(1)}]\\
=& \mathbb{E} \left[\sum_{j}X_j \beta_{(0), j} \sum_{k}X_k \beta_{(1), k} \right]\\
=& \sum_{j}\sum_{k}\beta_{(0), j} \beta_{(1), k} \mathbb{E}[X_kX_j]\\
=& \langle \beta_{(0)} , \beta_{(1)}\rangle_{\Sigma}.
\end{align*}
Therefore, 
\begin{align*}
V_{\textrm{\tiny RR,G}} = \tau_{\textrm{\tiny RR}}^2\Var\left(\frac{\mu_1(X)}{\mathbb{E}\left[Y^{(1)} \right]}- \frac{\mu_0(X)}{\mathbb{E}\left[Y^{(0)} \right]}\right) &= \tau_{\textrm{\tiny RR}}^2 \left(\frac{\Var(\mu_1(X))}{\mathbb{E}\left[Y^{(1)} \right]^2} + \frac{\Var(\mu_0(X))}{\mathbb{E}\left[Y^{(0)}\right]^2} - 2\frac{\Cov(\mu_0(X), \mu_1(X))}{\mathbb{E}\left[Y^{(0)}\right]\mathbb{E}\left[Y^{(1)}\right]}\right)\\
&= \tau_{\textrm{\tiny RR}}^2 \left( \left\|\frac{\beta_{(1)}}{\mathbb{E}\left[Y^{(1)}\right]} \right\|_{\Sigma}^2  + \left\| \frac{\beta_{(0)}}{\mathbb{E}\left[Y^{(0)}\right]} \right\|_{\Sigma}^2 -2\frac{\langle \beta_{(0)} , \beta_{(1)}\rangle_{\Sigma} }{\mathbb{E}\left[Y^{(0)}\right] \mathbb{E}\left[Y^{(1)}\right]}\right) \\
&= \tau_{\textrm{\tiny RR}}^2 \left\|\frac{\beta_{(1)}}{\mathbb{E}\left[Y^{(1)}\right]} - \frac{\beta_{(0)}}{\mathbb{E}\left[Y^{(0)}\right]} \right\|_{\Sigma}^2.
\end{align*}

\textbf{Last inequality}
A simple computation leads to 
\begin{align*}
    \frac{V_{\textrm{\tiny RR,N}} - V_{\textrm{\tiny RR,G}}}{\tau_{\textrm{\tiny RR}}^2} &= \frac{\left\|\beta_{(1)}\right\|_{\Sigma}^2 + \sigma^2}{e \mathbb{E}\left[Y^{(1)}\right]^2} +  \frac{\left\|\beta_{(0)}\right\|_{\Sigma}^2 + \sigma^2}{(1-e) \mathbb{E}\left[Y^{(0)}\right]^2} \\
    & \quad - \left( \left\Vert\frac{\beta_{(1)}}{\mathbb{E}\left[Y^{(1)}\right]} - \frac{\beta_{(0)}}{\mathbb{E}\left[Y^{(0)}\right]}\right\Vert_{\Sigma}^2 + \sigma^2 \left(\frac{1}{e\mathbb{E}\left[Y^{(1)}\right]^2} + \frac{1}{(1-e)\mathbb{E}\left[Y^{(0)}\right]^2}\right)\right)\\
    &= \left(\frac{1-e}{e}\right) \frac{\left\|\beta_{(1)}\right\|_{\Sigma}^2}{ \mathbb{E}\left[Y^{(1)}\right]^2}  +  \left(\frac{e}{1-e}\right) \frac{\left\|\beta_{(0)}\right\|_{\Sigma}^2}{ \mathbb{E}\left[Y^{(0)}\right]^2} + \frac{2 \langle \beta_{(1)}, \beta_{(0)} \rangle_{\Sigma}}{\mathbb{E}\left[Y^{(1)}\right] \mathbb{E}\left[Y^{(0)}\right]}\\
    & = e(1-e) \left\|   \frac{\beta_{(1)}}{e \mathbb{E}\left[Y^{(1)}\right]}-   \frac{\beta_{(0)}}{(1-e) \mathbb{E}\left[Y^{(0)}\right]}\right\|_{\Sigma}^2.
\end{align*}

\end{proof}

\subsubsection{Risk Ratio one-step estimator}

\begin{proof}[Proof of \Cref{prop:one_step}]\label{proof:one_step}
 We will use \cite{kennedy2022semiparametric, Kennedy} notation in this proof. If you are not familiar on how to compute an influence function, note that it is very similar to compute the derivative of a function. We define our estimand quantity
    \[\psi = \frac{\mathbb{E}\left[\mathbb{E}\left[Y| T= 1, X\right]\right]}{\mathbb{E}\left[\mathbb{E}\left[Y| T= 0, X\right]\right]} = \frac{\psi_1}{\psi_0}.\]
    We can now compute the influence function \(\varphi\) of \(\psi\).
    \begin{align*}
        \varphi = \mathbb{I}\mathbb{F}\left(\psi\right) = \mathbb{I}\mathbb{F}\left(\frac{\psi_1}{\psi_0}\right) &= \frac{\mathbb{I}\mathbb{F}\left(\psi_1\right)\psi_0 - \mathbb{I}\mathbb{F}\left(\psi_0\right)\psi_1}{\psi_0^2}\\
        &= \frac{\mathbb{I}\mathbb{F}\left(\psi_1\right)}{\psi_0} - \psi \frac{\mathbb{I}\mathbb{F}\left(\psi_0\right)}{\psi_0}.
    \end{align*}
    According to Example 2 in \cite{kennedy2022semiparametric}, we have
    \begin{align*}
        & \mathbb{I}\mathbb{F}\left(\psi_1\right) = \mu_1(X) + T\frac{Y-\mu_1(X)}{e(X)} - \psi_1 \\
        \textrm{and} \quad & \mathbb{I}\mathbb{F}\left(\psi_0\right) = \mu_0(X) + (1-T)\frac{Y-\mu_0(X)}{1-e(X)} - \psi_0.
    \end{align*}
    Therefore, 
    \begin{align*}
        \varphi &= \frac{\mathbb{I}\mathbb{F}\left(\psi_1\right)}{\psi_0} - \psi \frac{\mathbb{I}\mathbb{F}\left(\psi_0\right)}{\psi_0}\\
        &= \frac{\mu_1(X) + T\frac{Y-\mu_1(X)}{e(X)} - \psi_1}{\psi_0} - \psi \frac{\mu_0(X) + (1-T)\frac{Y-\mu_0(X)}{1-e(X)} - \psi_0}{\psi_0}\\
        &= \frac{\mu_1(X) + T\frac{Y-\mu_1(X)}{e(X)} }{\psi_0} - \psi  - \psi \left(\frac{\mu_0(X) + (1-T)\frac{Y-\mu_0(X)}{1-e(X)}}{\psi_0} -1 \right)\\
        &= \frac{\mu_1(X) + T\frac{Y-\mu_1(X)}{e(X)} }{\psi_0} -\psi \frac{\mu_0(X) + (1-T)\frac{Y-\mu_0(X)}{1-e(X)}}{\psi_0}.
    \end{align*}
    As referenced in \cite{kennedy2022semiparametric} regarding the semiparametric von Mises expansion, consider the functional \(\psi : P \to \mathbb{R}\), where \(P\) represents the true data distribution and \(\hat P\) its estimation. The expansion is formulated as:
    \begin{align}
    \psi(\hat P) - \psi(P) = \int \varphi(z; \hat P) d(\hat P - P)(z) + R_2(\hat P, P),\label{eq_proof_RR_OS_VM_expansion}
    \end{align}
    for all distributions \(\hat{P}\) and \(P\). The influence function \(\varphi(z; P)\), associated with \(\psi\), is a function with zero mean and finite variance as defined by \cite{Tsiatis}
    \begin{align}
    \int \varphi(z; P) dP(z) = 0 \quad \text{and} \quad \int \varphi(z; P)^2 dP(z) < \infty, \label{eq_caracteristique_influence_function}
    \end{align}
    and \( R_2(\hat{P}, P) \) denotes a second-order remainder term.
    %, implying its reliance on higher degree differences between \(\hat{P}\) and \(P\). 
    According to the expansion in \eqref{eq_proof_RR_OS_VM_expansion}, most plug-in estimators \( \psi(\hat{P})\) are biased to the first order, evidenced by:
    \[
    \psi(P) = \psi(\hat{P}) + \int \varphi(z; \hat{P}) d{P}(z) + R_2(\hat{P}, P),
    \]
    since $\int \varphi(z; \hat{P}) d\hat{P}(z) = 0$.
    Therefore, a first-order approximation of  \(\psi(P)\) is given by $\psi(\hat{P}) + \int \varphi(z; \hat{P}) d{P}(z)$ which can be estimated via
    \begin{align*}
        \hat{\tau}_{\textrm{\tiny RR-OS}} & = 
        %\psi(\hat{P}) + \int \varphi(z; \hat{P}) d{P}(z) 
         \hat \psi + \frac{1}{n}\sum_{i=1}^{n}\varphi(Z_i) \\
        &= \hat \psi + \frac{1}{n}\sum_{i=1}^{n}\frac{\mu_1(X_i) + T_i\frac{Y_i-\mu_1(X_i)}{e(X_i)} }{\hat \psi_0} - \hat \psi \frac{\mu_0(X_i) + (1-T_i)\frac{Y_i-\mu_0(X_i)}{1-e(X_i)}}{\hat \psi_0} \\
        &= \hat \psi \left(1-\frac{\frac{1}{n}\sum_{i=1}^{n}\mu_0(X_i) + (1-T_i)\frac{Y_i-\mu_0(X_i)}{1-e(X_i)}}{\hat \psi_0}\right) + \frac{\frac{1}{n}\sum_{i=1}^{n}\mu_1(X_i) + T_i\frac{Y_i-\mu_1(X_i)}{e(X_i)} }{\hat \psi_0}\\
         &= \frac{\sum_{i=1}^n \hat \mu_1(X_i)}{\sum_{i=1}^n \hat \mu_0(X_i)} \left(1- \frac{\sum_{i=1}^n \hat\mu_0(X_i) + \frac{(1-T_i)(Y_i- \hat\mu_0(X_i))}{1-\hat e(X_i)}}{\sum_{i=1}^n \hat \mu_0(X_i)}\right) + \frac{\sum_{i=1}^n \hat\mu_1(X_i) +  \frac{T_i(Y_i- \hat\mu_1(X_i))}{\hat e(X_i)}}{\sum_{i=1}^n \hat \mu_0(X_i)}.
    \end{align*}
\end{proof}

\begin{proof}[Proof of \Cref{prop:one_step2}]\label{proof:one_step2}
\ \\
\textbf{Asymptotic bias and variance of the cross-fitted One-step estimator}
%Let us begin by examining the expression \(\sqrt{n}\left(\hat{\tau}_{\textrm{\tiny RR-OS}} - \tau_{\textrm{\tiny RR}} \right)\). 
Recall that 
\begin{align}
     \psi(P) & = \frac{ \mathbb{E}_P \left[ \mathbb{E}_P[Y \mid X, T = 1] \right]}{\mathbb{E}_P \left[ \mathbb{E}_P[Y \mid X, T = 0] \right]} = \frac{\psi_1}{\psi_0}\\
     \psi(\hat{P}) &=  \frac{\sum_{i=1}^n \hat \mu_1(X_i)}{\sum_{i=1}^n \hat \mu_0(X_i)} = \frac{\hat{\psi}_1}{\hat{\psi}_0}\\
     \varphi(Z;  \hat{P}) & = \frac{\hat \mu_1(X_i) + T_i\frac{Y_i-\hat \mu_1(X_i)}{\hat e(X_i)} }{\hat \psi_0} - \hat \psi \frac{\hat \mu_0(X_i) + (1-T_i)\frac{Y_i-\hat \mu_0(X_i)}{1-\hat e(X_i)}}{\hat \psi_0} 
\end{align}
where $P$ represents the true underlying data distribution and $\hat{P}$ the distribution where oracle quantities have been replaced by plug-in estimates. We express \(\psi(P)\) as follows:
\begin{align*}
    \psi(P) &= \psi(\hat{P}) + \int \varphi(z; \hat{P}) dP(z) +   R_2(\hat{P}, P),
\end{align*}
where \(R_2\) encapsulates higher order remainder terms.

To elucidate, we rearrange to find \(\psi(\hat{P}) - \psi(P)\):
\begin{align*}
    \psi(\hat{P}) - \psi(P) &= R_2(P,\hat{P}) - \int \varphi(z; \hat{P}) dP(z)\\
    &= \frac{1}{n} \sum_{i=1}^n \varphi(Z_i; P) - \frac{1}{n} \sum_{i=1}^n \varphi(Z_i; \hat{P}) \\
    & \quad + \frac{1}{n} \sum_{i=1}^n \left(\varphi(Z_i; \hat{P}) - \varphi(Z_i; P)\right) - \int \left(\varphi(z; \hat{P})- \varphi(z; P) \right)dP(z) \\
    & \quad +   R_2(P,\hat{P}).
\end{align*}
Recalling that $\hat{\tau}_{\textrm{\tiny RR-OS}} = \psi(\hat{P}) + \frac{1}{n} \sum_{i=1}^n \varphi(Z_i; \hat{P})$ and $\tau_{\textrm{\tiny RR}} = \psi(P)$, we have
%Expanding on the difference \(\hat{\tau}_{\textrm{\tiny RR-OS}} - \tau_{\textrm{\tiny RR}}\), we find:
\begin{align}
    \hat{\tau}_{\textrm{\tiny RR-OS}} - \tau_{\textrm{\tiny RR}}  &= \psi(\hat{P}) + \frac{1}{n} \sum_{i=1}^n \varphi(Z_i; \hat{P}) - \psi(P)\\
    &= \frac{1}{n} \sum_{i=1}^n \varphi(Z_i; P) \\
    & \quad + \frac{1}{n} \sum_{i=1}^n \left(\varphi(Z_i; \hat{P}) - \varphi(Z_i; P)\right) - \int \left(\varphi(z; \hat{P})- \varphi(z; P) \right)dP(z) \\
    & \quad +   R_2(P,\hat{P}). \label{eq_proof_OS-decomposition1}
\end{align}

The first term is a sample average of centered i.i.d. terms since, by definition \eqref{eq_caracteristique_influence_function},   \(\int \varphi(z; P) dP(z) = 0 \). According to the central limit theorem, it converges to a normally distributed random variable with variance \(\Var(\varphi(Z))/n\).

Following the work of \cite{Vaart_1998}, we consider the second term in \eqref{eq_proof_OS-decomposition1}, that is
\begin{align*}
    \frac{1}{n} \sum_{i=1}^n \left(\varphi(Z_i; \hat{P}) - \varphi(Z_i; P)\right) - \int \left(\varphi(z; \hat{P})- \varphi(z; P) \right)dP(z).
    %= (P_n - P) \{ \varphi(Z; \hat{P}) - \varphi(Z; P) \},
\end{align*}
%where \(P_n\) denotes the empirical measure.
Since our estimator is built on a cross-fitting strategy with $K$ folds \(\mathcal{I}_1, \hdots \mathcal{I}_K\), containing respectively $n_1, \hdots, n_K$ observations, the above quantity may be written as 
\begin{align*}
   \frac{1}{n} \sum_{k=1}^K \sum_{i \in \mathcal{I}_k}
      \left(\varphi(Z_i; \hat{P}^{-k}) - \varphi(Z_i; P)\right) - \int \left(\varphi(z; \hat{P})- \varphi(z; P) \right)dP(z),
    %= (P_n - P) \{ \varphi(Z; \hat{P}) - \varphi(Z; P) \},
\end{align*}
where $\hat{P}^{-k}$ corresponds to a data distribution where oracle quantity are replaced by plug-in estimates built on all observations except those in $\mathcal{I}_k$. We denote this set of observations as $\mathcal{I}_{-k}$. We let $\hat \varphi^{-k}(Z) = \varphi(Z; \hat{P}^{-k})$ and
\begin{align}
U_k = \left( \mathbb{P}_n^{(k)} - P \right) \left( \hat \varphi^{-k}(Z) - \varphi(Z) \right),    
\end{align}
where \(\mathbb{P}_n^{(k)}\) is the empirical measure over \(\mathcal{I}_k\). The quantity of interest can thus be written as 
\begin{align}
    \frac{1}{n} \sum_{k=1}^K \sum_{i \in \mathcal{I}_k}
      \left(\varphi^{-k}(Z_i) - \varphi(Z_i; P)\right) - \int \left(\varphi(z; \hat{P})- \varphi(z; P) \right)dP(z)
      & = \frac{1}{n} \sum_{k=1}^K n_k U_k. \label{eq_proof_decomp_OS_proba_empirique}
\end{align}
%We'll start by computing the mean and variance of this random variable conditioned on having observed a fixed dataset \(\mathcal{I}_{-k}\) comprising of all the folds except the \(k\)th. By 
The expectation and variance of $U_k$ satisfy
\begin{align}
    \mathbb{E} \left[ U_k \mid \mathcal{I}_{-k} \right] &= \mathbb{E} \left[ \left( \mathbb{P}_n^{(k)} - P \right) (\hat \varphi^{-k} - \varphi) \mid \mathcal{I}_{-k} \right] \\
    &= \mathbb{E} \left[ \mathbb{P}_n^{(k)} (\hat \varphi^{-k} - \varphi) \mid \mathcal{I}_{-k} \right] - \mathbb{E} \left[ P (\hat \varphi^{-k} - \varphi) \mid \mathcal{I}_{-k} \right] \\
    &= \mathbb{E} \left[\hat \varphi^{-k}(Z) - \varphi(Z) \right] - \mathbb{E} \left[\hat \varphi^{-k}(Z) - \varphi(Z) \right] \\
    &= 0,
\end{align}
and
\begin{align}
    \text{Var} \left[ U_k \mid \mathcal{I}_{-k} \right] &= \text{Var} \left[ \left( \mathbb{P}_n^{(k)} - P \right) (\hat \varphi^{-k} - \varphi) \mid \mathcal{I}_{-k} \right] \\
    &= \text{Var} \left[ \mathbb{P}_n^{(k)} (\hat \varphi^{-k} - \varphi) - P (\hat\varphi^{-k} - \varphi)  \mid \mathcal{I}_{-k} \right] \\
    &= \text{Var} \left[ \frac{1}{n_k} \sum_{i=1}^{n_k} \left(\hat \varphi^{-k}(Z_i) - \varphi(Z_i) \right)  \mid \mathcal{I}_{-k} \right] \\
    &= \frac{1}{n_k} \text{Var} \left[\hat \varphi^{-k}(Z) - \varphi(Z)   \mid \mathcal{I}_{-k} \right] \\
    &\leq \frac{1}{n_k} \mathbb{E}\left[ (\hat \varphi^{-k}(Z) - \varphi(Z) )^2  \mid \mathcal{I}_{-k} \right].
    %\left\|\hat \varphi^{-k} - \varphi \right\|^2
\end{align}
Let $a>0$. Applying Chebyshev's inequality leads to 
%, which generally quantifies the probability that a random variable takes a (de-meaned) value that's larger than its standard deviation:
\begin{align}
& \mathbb{P} \left( \frac{\left| U_k - \mathbb{E}[U_k \mid \mathcal{I}_{-k}] \right|}{\sqrt{\text{Var}[U_k \mid \mathcal{I}_{-k}]}} \geq a  \mid \mathcal{I}_{-k} \right) \leq \frac{1}{a^2}\\
 \Longleftrightarrow &  \mathbb{P} \left( \frac{\left| U_k  \right|}{\sqrt{\text{Var}[U_k \mid \mathcal{I}_{-k}]}} \geq a  \mid \mathcal{I}_{-k} \right) \leq \frac{1}{a^2}.
\end{align}
Thus, 
\begin{align}
    \mathbb{P} \left(  \frac{\left| U_k  \right| \sqrt{n_k}}{\sqrt{\mathbb{E}\left[ (\hat \varphi^{-k}(Z) - \varphi(Z) )^2  \mid \mathcal{I}_{-k} \right]}} \geq a  \mid \mathcal{I}_{-k} \right) \leq \mathbb{P} \left( \frac{\left| U_k  \right|}{\sqrt{\text{Var}[U_k \mid \mathcal{I}_{-k}]}} \geq a  \mid \mathcal{I}_{-k} \right) \leq \frac{1}{a^2}, 
\end{align}
which leads to 
\begin{align}
    \mathbb{P} \left( \left| U_k  \right| \sqrt{n_k} \geq a  \mid \mathcal{I}_{-k} \right) \leq  \frac{\mathbb{E}\left[ (\hat \varphi^{-k}(Z) - \varphi(Z) )^2  \mid \mathcal{I}_{-k} \right]}{a^2}.
\end{align}
Finally, taking the expectation on both sides leads to 
\begin{align}
    \mathbb{P} \left( \left| U_k  \right| \sqrt{n_k} \geq a \right)  \leq  \frac{\mathbb{E}\left[ (\hat \varphi^{-k}(Z) - \varphi(Z) )^2  \right]}{a^2}.
\end{align}
%\begin{align}
%    \mathbb{P} \left(  \frac{\left| U_k  \right| \sqrt{n_k}}{\sqrt{\mathbb{E}\left[ (\hat \varphi^{-k}(Z) - \varphi(Z) )^2  \mid \mathcal{I}_{-k} \right]}} \geq a  \mid \mathcal{I}_{-k} \right)  \leq \frac{1}{a^2}.
%\end{align}
%Instead of using \(U\), we use \(U \mid \mathcal{I}_{-k}\), which gives
%\[
%\mathbb{P} \left( \sqrt{\frac{1}{n_k}} \left| \left( \mathbb{P}_n^{(k)} - P \right)(\hat \varphi - \varphi) \right| \geq a \mid \mathcal{I}_{-k} \right) \leq \frac{1}{a^2}
%\]
%Now we can take the expectation on both sides (averaging over \(\mathcal{I}_{-k}\)) to get:
%%
%\[
%\mathbb{P} \left( \sqrt{\frac{1}{n_k}} \left| \left( \mathbb{P}_n^{(k)} - P \right)(\hat \varphi - \varphi) \right| \geq a \right) \leq \frac{1}{a^2}
%\]
According to \eqref{eq_proof_decomp_OS_proba_empirique}, the quantity of interest takes the form 
\begin{align}
    \frac{1}{n} \sum_{k=1}^K n_k U_k =  \sum_{k=1}^K \frac{ n_k}{n} U_k.
\end{align}
Hence, 
\begin{align}
    \mathbb{P} \left( \sqrt{n} \frac{ n_k}{n} | U_k |  \leq a \frac{ \sqrt{n_k}}{\sqrt{n}}\right)  \geq 1 - \frac{\mathbb{E}\left[ (\hat \varphi^{-k}(Z) - \varphi(Z) )^2  \right]}{a^2}.
\end{align}
Therefore,
\begin{align}
    & \mathbb{P} \left( \sqrt{n} \sum_{k=1}^K  \frac{ n_k}{n} | U_k |  \leq a \sum_{k=1}^K \frac{ \sqrt{n_k}}{\sqrt{n}}\right)  \geq 1 - \sum_{k=1}^K \frac{\mathbb{E}\left[ (\hat \varphi^{-k}(Z) - \varphi(Z) )^2  \right]}{a^2}\\
    \Rightarrow & \quad  \mathbb{P} \left( \sqrt{n} \sum_{k=1}^K  \frac{ n_k}{n} | U_k |  \leq a K \right)  \geq 1 - \sum_{k=1}^K \frac{\mathbb{E}\left[ (\hat \varphi^{-k}(Z) - \varphi(Z) )^2  \right]}{a^2},
\end{align}
which proves that $\sum_{k=1}^K \frac{ n_k}{n} U_k = o_P(1/\sqrt{n})$ as $K$ is fixed and $\varphi^{-k}$ is $ L^2$ consistent. 
%\begin{align}
%\mathbb{P} \left( \sum_{k=1}^K \sqrt{n} \frac{ n_k}{n} | U_k |  \geq a \frac{ \sqrt{n_k}}{\sqrt{n}}\right) \geq \mathbb{P}\left[ ]
%    \mathbb{P} \left( \sum_{k=1}^K \sqrt{n} \frac{ n_k}{n} | U_k |  \geq a \frac{ \sqrt{n_k}}{\sqrt{n}}\right)  \leq  \frac{\mathbb{E}\left[ (\hat \varphi^{-k}(Z) - \varphi(Z) )^2  \right]}{a^2},
%\end{align}
%Which by the definition of stochastic boundedness shows that:
%\begin{align*}
%    \left( \mathbb{P}_n^{(k)} - P \right)(\hat \varphi - \varphi) &= O_P \left( \left\| \hat \varphi - \varphi \right\| / \sqrt{n_k} \right)\\
%    &= O_P \left( \left\| \hat \varphi - \varphi \right\| / \sqrt{n} \right) \\
%    &= \frac{\left\| \hat \varphi - \varphi \right\| }{\sqrt{n}}O_P(1)\\
%    &= o_P(1/\sqrt{n}) && \text{Using \(\mathcal{L}^2\) Consistency of \(\hat \varphi\)}
%\end{align*}

Regarding the last term, note that 
\begin{align}
     R_2(P,\hat{P}) &= \psi(\hat{P}) - \psi(P) + \int \varphi(z; \hat{P}) d{P}(z) \\
     &=  \psi(\hat{P}) - \psi(P) + \mathbb{E}[\varphi(Z; \hat{P})] \\
     &= \psi(\hat{P}) - \psi(P) + \mathbb{E}\left[\frac{\hat \mu_1(X) + T\frac{Y-\hat \mu_1(X)}{\hat e(X)} }{\hat \psi_0} -\hat \psi \frac{\hat \mu_0(X) + (1-T)\frac{Y-\hat \mu_0(X)}{1-\hat e(X)}}{\hat \psi_0}\right] \\
     &= \psi(\hat{P}) - \psi(P) + \frac{\mathbb{E}\left[\hat \mu_1(X) + T\frac{Y-\hat \mu_1(X)}{\hat e(X)}\right]}{\hat \psi_0} -\hat \psi \frac{\mathbb{E}\left[\hat \mu_0(X) + (1-T)\frac{Y-\hat \mu_0(X)}{1-\hat e(X)}\right]}{\hat \psi_0} \\
     &= \psi(\hat{P}) - \psi(P) + \frac{\mathbb{E}\left[\hat \mu_1(X) - \mu_1(X)+ T\frac{Y-\hat \mu_1(X)}{\hat e(X)}\right]}{\hat \psi_0}+ \frac{\psi_1}{\hat \psi_0} \\
     & \quad -\hat \psi \frac{\mathbb{E}\left[\hat \mu_0(X) - \mu_0(X)+ (1-T)\frac{Y-\hat \mu_0(X)}{1-\hat e(X)}\right]}{\hat \psi_0} - \hat \psi\frac{\psi_0}{\hat \psi_0}. \label{eq_proof_decomp_R2_AIPW}
\end{align}
Note that
\begin{align}
    \mathbb{E}\left[\hat \mu_1(X) - \mu_1(X)+ T\frac{Y-\hat \mu_1(X)}{\hat e(X)}\right] &= \mathbb{E}\left[\frac{1}{\hat e(X)}( \mu_1(X) - \hat\mu_1(X))(\hat e(X) - e(X))\right] \\
    \text{Positivity}\qquad &\leq \frac{1}{\eta} \mathbb{E}\left[(\mu_1(X) - \hat\mu_1(X))(\hat e(X) - e(X))\right]\\
    \text{Cauchy-Schwarz}\qquad &\leq \frac{1}{\eta} \mathbb{E}\left[\left(\hat e(X) - e(X)\right)^2\right]^{1/2}\mathbb{E}\left[\left(\hat \mu_1(X) - \mu_1(X)\right)^2\right]^{1/2}\\
    &= o_p\left(\frac{1}{\sqrt{n}}\right).
\end{align}
Similarly, 
\begin{align}
 \mathbb{E}\left[\hat \mu_1(X) - \mu_1(X)+ T\frac{Y-\hat \mu_1(X)}{\hat e(X)}\right] = o_p\left(\frac{1}{\sqrt{n}}\right).   
\end{align}
For the last term in \eqref{eq_proof_decomp_R2_AIPW}, since $ \psi = \psi_1/\psi_0$ and $ \hat \psi = \hat \psi_1/\hat \psi_0$, 
\begin{align*}
 \hat \psi - \psi + \frac{\psi_1}{\hat \psi_0} - \hat \psi\frac{\psi_0}{\hat \psi_0} &= \psi_1 \left(\frac{1}{\hat \psi_0} -  \frac{1}{\psi_0}\right) + \hat \psi \left(1-\frac{\psi_0}{\hat \psi_0}\right)\\
 &= \psi_1 \frac{\psi_0-\hat \psi_0}{\psi_0 \hat \psi_0}  + \hat \psi \left(\frac{\hat \psi_0-\psi_0}{\hat \psi_0}\right)\\
 &= \left(\frac{\hat \psi_0-\psi_0}{\hat \psi_0}\right)\left(\hat \psi - \psi\right)\\
 &= \left(\frac{\hat \psi_0-\psi_0}{\hat \psi_0}\right) \left( \left(\frac{1}{\hat \psi_0} -  \frac{1}{\psi_0}\right) \hat \psi_1 + \frac{1}{\psi_0}\left(\hat \psi_1 - \psi_1\right)\right)\\
 &= \frac{1}{\psi_0 \hat \psi_0}\left((\hat \psi_0-\psi_0)(\hat \psi_1-\psi_1)-\hat \psi (\hat \psi_0-\psi_0)(\hat \psi_0-\psi_0)\right).
\end{align*}
By assumption,  we have 
\begin{align}
    (\hat \psi_0-\psi_0)(\hat \psi_1-\psi_1) & = \mathbb{E}\left[\hat \mu_0(X) - \mu_0(X)\right]\mathbb{E}\left[\hat \mu_1(X) - \mu_1(X)\right] \\
    & \leq (\mathbb{E}\left[(\hat \mu_0(X) - \mu_0(X))^2\right])^{1/2} (\mathbb{E}\left[(\hat \mu_1(X) - \mu_1(X))^2\right])^{1/2} \\
    & = o_p\left(\frac{1}{\sqrt{n}}\right)
\end{align}
and 
\begin{align*}
    (\hat \psi_0-\psi_0)(\hat \psi_0-\psi_0) & = (\mathbb{E}\left[\hat \mu_0(X) \right] - \mathbb{E}\left[\mu_0(X)\right])^2 \\
    & = o_p\left(\frac{1}{\sqrt{n}}\right).
\end{align*}
By assumption, $\mathbb{E}\left[(\hat \mu_0(X) - \mu_0(X))^2\right]$ tends to zero. Thus, $\hat{\psi}_0 = \mathbb{E}[\hat{\mu}_0]$ tends to $\psi_0 = \mathbb{E}[\mu_0(X)]$. Thus, 
\begin{align*}
     \hat \psi - \psi + \frac{\psi_1}{\hat \psi_0} - \hat \psi\frac{\psi_0}{\hat \psi_0} = o_p\left(\frac{1}{\sqrt{n}}\right),
\end{align*}
which implies that $R_2(P,\hat{P}) = o_p\left(n^{-1/2}\right)$.
 Finally,
 \begin{align*}
     \sqrt{n}\left(\hat{\tau}_{\textrm{\tiny RR-OS}} - \tau_{\textrm{\tiny RR}} \right) = \frac{1}{\sqrt{n}} \sum_{i=1}^n \varphi(Z_i; P) + o_p\left(\frac{1}{\sqrt{n}}\right)
 \end{align*}
and thus
 \begin{align*}
 \sqrt{n}\left(\hat{\tau}_{\textrm{\tiny RR-OS}} - \tau_{\textrm{\tiny RR}} \right) \stackrel{d}{\rightarrow} \mathcal{N}\left(0, \Var(\varphi) \right),
  \end{align*}
  where
  \begin{align}
      \Var(\varphi) &= \Var\left(\frac{\mu_1(X) + T\frac{Y-\mu_1(X)}{e(X)} }{\psi_0} -\psi \frac{\mu_0(X) + (1-T)\frac{Y-\mu_0(X)}{1-e(X)}}{\psi_0}\right) \\
      &=\psi^2 \Var\left(\frac{\mu_1(X) + T\frac{Y-\mu_1(X)}{e(X)} }{\psi_1}  \frac{\mu_0(X) + (1-T)\frac{Y-\mu_0(X)}{1-e(X)}}{\psi_0}\right)\\
      &= \tau_{\textrm{\tiny RR}}^2\Var\left(\frac{g_1(Z)}{\mathbb{E}\left[Y^{(1)}\right]} - \frac{g_0(Z)}{\mathbb{E}\left[Y^{(0)}\right]}\right).
  \end{align}
Using Bienaymé's identity, we get
\begin{align}
    \Var\left(\frac{g_1(Z)}{\mathbb{E}\left[Y^{(1)}\right]} - \frac{g_0(Z)}{\mathbb{E}\left[Y^{(0)}\right]}\right)
    &= \Var\left(\frac{\mu_1(X)}{\mathbb{E}\left[Y^{(1)}\right]} - \frac{\mu_0(X)}{\mathbb{E}\left[Y^{(0)}\right]}\right) +  \Var\left(\frac{T(Y-\mu_1(X))}{\mathbb{E}\left[Y^{(1)}\right]e(X)} - \frac{(1-T)(Y-\mu_0(X))}{\mathbb{E}\left[Y^{(0)}\right](1-e(X))}\right)\\
    &+ 2 \Cov\left(\frac{\mu_1(X)}{\mathbb{E}\left[Y^{(1)}\right]} - \frac{\mu_0(X)}{\mathbb{E}\left[Y^{(0)}\right]}; \frac{T(Y-\mu_1(X))}{\mathbb{E}\left[Y^{(1)}\right]e(X)} - \frac{(1-T)(Y-\mu_0(X))}{\mathbb{E}\left[Y^{(0)}\right](1-e(X))}\right). \label{eq_proof_39}
\end{align}
The second term can be rewritten as 
\begin{align}
& \Var\left(\frac{T(Y-\mu_1(X))}{\mathbb{E}\left[Y^{(1)}\right]e(X)} - \frac{(1-T)(Y-\mu_0(X))}{\mathbb{E}\left[Y^{(0)}\right](1-e(X))}\right) \\
    &= \Var\left(\frac{T(Y-\mu_1(X))}{\mathbb{E}\left[Y^{(1)}\right]e(X)}\right) + \Var\left(\frac{(1-T)(Y-\mu_0(X))}{\mathbb{E}\left[Y^{(0)}\right](1-e(X))}\right) - 2\Cov\left(\frac{T(Y-\mu_1(X))}{\mathbb{E}\left[Y^{(1)}\right]e(X)},\frac{(1-T)(Y-\mu_0(X))}{\mathbb{E}\left[Y^{(0)}\right](1-e(X))} \right), \label{eq_proof_38}
\end{align}
with
\begin{align}
    \Var\left(\frac{T(Y-\mu_1(X))}{\mathbb{E}\left[Y^{(1)}\right]e(X)}\right) = \mathbb{E}\left[\left(\frac{T(Y-\mu_1(X))}{\mathbb{E}\left[Y^{(1)}\right]e(X)}\right)^2\right] - \mathbb{E}\left[\frac{T(Y-\mu_1(X))}{\mathbb{E}\left[Y^{(1)}\right]e(X)}\right]^2. \label{eq_ref_proof_37}
\end{align}
For the first term in \eqref{eq_ref_proof_37}, 
\begin{align*}
    & \mathbb{E}\left[ \left(T\frac{Y-\mu_1(X)}{e(X)\mathbb{E}\left[Y^{(1)}\right]} \right)^2 \right] \\
    &=   \mathbb{E}\left[ \left(T\frac{Y^{(1)}-\mu_1(X)}{e(X)\mathbb{E}\left[Y^{(1)}\right]} \right)^2 \right] && \text{Consistency} \\
    &=\mathbb{E}\left[ \mathbb{E}\left[\left(T\frac{Y^{(1)}-\mu_1(X)}{e(X)\mathbb{E}\left[Y^{(1)}\right]} \right)^2 \mid X \right]\right] && \text{Total expectation}\\
    &= \mathbb{E}\left[ \mathbb{E}\left[T\left(\frac{Y^{(1)}-\mu_1(X)}{e(X)\mathbb{E}\left[Y^{(1)}\right]} \right)^2 \mid X \right]\right] && \text{T is binary}\\
    &= \mathbb{E}\left[ \mathbb{E}\left[\mathbf{1}_{\left\{T=1\right\}}\left(\frac{Y^{(1)}-\mu_1(X)}{e(X)\mathbb{E}\left[Y^{(1)}\right]} \right)^2 \mid X \right]\right] && \text{T written as an indicator} \\
     &= \mathbb{E}\left[ \frac{1}{e(X)^2\mathbb{E}\left[Y^{(1)}\right]^2} \mathbb{E}\left[\mathbf{1}_{\left\{T=1\right\}}\left(Y^{(1)}-\mu_1(X)\right)^2 \mid X \right]\right] && \text{$e(X)$ is a function of $X$} \\
      &= \mathbb{E}\left[ \frac{\Var\left(Y^{(1)}|X\right) }{e(X)^2\mathbb{E}\left[Y^{(1)}\right]^2} \mathbb{E}\left[\mathbf{1}_{\left\{T=1\right\}}\mid X \right]\right] && \text{Uncounf. \& $\mu_1(.)$ is func. of $X$} \\
      &= \mathbb{E}\left[ \frac{\Var\left(Y^{(1)}|X\right) }{e(X)^2\mathbb{E}\left[Y^{(1)}\right]^2} e(X)\right] && \text{Definition of $e(X)$} \\
      &= \mathbb{E}\left[ \frac{\Var\left(Y^{(1)}|X\right) }{e(X)\mathbb{E}\left[Y^{(1)}\right]^2} \right].
\end{align*}
For the second term in \eqref{eq_ref_proof_37}, 
\begin{align*}
& \mathbb{E}\left[\frac{T(Y-\mu_1(X))}{\mathbb{E}\left[Y^{(1)}\right]e(X)}\right] \\
&= \mathbb{E}\left[\frac{T(Y^{(1)}-\mu_1(X))}{\mathbb{E}\left[Y^{(1)}\right]e(X)}\right] && \text{Consistency} \\
    &=\mathbb{E}\left[ \mathbb{E}\left[T\frac{Y^{(1)}-\mu_1(X)}{e(X)\mathbb{E}\left[Y^{(1)}\right]}  \mid X \right]\right] && \text{Total expectation}\\
    &=\mathbb{E}\left[ \frac{1}{e(X)\mathbb{E}\left[Y^{(1)}\right]}\mathbb{E}\left[T(Y^{(1)}-\mu_1(X))  \mid X \right]\right] && \text{$e(X)$ is a function of $X$} \\
    &=  \mathbb{E}\left[ \frac{e(X)}{e(X)\mathbb{E}\left[Y^{(1)}\right]}(\mu_1(X)- \mu_1(X))\right] && \text{Uncounf. \& $\mu_1(.)$ is func. of $X$} \\
    &= 0.
\end{align*}
Therefore
\begin{align*}
    \Var\left(\frac{T(Y-\mu_1(X))}{\mathbb{E}\left[Y^{(1)}\right]e(X)}\right) = \mathbb{E}\left[ \frac{\Var\left(Y^{(1)}|X\right)}{e(X)\mathbb{E}\left[Y^{(1)}\right]^2} \right],
\end{align*}
and similarly
\begin{align*}
    \Var\left(\frac{(1-T)(Y-\mu_0(X))}{\mathbb{E}\left[Y^{(0)}\right](1-e(X))}\right) =  \mathbb{E}\left[\frac{\Var\left(Y^{(0)}|X\right)}{(1-e(X))\mathbb{E}\left[Y^{(0)}\right]^2} \right].
\end{align*}
Besides, 
\begin{align*}
    \Cov\left(\frac{T(Y-\mu_1(X))}{\mathbb{E}\left[Y^{(1)}\right]e(X)},\frac{(1-T)(Y-\mu_0(X))}{\mathbb{E}\left[Y^{(0)}\right](1-e(X))} \right) &= \mathbb{E}\left[\frac{T(Y-\mu_1(X))}{\mathbb{E}\left[Y^{(1)}\right]e(X)}\frac{(1-T)(Y-\mu_0(X))}{\mathbb{E}\left[Y^{(0)}\right](1-e(X))}\right]\\ 
    & \quad - \mathbb{E}\left[\frac{T(Y-\mu_1(X))}{\mathbb{E}\left[Y^{(1)}\right]e(X)}\right] \mathbb{E}\left[\frac{(1-T)(Y-\mu_0(X))}{\mathbb{E}\left[Y^{(0)}\right](1-e(X))}\right]\\
    &= 0.
\end{align*}
Gathering all these results into \eqref{eq_proof_38}, we obtain
\begin{align*}
    \Var\left(\frac{T(Y-\mu_1(X))}{\mathbb{E}\left[Y^{(1)}\right]e(X)} - \frac{(1-T)(Y-\mu_0(X))}{\mathbb{E}\left[Y^{(0)}\right](1-e(X))}\right) = \mathbb{E}\left[ \frac{\Var\left(Y^{(1)}|X\right)}{e(X)\mathbb{E}\left[Y^{(1)}\right]^2} \right] + \mathbb{E}\left[\frac{\Var\left(Y^{(0)}|X\right)}{(1-e(X))\mathbb{E}\left[Y^{(0)}\right]^2} \right].
\end{align*}
%Moreover we also have that:
%\begin{align*}
%    \Cov\left(\frac{\mu_1(X)}{\mathbb{E}\left[Y^{(1)}\right]} - \frac{\mu_0(X)}{\mathbb{E}\left[Y^{(0)}\right]}; \frac{T(Y-\mu_1(X))}{\mathbb{E}\left[Y^{(1)}\right]e(X)} - \frac{(1-T)(Y-\mu_0(X))}{\mathbb{E}\left[Y^{(0)}\right](1-e(X))}\right) &= 0
%\end{align*}
In order to rewrite the last term in \eqref{eq_proof_39}, note that
\begin{align*}
    \Cov\left(\mu_1(X), \frac{T(Y-\mu_1(X))}{e(X)}\right) &=  \mathbb{E}\left[\mu_1(X) \frac{T(Y-\mu_1(X))}{e(X)}\right] - \mathbb{E}\left[\mu_1(X)\right]\mathbb{E}\left[ \frac{T(Y-\mu_1(X))}{e(X)}\right]\\
    &= \mathbb{E}\left[\mu_1(X) \frac{T(Y-\mu_1(X))}{e(X)}\right]\\
    &= \mathbb{E}\left[\frac{\mu_1(X)}{e(X)}\mathbb{E}\left[ T(Y^{(1)}-\mu_1(X))|X\right]\right]\\
    &= \mathbb{E}\left[\frac{\mu_1(X)}{e(X)}\mathbb{E}\left[ \frac{T(Y^{(1)}-\mu_1(X))}{e(X)}|X\right]\right]\\
    &= \mathbb{E}\left[\frac{\mu_1(X)e(X)}{e(X)}(\mu_1(X))-\mu_1(X))\right]\\
    &=0.
\end{align*}
Similar calculations leads to 
\begin{align*}
    \Cov\left(\frac{\mu_1(X)}{\mathbb{E}\left[Y^{(1)}\right]} - \frac{\mu_0(X)}{\mathbb{E}\left[Y^{(0)}\right]}; \frac{T(Y-\mu_1(X))}{\mathbb{E}\left[Y^{(1)}\right]e(X)} - \frac{(1-T)(Y-\mu_0(X))}{\mathbb{E}\left[Y^{(0)}\right](1-e(X))}\right) &= 0.
\end{align*}
%We similarly show that all other terms are null:
%\begin{align*}
%    \Cov\left(\mu_0(X), \frac{T(Y-\mu_1(X))}{e(X)}\right) = 0 \quad \Cov\left(\mu_1(X), \frac{(1-T)(Y-\mu_0(X))}{1-e(X)}\right) = 0
%\end{align*}
%\begin{align*}
%    \Cov\left(\mu_0(X), \frac{(1-T)(Y-\mu_0(X))}{1-e(X)}\right) = 0
%\end{align*}
Finally
\begin{align*}
    V_{\textrm{\tiny RR,OS}} = \tau_{\textrm{\tiny RR}}^2\left(\Var\left(\frac{\mu_1(X)}{\mathbb{E}\left[Y^{(1)}\right]} - \frac{\mu_0(X)}{\mathbb{E}\left[Y^{(0)}\right]}\right) +  \mathbb{E}\left[ \frac{\Var\left(Y^{(1)}|X\right)}{e(X)\mathbb{E}\left[Y^{(1)}\right]^2} \right] + \mathbb{E}\left[\frac{\Var\left(Y^{(0)}|X\right)}{(1-e(X))\mathbb{E}\left[Y^{(0)}\right]^2} \right] \right).
\end{align*}
An estimator \(\hat{V}_{RR,OS}\) can be derived as follows:

\begin{align}\label{V_RR_OS}
\hat{V}_{RR,OS} = \frac{\hat{\tau}_{\textrm{\tiny RR,OS,n}}^2}{n} \sum_{i=1}^n \left(\Delta_i - \frac{1}{n}\sum_{i=1}^n \Delta_i \right)^2
\end{align}

where

\[
\Delta_i = \frac{\hat \Gamma_i(1)}{\hat S(1)} - \frac{\hat \Gamma_i(0)}{\hat S(0)}
\]

with the intermediate for \(t \in \{0,1\}\) quantities defined as:

\[\hat \Gamma_i(t) = \hat \mu_t(X_i) + \mathds{1}_{T_i=t} \frac{Y_i - \hat \mu_t(X_i)}{\hat e_t(X_i)} \quad \text{and} \quad \hat S(t) = \frac{1}{n}\sum_{j=1}^n \hat \Gamma_j(t)\] 

\end{proof}

\subsubsection{Risk Ratio Augmented Inverse Propensity Weighting}

\begin{proof}[Proof of \Cref{RAIPW}]\label{proof:aipw_normality_obs}
We use the derivations established in the proof of \Cref{prop:one_step2}. Indeed, we showed in  \Cref{proof:one_step} that the influence function \(\varphi\) of \(\psi = \frac{\mathbb{E}\left[\mathbb{E}\left[Y| T= 1, X\right]\right]}{\mathbb{E}\left[\mathbb{E}\left[Y| T= 0, X\right]\right]}\) can be written: 
\begin{align*}
    \varphi(Z; P) =\frac{\mu_1(X) + T\frac{Y-\mu_1(X)}{e(X)} }{\psi_0} -\psi \frac{\mu_0(X) + (1-T)\frac{Y-\mu_0(X)}{1-e(X)}}{\psi_0}.
\end{align*}
Using \Cref{eq_proof_RR_OS_VM_expansion}, and knowing that \(\int \varphi(z; \hat{P}) d\hat P(z) = 0\), we have:
\begin{align*}
\psi(\hat P) - \psi(P) &= \int \varphi(z; \hat P) d(\hat P - P)(z) + R_2(\hat P, P) \\
&= R_2(P,\hat{P}) - \int \varphi(z; \hat{P}) dP(z) \\
&= \frac{1}{n} \sum_{i=1}^n \varphi(Z_i; P) - \frac{1}{n} \sum_{i=1}^n \varphi(z_i; \hat{P}) \\
& \quad + \frac{1}{n} \sum_{i=1}^n \left(\varphi(Z_i; \hat{P}) - \varphi(Z_i; P)\right) - \int \left(\varphi(z; \hat{P})- \varphi(z; P) \right)dP(z) \\
& \quad +   R_2(P,\hat{P}).
\end{align*}

As outlined in \cite{Schuler}, in the estimating equation approach, we assume that the efficient influence function for any given distribution depends solely on the target parameter \(\psi\) and a set of nuisance parameters \(\eta\). Therefore, instead of denoting the efficient influence function as \(\varphi(z; P)\), we can express it as \(\varphi(Z; \psi, \eta)\). If the influence function can be represented in this form, we proceed by first estimating \(\hat{\eta} = (\hat{e}, \hat{\mu}_1, \hat{\mu}_0)\) with crossfitting. For any fixed value \(\hat{\eta}\), we find a value \(\hat{\psi}\) such that \(P_n \varphi_{\hat{\psi}, \hat{\eta}} = 0\), that is 
\begin{align*}
    \frac{1}{n}\sum_{i=1}^{n} \frac{\hat \mu_1(X_i) + T_i\frac{Y_i-\hat \mu_1(X_i)}{\hat e(X_i)} }{\hat \psi_0} -\hat\psi \frac{\hat \mu_0(X_i) + (1-T_i)\frac{Y_i- \hat \mu_0(X_i)}{1- \hat e(X_i)}}{\hat \psi_0} = 0,
\end{align*}
which implies
\begin{align*}
    \hat\psi = \frac{\sum_{i=1}^{n}\hat \mu_1(X_i) + T_i\frac{Y_i-\hat \mu_1(X_i)}{\hat e(X_i)} }{\sum_{i=1}^{n}\hat \mu_0(X_i) + (1-T_i)\frac{Y_i- \hat \mu_0(X_i)}{1- \hat e(X_i)}}.
\end{align*}
\end{proof}
Hereafter, we propose a proof of \Cref{prop:aipw_normality_obs} which does not use the influence function theory
\begin{proof}[Proof of \Cref{prop:aipw_normality_obs}]
\ \\
\textbf{Asymptotic bias and variance of the crossfitted Ratio AIPW estimator} In this alternative proof, we further assume that $\textrm{Var}[Y|X] \leq \sigma^2$ for some $\sigma>0$. 
Recall that we want to analyze \(\sqrt{n}\left(\hat\tau_{\textrm{\tiny RR,AIPW}} - \tau_{\textrm{\tiny RR}} \right)\). Letting 
\begin{align}
\tau_{\textrm{\tiny RR,AIPW}}^\star = \frac{\sum_{i=1}^n \mu_1(X_i) + \frac{T_i(Y_i- \mu_1(X_i))}{ e(X_i)}}{\sum_{i=1}^n \mu_0(X_i) +  \frac{(1-T_i)(Y_i- \mu_0(X_i))}{1- e(X_i)}}:= \frac{\tau_{\textrm{\tiny RR,AIPW, 1}}^\star}{\tau_{\textrm{\tiny RR,AIPW, 0}}^\star} 
\end{align}
be the oracle version of \(\hat\tau_{\textrm{\tiny RR,AIPW}}\) where the propensity score and both response surfaces are assumed to be known, we can rewrite
\begin{align}
    \hat\tau_{\textrm{\tiny RR,AIPW}}- \tau_{\textrm{\tiny RR}} & = \hat\tau_{\textrm{\tiny RR,AIPW}} - \tau_{\textrm{\tiny RR,AIPW}}^\star+ \tau_{\textrm{\tiny RR,AIPW}}^\star - \tau_{\textrm{\tiny RR}}.   \label{AIPW:eq_proof2}
\end{align}
Regarding the first term in \eqref{AIPW:eq_proof2}, we have
\begin{align}
    & \left|\hat\tau_{\textrm{\tiny RR,AIPW}}- \tau_{\textrm{\tiny RR,AIPW}}^\star \right| \\&= \left|\frac{\hat \tau_{\textrm{\tiny RR,AIPW, 1}}}{\hat\tau_{\textrm{\tiny RR,AIPW, 0}}} - \frac{\tau_{\textrm{\tiny RR,AIPW, 1}}^\star}{\tau_{\textrm{\tiny RR,AIPW, 0}}^\star} \right|\\
    &= \left|\left(\left(\hat \tau_{\textrm{\tiny RR,AIPW, 0}}\right)^{-1} - \left(\tau_{\textrm{\tiny RR,AIPW, 0}}^\star\right)^{-1}\right)\hat \tau_{\textrm{\tiny RR,AIPW, 1}}\right. \\
    & \quad + \left.\left(\tau_{\textrm{\tiny RR,AIPW, 0}}^\star\right)^{-1}\left(\hat \tau_{\textrm{\tiny RR,AIPW, 1}} - \tau_{\textrm{\tiny RR,AIPW, 1}}^\star\right)\right|\\
    &\leq \left|\left(\left(\hat \tau_{\textrm{\tiny RR,AIPW, 0}}\right)^{-1} - \left(\tau_{\textrm{\tiny RR,AIPW, 0}}^\star\right)^{-1}\right)\hat \tau_{\textrm{\tiny RR,AIPW, 1}}\right| \\
    & \quad + \left|\left(\tau_{\textrm{\tiny RR,AIPW, 0}}^\star\right)^{-1}\left(\hat \tau_{\textrm{\tiny RR,AIPW, 1}} - \tau_{\textrm{\tiny RR,AIPW, 1}}^\star\right)\right|. \label{eq_proof_41}
\end{align}
We now show that
\[\left|\hat \tau_{\textrm{\tiny RR,AIPW, 1}} - \tau_{\textrm{\tiny RR,AIPW, 1}}^\star\right| = o_p\left(\frac{1}{\sqrt{n}}\right) \quad \text{and} \quad \left|\hat \tau_{\textrm{\tiny RR,AIPW, 0}} - \tau_{\textrm{\tiny RR,AIPW, 0}}^\star\right| = o_p\left(\frac{1}{\sqrt{n}}\right).\]
The following decomposition holds
\begin{align*}
    \sqrt{n}\left|\hat \tau_{\textrm{\tiny RR,AIPW, 1}} - \tau_{\textrm{\tiny RR,AIPW, 1}}^\star\right| =& \frac{1}{\sqrt{n}} \sum_{i \in \mathcal{I}_k}\left(\hat \mu_1^{\mathcal{I}_{-k}}\left(X_{i}\right)+T_{i} \frac{Y_{i}-\hat \mu_1^{\mathcal{I}_{-k}}\left(X_{i}\right)}{\hat{e}\left(X_{i}\right)}-\mu_{1}\left(X_{i}\right)-T_{i} \frac{Y_{i}-\mu_{1}\left(X_{i}\right)}{e\left(X_{i}\right)}\right) \\
\text{Further denoted $A_n^k$} \qquad=& \frac{1}{\sqrt{n}}  \sum_{i \in \mathcal{I}_k}\left(\left(\hat \mu_1^{\mathcal{I}_{-k}}\left(X_{i}\right)-\mu_{1}\left(X_{i}\right)\right)\left(1-\frac{T_{i}}{e\left(X_{i}\right)}\right)\right) \\
\text{Further denoted $B_n^k$} \qquad &+\frac{1}{\sqrt{n}}  \sum_{i \in \mathcal{I}_k} T_{i}\left(\left(Y_{i}-\mu_{1}\left(X_{i}\right)\right)\left(\frac{1}{\hat{e}\left(X_{i}\right)}-\frac{1}{e\left(X_{i}\right)}\right)\right) \\
\text{Further denoted $C_n^k$} \qquad &-\frac{1}{\sqrt{n}}  \sum_{i \in \mathcal{I}_k} T_{i}\left(\left(\hat \mu_1^{\mathcal{I}_{-k}}\left(X_{i}\right)-\mu_{1}\left(X_{i}\right)\right)\left(\frac{1}{\hat{e}\left(X_{i}\right)}-\frac{1}{e\left(X_{i}\right)}\right)\right).
\end{align*}

In the following, we prove that the first two terms tend to zero in $L^2$. 

\textbf{Regarding $A_n^k$} One can show that the expectation of $A_n^k/\sqrt{n}$ is null:
\begin{align*}
     \mathbb{E}\left[\frac{A_n^k}{\sqrt{n}} \mid \mathcal{I}_{-k}\right] 
     &= \frac{1}{n} \sum_{i \in \mathcal{I}_k}\mathbb{E}\left[\left(\hat \mu_1^{\mathcal{I}_{-k}}\left(X_{i}\right)-\mu_{1}\left(X_{i}\right)\right)\left(1-\frac{T_{i}}{e\left(X_{i}\right)}\right) \mid \mathcal{I}_{-k} \right] \\
     &= \frac{|\mathcal{I}_k|}{n} \mathbb{E}\left[\left(\hat \mu_1^{\mathcal{I}_{-k}}\left(X\right)-\mu_{1}\left(X\right)\right)\left(1-\frac{T}{e\left(X\right)}\right) \mid \mathcal{I}_{-k} \right] && \text{i.i.d.} \\
     &= \frac{|\mathcal{I}_k|}{n} \mathbb{E}\left[ \mathbb{E}\left[ \left(\hat \mu_1^{\mathcal{I}_{-k}}\left(X\right)-\mu_{1}\left(X\right)\right)\left(1-\frac{T}{e\left(X\right)}\right)\mid X, \mathcal{I}_{-k}\right] \mid \mathcal{I}_{-k} \right] \\
     &= \frac{|\mathcal{I}_k|}{n} \mathbb{E}\left[  \left(\hat \mu_1^{\mathcal{I}_{-k}}\left(X\right)-\mu_{1}\left(X\right)\right) \mathbb{E}\left[ \left(1-\frac{T}{e\left(X\right)}\right)\mid X, \mathcal{I}_{-k}\right] \mid \mathcal{I}_{-k} \right] \\
     %&= \frac{|\mathcal{I}_k|}{n}  \mathbb{E}\left[\mathbb{E}\left[ \left(\hat \mu_1^{\mathcal{I}_{-k}}\left(X\right)-\mu_{1}\left(X\right)\right)\left(1-\frac{\mathbb{E}\left[T\mid X \right] }{e\left(X\right)}\right)\right]\right] \\ 
     &= \frac{|\mathcal{I}_k|}{n}  \mathbb{E}\left[ \left(\hat \mu_1^{\mathcal{I}_{-k}}\left(X\right)-\mu_{1}\left(X\right)\right)\left(1-\frac{e\left(X\right) }{e\left(X\right)}\right) \mid \mathcal{I}_{-k} \right] \\ 
     &= 0.
\end{align*}
We will make use of this results in several calculations. Now, 
\begin{align*}
    \mathbb{E}\left[ \left(\frac{A_n^k}{\sqrt{n}}\right)^2 \mid \mathcal{I}_{-k}\right] 
    &= \operatorname{Var}\left[ \frac{1}{n}  \sum_{i \in \mathcal{I}_k}\left(\left(\hat \mu_1^{\mathcal{I}_{-k}}\left(X_{i}\right)-\mu_{1}\left(X_{i}\right)\right)\left(1-\frac{T_{i}}{e\left(X_{i}\right)}\right)\right) \mid \mathcal{I}_{-k}\right] \\
    &= \frac{1}{n^2} \operatorname{Var}\left[  \sum_{i \in \mathcal{I}_k}\left(\left(\hat \mu_1^{\mathcal{I}_{-k}}\left(X_{i}\right)-\mu_{1}\left(X_{i}\right)\right)\left(1-\frac{T_{i}}{e\left(X_{i}\right)}\right)\right) \mid \mathcal{I}_{-k}\right]\\
    &= \frac{1}{n^2}  \sum_{i \in \mathcal{I}_k} \operatorname{Var}\left[ \left(\hat \mu_1^{\mathcal{I}_{-k}}\left(X_{i}\right)-\mu_{1}\left(X_{i}\right)\right)\left(1-\frac{T_{i}}{e\left(X_{i}\right)}\right)\mid \mathcal{I}_{-k} \right] &&\text{iid} \\
    &= \frac{|\mathcal{I}_k|}{n^2} \mathbb{E}\left[\left( \left(\hat \mu_1^{\mathcal{I}_{-k}}\left(X\right)-\mu_{1}\left(X\right)\right)\left(1-\frac{T}{e\left(X\right)}\right)\right)^2\mid \mathcal{I}_{-k} \right]  \\
    &=  \frac{|\mathcal{I}_k|}{n^2} \mathbb{E}\left[\mathbb{E}\left[\left( \left(\hat \mu_1^{\mathcal{I}_{-k}}\left(X\right)-\mu_{1}\left(X\right)\right)\left(1-\frac{T}{e\left(X\right)}\right)\right)^2 |X, \mathcal{I}_{-k} \right] \mid \mathcal{I}_{-k}\right]  \\
    &=  \frac{|\mathcal{I}_k|}{n^2} \mathbb{E}\left[ \left(\hat \mu_1^{\mathcal{I}_{-k}}\left(X\right)-\mu_{1}\left(X\right)\right)^2\mathbb{E}\left[\left(1-\frac{T}{e\left(X\right)}\right)^2 |X, \mathcal{I}_{-k} \right] \mid \mathcal{I}_{-k}\right]  \\
    &=  \frac{|\mathcal{I}_k|}{n^2} \mathbb{E}\left[\left(\hat \mu_1^{\mathcal{I}_{-k}}\left(X\right)-\mu_{1}\left(X\right)\right)^2\frac{1}{e\left(X\right)^2}\mathbb{E}\left[ \left(e\left(X\right)-T\right)^2 |X, \mathcal{I}_{-k} \right] \mid \mathcal{I}_{-k}\right]  \\
    &=  \frac{|\mathcal{I}_k|}{n^2} \mathbb{E}\left[\left(\hat \mu_1^{\mathcal{I}_{-k}}\left(X\right)-\mu_{1}\left(X\right)\right)^2\frac{e\left(X\right)(1-e\left(X\right))}{e\left(X\right)^2} \mid \mathcal{I}_{-k}\right]  \\
    &=  \frac{|\mathcal{I}_k|}{n^2} \mathbb{E}\left[\left(\hat \mu_1^{\mathcal{I}_{-k}}\left(X\right)-\mu_{1}\left(X\right)\right)^2\left(\frac{1}{e\left(X\right)} -1\right) \mid \mathcal{I}_{-k}\right]  \\
    &\leq  \frac{|\mathcal{I}_k|}{\eta n^2} \mathbb{E}\left[\left(\hat \mu_1^{\mathcal{I}_{-k}}\left(X\right)-\mu_{1}\left(X\right)\right)^2\mid \mathcal{I}_{-k} \right]  && \text{Overlap}.
    %\\
    %&\leq  \frac{|\mathcal{I}_k|}{ n^2}o_{\mathbb{P}}(1)
\end{align*}
Taking the expectation, we obtain
\begin{align}
\mathbb{E}\left[ \left(\frac{A_n^k}{\sqrt{n}}\right)^2 \right] \leq \frac{|\mathcal{I}_k|}{\eta n^2} \mathbb{E}\left[\left(\hat \mu_1^{\mathcal{I}_{-k}}\left(X\right)-\mu_{1}\left(X\right)\right)^2  \right],  
\end{align}
that is 
\begin{align}
\mathbb{E}\left[ \left(A_n^k \right)^2 \right] \leq \frac{1}{\eta} \mathbb{E}\left[\left(\hat \mu_1^{\mathcal{I}_{-k}}\left(X\right)-\mu_{1}\left(X\right)\right)^2  \right].    
\end{align}
Thus $A_n^k$ converges to zero in $L^2$ and thus in probability. 
%Therefore, because convergence on $L^{2}$-norm provides convergence in probability (Chebyshev inequality), we have for $k \in \{1, K\}$:
%$$\sqrt{n}\, \frac{1}{n}  \sum_{i \in \mathcal{I}_k}\left(\left(\hat \mu_1^{\mathcal{I}_{-k}}\left(X_{i}\right)-\mu_{1}\left(X_{i}\right)\right)\left(1-\frac{T_i}{e\left(X_{i}\right)}\right)\right)  \stackrel{p}{\longrightarrow} 0$$

\textbf{Regarding $B_n^k$} The second term $B_n^k$ can also be controlled using similar arguments. By assumption, 
%Before detailed derivation, note that due to the uniform convergence of $\hat e(.)$ and the overlap assumption, there exist $M$ such that for all $n > M$, and for all $X_i$,
$$ \frac{\eta}{2} \le \hat e(X) \le 1- \frac{\eta}{2}.$$
Thus, 
%Therefore, there exist $M$ such that for all $n > M$, and for all $X_i$,
\begin{align*}
    \frac{1}{\hat e(X)} - \frac{1}{e(X)} &= \frac{e(X) - \hat e(X)}{\hat e(X) e(X)}  \le 2 \left( \frac{e(X) - \hat e(X)}{\eta^2} \right).
\end{align*}

Derivations are very close to the ones for the first term, noting that, $$\mathbb{E}\left[ \mathbb{E}\left[\frac{1}{n} \sum_{i \in \mathcal{I}_k} T_i\left(\left(Y_{i}-\mu_{1}\left(X_{i}\right)\right)\left(\frac{1}{\hat{e}^{\mathcal{I}_{-k}}\left(X_{i}\right)}-\frac{1}{e\left(X_{i}\right)}\right)\right)\mid X_i, \mathcal{I}_{-k} \right]\mid \mathcal{I}_{-k}\right] =0,$$ so that,

\begin{align*}
\mathbb{E}\left[ \left( \frac{B_n^k}{\sqrt{n}} \right)^2 \mid \mathcal{I}_{-k}\right] 
&= \operatorname{Var}\left[ \frac{1}{n} \sum_{i \in \mathcal{I}_k} T_i\left(Y_{i}-\mu_{1}\left(X_{i}\right)\right)\left(\frac{1}{\hat{e}^{\mathcal{I}_{-k}}\left(X_{i}\right)}-\frac{1}{e\left(X_{i}\right)}\right) \mid \mathcal{I}_{-k}\right]\\
&= \frac{1}{n^2} \sum_{i \in \mathcal{I}_k} \operatorname{Var}\left[ T_i\left(Y_{i}-\mu_{1}\left(X_{i}\right)\right)\left(\frac{1}{\hat{e}^{\mathcal{I}_{-k}}\left(X_{i}\right)}-\frac{1}{e\left(X_{i}\right)}\right) \mid \mathcal{I}_{-k}\right] && \text{iid}\\
&= \frac{|\mathcal{I}_k|}{n^2} \mathbb{E}\left[ T\left(Y-\mu_{1}\left(X\right)\right)^2\left(\frac{1}{\hat{e}^{\mathcal{I}_{-k}}\left(X\right)}-\frac{1}{e\left(X\right)}\right)^2 \mid \mathcal{I}_{-k}\right] \\
&\leq \frac{4|\mathcal{I}_k|}{\eta^4 n^2} \mathbb{E}\left[ T\left(Y-\mu_{1}\left(X\right)\right)^2\left(\hat{e}^{\mathcal{I}_{-k}}\left(X\right)-e\left(X\right)\right)^2 \mid \mathcal{I}_{-k}\right]\\
&\leq \frac{4|\mathcal{I}_k|}{\eta^4 n^2} \mathbb{E}\left[ \left(Y-\mu_{1}\left(X\right)\right)^2\left(\hat{e}^{\mathcal{I}_{-k}}\left(X\right)-e\left(X\right)\right)^2 \mid \mathcal{I}_{-k}\right]&& \text{Sicne $T\leq 1$}\\
&\leq \frac{4|\mathcal{I}_k|}{\eta^4 n^2} \mathbb{E}\left[\mathbb{E}\left[ \left(Y-\mu_{1}\left(X\right)\right)^2\left(\hat{e}^{\mathcal{I}_{-k}}\left(X\right)-e\left(X\right)\right)^2 | X, \mathcal{I}_{-k}\right]\mid \mathcal{I}_{-k}\right]\\
&\leq \frac{4|\mathcal{I}_k|}{\eta^4 n^2} \mathbb{E}\left[ \mathbb{E}\left[\left(Y-\mu_{1}\left(X\right)\right)^2| X, \mathcal{I}_{-k} \right]\left(\hat{e}^{\mathcal{I}_{-k}}\left(X\right)-e\left(X\right)\right)^2 \mid \mathcal{I}_{-k}\right]\\
&\leq \frac{4|\mathcal{I}_k|}{\eta^4 n^2} \mathbb{E}\left[ \operatorname{Var}\left[Y| X \right]\left(\hat{e}^{\mathcal{I}_{-k}}\left(X\right)-e\left(X\right)\right)^2 \mid \mathcal{I}_{-k}\right]\\
&\leq \frac{4\operatorname{Var}\left[Y| X \right]|\mathcal{I}_k|}{\eta^4 n^2} \mathbb{E}\left[ \left(\hat{e}^{\mathcal{I}_{-k}}\left(X\right)-e\left(X\right)\right)^2 \mid \mathcal{I}_{-k}\right].
%&\leq \frac{4\operatorname{Var}\left[Y| X \right]|\mathcal{I}_k|}{\eta^4 n^2} o_{\mathbb{P}}(1)\\
\end{align*}
Taking the expectation on both sides, since $\operatorname{Var}\left[Y| X \right] \leq \sigma^2$, we get 
\begin{align}
\mathbb{E}\left[ \left( \frac{B_n^k}{\sqrt{n}} \right)^2 \right]  \leq \frac{4 \sigma^2|\mathcal{I}_k|}{\eta^4 n^2} \mathbb{E}\left[ \left(\hat{e}^{\mathcal{I}_{-k}}\left(X\right)-e\left(X\right)\right)^2 \mid \mathcal{I}_{-k}\right],
\end{align}
which leads to 
\begin{align}
\mathbb{E}\left[ \left( B_n^k \right)^2 \right]  \leq \frac{4 \sigma^2}{\eta^4} \mathbb{E}\left[ \left(\hat{e}^{\mathcal{I}_{-k}}\left(X\right)-e\left(X\right)\right)^2  \mid \mathcal{I}_{-k}\right].  
\end{align}
Since, by assumption, the right-hand side term converges to zero, $B_n^k$ converges to zero in $L^2$. 

\textbf{Regarding $C_n^k$} Regarding the last term, the approach is different and will involve another assumption on the product of residuals. More precisely, 
%\es{l'hypothèse porte sur les espérances, pas sur les moyennes empiriques :/}
\begin{align*}
    \mathbb{E}[|C_n^k|] &= \sqrt{n}\frac{1}{n} \sum_{i \in \mathcal{I}_k} \mathbb{E}\left[ \left| T_i\left(\hat \mu_1^{\mathcal{I}_{-k}}\left(X_{i}\right)-\mu_{1}\left(X_{i}\right)\right)\left(\frac{1}{\hat{e}^{\mathcal{I}_{-k}}\left(X_{i}\right)}-\frac{1}{e\left(X_{i}\right)}\right) \right| \right]\\
    & = \frac{\sqrt{n}}{\eta^2}\frac{1}{n} \sum_{i \in \mathcal{I}_k} \mathbb{E}\left[ \left| T_i\left(\hat \mu_1^{\mathcal{I}_{-k}}\left(X_{i}\right)-\mu_{1}\left(X_{i}\right)\right)
    \left(e(X_i) - \hat e^{\mathcal{I}_{-k}}(X_i)\right) \right|  
    \right]\\
    & = \frac{\sqrt{n} |\mathcal{I}_k|}{\eta^2}\frac{1}{n}  \mathbb{E}\left[ \left| T\left(\hat \mu_1^{\mathcal{I}_{-k}}\left(X\right)-\mu_{1}\left(X\right)\right)
    \left(e(X) - \hat e^{\mathcal{I}_{-k}}(X)\right)  \right| 
    \right]\\
    & \leq \frac{\sqrt{n} }{\eta^2} \mathbb{E}\left[ \left| \left(\hat \mu_1^{\mathcal{I}_{-k}}\left(X\right)-\mu_{1}\left(X\right)\right)
    \left(e(X) - \hat e^{\mathcal{I}_{-k}}(X)\right)   \right|
    \right]\\
    & \leq \frac{\sqrt{n} }{\eta^2} \sqrt{\mathbb{E}\left[  \left(\hat \mu_1^{\mathcal{I}_{-k}}\left(X\right)-\mu_{1}\left(X\right)\right)^2 \right] 
    \mathbb{E}\left[ \left(e(X) - \hat e^{\mathcal{I}_{-k}}(X)\right)^2   
    \right]},
\end{align*}
which tends to zero by assumption. 
%\begin{align*}
%    \mathbb{E}[C_n^k] &= \sqrt{n}\frac{1}{n} \sum_{i \in \mathcal{I}_k}T_i\left(\hat \mu_1^{\mathcal{I}_{-k}}\left(X_{i}\right)-\mu_{1}\left(X_{i}\right)\right)\left(\frac{1}{\hat{e}^{\mathcal{I}_{-k}}\left(X_{i}\right)}-\frac{1}{e\left(X_{i}\right)}\right) \\
%    & \le  \sqrt{n}\sqrt{\frac{1}{n} \sum_{i \in \mathcal{I}_k}T_i\left(\hat \mu_1^{\mathcal{I}_{-k}}\left(X_{i}\right)-\mu_{1}\left(X_{i}\right)\right)^2 } \sqrt{\frac{1}{n} \sum_{i \in \mathcal{I}_k} \left(\frac{1}{\hat{e}^{\mathcal{I}_{-k}}\left(X_{i}\right)}-\frac{1}{e\left(X_{i}\right)}\right)^2}  && \text{C.S.} \\
%    &  \le \frac{\sqrt{n}}{ \eta^2} \sqrt{\frac{1}{n} \sum_{i \in \mathcal{I}_k}T_i\left(\hat \mu_1^{\mathcal{I}_{-k}}\left(X_{i}\right)-\mu_{1}\left(X_{i}\right)\right)^2 } \sqrt{\frac{1}{n} \sum_{i \in \mathcal{I}_k} \left(e(X_i) - \hat e^{\mathcal{I}_{-k}}(X_i)\right)^2  }  && \text{Overlap} \\
%    & = \frac{\sqrt{n}}{ \eta^2}\, o_{\mathbb{P}}(\frac{1}{\sqrt{n}}) && \text{Assumption} \\
%    & = \frac{1}{ \eta}\, o_{\mathbb{P}}(1)
%\end{align*}
%
Each term $A_n^k$, $B_n^k$, and $C_n^k$ has been shown to be bounded by a term in $o_{\mathbb{P}}(1)$. Thus, 
\begin{align}
    \sqrt{n}\left|\hat \tau_{\textrm{\tiny RR,AIPW, 1}} - \tau_{\textrm{\tiny RR,AIPW, 1}}^\star\right| &= \sum_{k=1}^K A_n^k + B_n^k + C_n^k
\end{align}
tends to zero in probability. Similarly, one can show that 
\begin{align}
    \sqrt{n}\left|\hat \tau_{\textrm{\tiny RR,AIPW, 1}} - \tau_{\textrm{\tiny RR,AIPW, 1}}^\star\right|  \stackrel{p}{\longrightarrow} 0.
\end{align}
%\es{je ne comprends pas cette phrase}
%The remainder term with function $\hat m^k_0(Y,A,X) -  m_0(Y,A,X)$ can also be controlled with the same derivation, except using the uniform convergence of $\hat \mu^{\mathcal{I}_{-k}}_0(.)$.
%Since we have:
%\[\sqrt{n}  (\hat \tau_{\text{AIPW}} - \hat\tau_{\text{AIPW}}^\star) = \sum_{k=1}^K A_n^k + B_n^k + C_n^k \]
%Therefore, \[\sqrt{n}  (\hat \tau_{\textrm{\tiny RR,AIPW, 1}} - \tau_{\textrm{\tiny RR,AIPW, 1}}^\star) \stackrel{p}{\longrightarrow} 0.\]
%And similarly:
%\[\sqrt{n}  (\hat \tau_{\textrm{\tiny RR,AIPW, 0}} - \tau_{\textrm{\tiny RR,AIPW, 0}}^\star) \stackrel{p}{\longrightarrow} 0\]

According to \eqref{eq_proof_41}, since for all $t \in \{0,1\}$,  $|\hat \tau_{\textrm{\tiny RR,AIPW,t}}|$ tends to $\tau_{\textrm{\tiny RR,AIPW,t}}^\star$ which is lower and upper bounded, we have
\begin{align*}
    \sqrt{n} \left|\hat \tau_{\textrm{\tiny RR,AIPW}} - \tau_{\textrm{\tiny RR,AIPW}}^\star \right| &\leq \sqrt{n} \left| \frac{\hat \tau_{\textrm{\tiny RR,AIPW, 1}}}{\hat \tau_{\textrm{\tiny RR,AIPW, 0}} \tau_{\textrm{\tiny RR,AIPW, 0}}^\star } \right|\left| \hat \tau_{\textrm{\tiny RR,AIPW, 0}} - \tau_{\textrm{\tiny RR,AIPW, 0}}^\star \right| \\
    & \quad + \sqrt{n} \left| \frac{1}{\tau_{\textrm{\tiny RR,AIPW, 0}}^\star}\right| \left| \hat \tau_{\textrm{\tiny RR,AIPW, 1}} - \tau_{\textrm{\tiny RR,AIPW, 1}}^\star\right|
\end{align*}
which tends to zero. 
%Consequently, in probability, as \(n\) tends to infinity,
%\begin{align}
%    \sqrt{n} \left|\hat \tau_{\textrm{\tiny RR,AIPW}} - \tau_{\textrm{\tiny RR,AIPW}}^\star \right| \to 0. \label{AIPW:eq_proof3}
%\end{align}

Regarding the second term in \eqref{AIPW:eq_proof2}, we can use \Cref{Th1} with \(g_1(Z) = \mu_1(X) + \frac{T(Y- \mu_1(X))}{ e(X)}\) and \(g_0(Z) = \mu_0(X) + \frac{(1-T)(Y- \mu_0(X))}{ (1-e(X))}\) where \(Z= (T,X,Y)\). Hence, we have that  \(g_1(Z)\) is square integrable:
\begin{align*}
    \mathbb{E}\left[g_1(Z)^2\right] & \leq 
     2 \mathbb{E}\left[(\mu_1(X)^2\right] + 2 \mathbb{E}\left[\left(\frac{T(Y- \mu_1(X))}{ e(X)}\right)^2\right],   
\end{align*}
where \(\mathbb{E}\left[(\mu_1(X)^2\right] = \Var(Y^{(1)}) + \mathbb{E}\left[Y^{(1)}\right]^2\)  is finite.
Using Consistency, Unconfoundedness, and definition or $\mu_1(X) = \mathbb{E}[Y \mid X, T=1]$, simple calculations show that
\begin{align*}
    \mathbb{E}\left[ \left(T\frac{Y-\mu_1(X)}{e(X)} \right)^2 \right] &=   \mathbb{E}\left[ \left(T\frac{Y^{(1)}-\mu_1(X)}{e(X)} \right)^2 \right] && \text{Consistency} \\
    &=\mathbb{E}\left[ \mathbb{E}\left[\left(T\frac{Y^{(1)}-\mu_1(X)}{e(X)} \right)^2 \mid X \right]\right] && \text{Total expectation}\\
    %&= \mathbb{E}\left[ \mathbb{E}\left[T\left(\frac{Y^{(1)}-\mu_1(X)}{e(X)} \right)^2 \mid X \right]\right] && \text{T is binary}\\
    %&= \mathbb{E}\left[ \mathbb{E}\left[\mathbf{1}_{\left\{T=1\right\}}\left(\frac{Y^{(1)}-\mu_1(X)}{e(X)} \right)^2 \mid X \right]\right] && \text{T written as an indicator} \\
    % &= \mathbb{E}\left[ \frac{1}{e(X)^2} \mathbb{E}\left[\mathbf{1}_{\left\{T=1\right\}}\left(Y^{(1)}-\mu_1(X)\right)^2 \mid X \right]\right] && \text{$e(X)$ is a function of $X$} \\
      %&= \mathbb{E}\left[ \frac{\left(Y^{(1)}-\mu_1(X)\right)^2 }{e(X)^2} \mathbb{E}\left[\mathbf{1}_{\left\{T=1\right\}}\mid X \right]\right] && \text{Uncounf. \& $\mu_1(.)$ is func. of $X$} \\
      %&= \mathbb{E}\left[ \frac{\left(Y^{(1)}-\mu_1(X)\right)^2 }{e(X)^2} e(X)\right] && \text{Definition of $e(X)$} \\
      &= \mathbb{E}\left[ \frac{\left(Y^{(1)}-\mu_1(X)\right)^2 }{e(X)} \right] \\
      &\leq \frac{\Var(\mu_1(X))}{\eta}.
\end{align*}
Similarly, we can show that \(g_0(Z)\) is square integrable. 
Since  \(\mathbb{E}\left[g_0(Z)\right] = \mathbb{E}\left[Y^{(0)} \right]\) and \(\mathbb{E}\left[g_1(Z)\right] = \mathbb{E}\left[Y^{(1)} \right]\), we can  apply \Cref{Th1} and conclude that
\begin{align}
    \sqrt{n} (\tau_{\textrm{\tiny RR,AIPW}}^\star - \tau_{\textrm{\tiny RR}}) \to \mathcal{N}(0, V_{\textrm{\tiny RR,OS}}). \label{eq_proof4}
\end{align}
Finally, 
\begin{equation*}
    \sqrt{n} (\hat \tau_{\text{AIPW}} - \tau_{\textrm{\tiny RR}}) = \underbrace{\sqrt{n} (\hat \tau_{\textrm{\tiny RR,AIPW}} - \tau_{\textrm{\tiny RR,AIPW}}^\star)}_\textrm{$\stackrel{p}{\longrightarrow} 0$}  +  \underbrace{\sqrt{n}(\tau_{\textrm{\tiny RR,AIPW}}^\star - \tau_{\textrm{\tiny RR}})}_\textrm{$\stackrel{d}{\rightarrow} \mathcal{N}\left(0, V_{\textrm{\tiny RR,OS}} \right)$},
\end{equation*}
where 
\begin{align*}
    V_{\textrm{\tiny RR,OS}} = \tau_{\textrm{\tiny RR}}^2\left(\Var\left(\frac{\mu_1(X)}{\mathbb{E}\left[Y^{(1)}\right]} - \frac{\mu_0(X)}{\mathbb{E}\left[Y^{(0)}\right]}\right) +  \mathbb{E}\left[ \frac{\Var\left(Y^{(1)}|X\right)}{e(X)\mathbb{E}\left[Y^{(1)}\right]^2} \right] + \mathbb{E}\left[\frac{\Var\left(Y^{(0)}|X\right)}{(1-e(X))\mathbb{E}\left[Y^{(0)}\right]^2} \right] \right)
\end{align*}
\end{proof}

\newpage

\section{Simulation} \label{Simulation-2}
For the simulations we have implemented all estimators in Python using Scikit-Learn for our regression and classification models. All our experiments were run on a 8GB M1 Mac. The propensity scores is estimated based on the provided training data and covariate names. Depending on the chosen method, it either uses logistic regression with a high regularization parameter (parametric) or a random forest classifier with parameters determined by the training data size (non-parametric). The response surface is estimated based on the training data, covariate names, the method (parametric or non-parametric), and whether the response is binary or continuous. For parametric methods, it uses a stochastic gradient descent classifier for binary responses and a linear regression model for continuous responses. For non-parametric methods, it employs a random forest classifier for binary responses and a random forest regressor for continuous responses. Both methods fit the model using the training data to estimate the respective scores and surfaces, enabling flexible handling of various datasets and assumptions for causal inference analysis.

\subsection{Randomized Controlled Trials}

In this part we will simulate Randomized Controlled Trials (RCT) and test the following Ratio estimators: Ratio Neyman, Ratio Horvitz Thomson and the Ratio G-formula. Since we are in a Randomized Controlled Trials, the propensity score \(e(.)\)  is constant.

\subsubsection{Linear RCT}
The first DGP has linear outcome models (linear treatment effect and the baseline). The data is generated using:

\[
    \begin{array}{lll}
        m(X) = (c_1 - c_0) + (\beta_1 - \beta_0)^{\top}X \\
        b(X) = c_0 + \beta_0^{\top}X  \\
        e(X) = 0.5
    \end{array} \qquad \begin{array}{lll}
    c_0 &= 6, \qquad c_1 = 12 \\
    \beta_1 &= (2, -5, 2, 8, -2, 8) \\
    \beta_0 &= (3, -7, 1, 4, -2, 2) 
\end{array}
\] 

\begin{figure}[h!]
    \centering
    \includegraphics[width=0.55\textwidth]{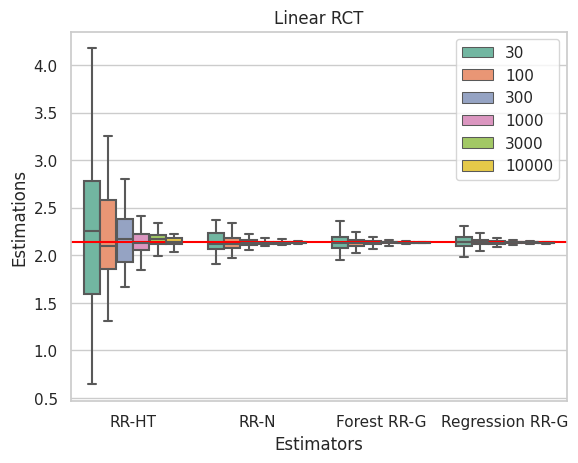}
    \caption{Comparison of RCT estimators in a Linear RCT}
    \label{fig:Linear_RCT}
\end{figure}

Given that \(X\) has a zero mean, it follows that \(\tau_{\textrm{\tiny RR}} = c_1/c_0= 2\). This scenario aligns with the linear setting outlined in \Cref{a:linear_model}. Referring to Figure \ref{fig:Linear_RCT}, as proved in the previous sections all estimators converge to the true Risk Ratio as \(n\) increases. Additionally, within this linear framework as per \Cref{lem-var-lin}, the variance of the Neyman estimator exceeds the one of the G-formula. In such a linear environment, the parametric G-formula performs better than its non-parametric counterpart. Additionally, the Ratio Neyman estimator demonstrates lower variance compared to the Horvitz-Thomson estimator as indicated in \Cref{HT_VS_N}.

\subsubsection{Non-Linear RCT}
This DGP is also a Randomized Controlled Trials however, the outcomes are not linear this time:

\begin{align*}
    m(X) =  \sin(X_1) \cdot X_2^2 + \frac{X_3}{X_4 + 1} - \log(X_5 + 1) + X_6^3 + 1 \\
    b(X) = 4*\max(X_1+X_2+X_3,0) -\min(X_4+X_6,0) \quad \text{and} \quad e(X) = 0.5
\end{align*}

\begin{figure}[h!]
    \centering
    \includegraphics[width=0.55\textwidth]{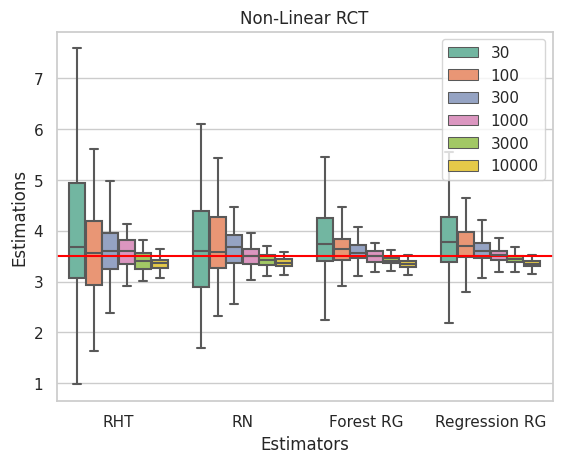}
    \caption{Comparison of RCT estimators in a Non-Linear RCT}
    \label{fig:non_Linear_RCT}
\end{figure}

The presence of trigonometric, exponential, logarithmic, and polynomial terms makes this setting non-linear. It's important to note that since we are in a Randomized Controlled Trial (RCT), the propensity function remains constant. As the sample size (\(n\)) increases, all proposed estimators converge. A bias can be seen in \ref{fig:non_Linear_RCT} but decreases to \(0\) as (\(n\)) increases as predicted in previous sections. Linear regression struggles with small \(n\) values, failing to capture the intricate relationships between features and non-linearities. On the other hand, Random Forest, a non-parametric method, excels in capturing these complexities by segmenting the feature space and predicting based on response averages within those segments. However, predicting the complex function can be challenging, the Neyman estimator might outperform the G-formula, particularly when both parametric and non-parametric responses may lack consistency. Although we do not fall in assumptions of \Cref{HT_VS_N} the Ratio Neyman estimator demonstrates lower variance compared to the Horvitz-Thomson estimator.

\subsection{Observational Studies}\label{sim:CI}
\begin{figure*}[h]
    \centering
    \includegraphics[width=0.8\textwidth]{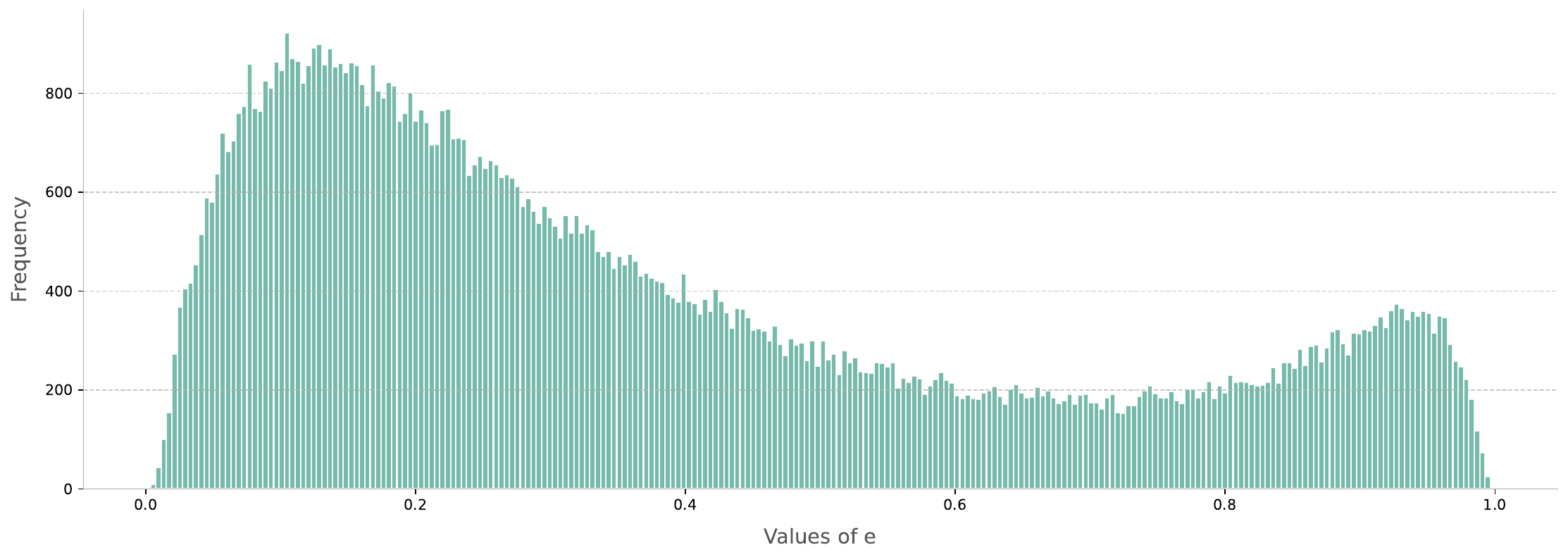}
    \caption{Histogram of the propensity score of Logistic DGPs}
    \label{fig:hist_e}
\end{figure*}

\subsubsection{Non-Linear and Non-Logistic DGP}
We use the same simulations as in \cite{Wager} using nonlinear models for every quantity, as detailed below, with $X \sim \text{Unif}(0,1)^6$
%Non linear baseline model and linear treatment effect and the treatment allocation function is Non-Logistic as detailed below:
\begin{align*}
        m(X) & = \sin \left(\pi X_{1}X_{2}\right)+2\left(X_{3}-0.5\right)^{2}+X_{4}+0.5 X_{5}  - \left(X_{1}+X_{2}\right) / 4\\
        b(X) & = \left(X_{1}+X_{2}\right) / 2  \\
        e(X) & = \max \{0.1, \min (\sin \left(\pi X_{1}\right), 0.9)\}.
\end{align*}
\begin{figure*}[h!]
  \includegraphics[width=\textwidth,height=4cm]{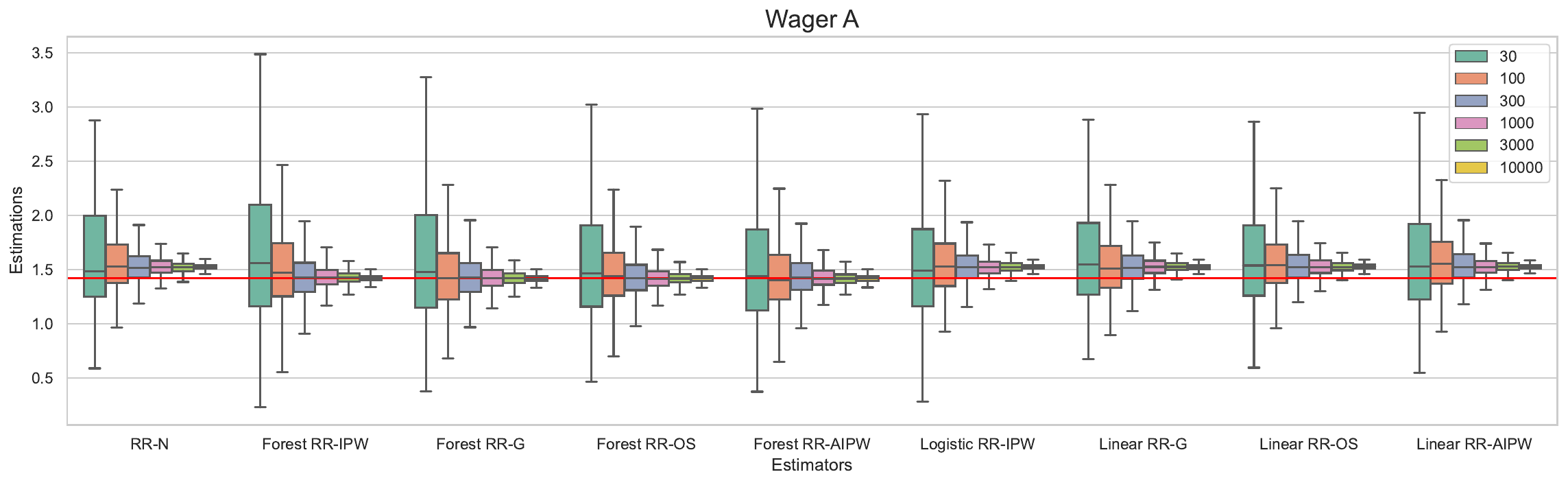}
    \caption{Estimations of the Risk Ratio with weighting, outcome based and augmented estimators as a function of the sample size for the non-Linear-non-Logistic DGP. Parametric (Regression) and non parametric (Forest) estimations of nuisance are displayed.}
    \label{fig:Wager_A}
\end{figure*}

\begin{figure*}[h!]
  \includegraphics[width=\textwidth,height=4cm]{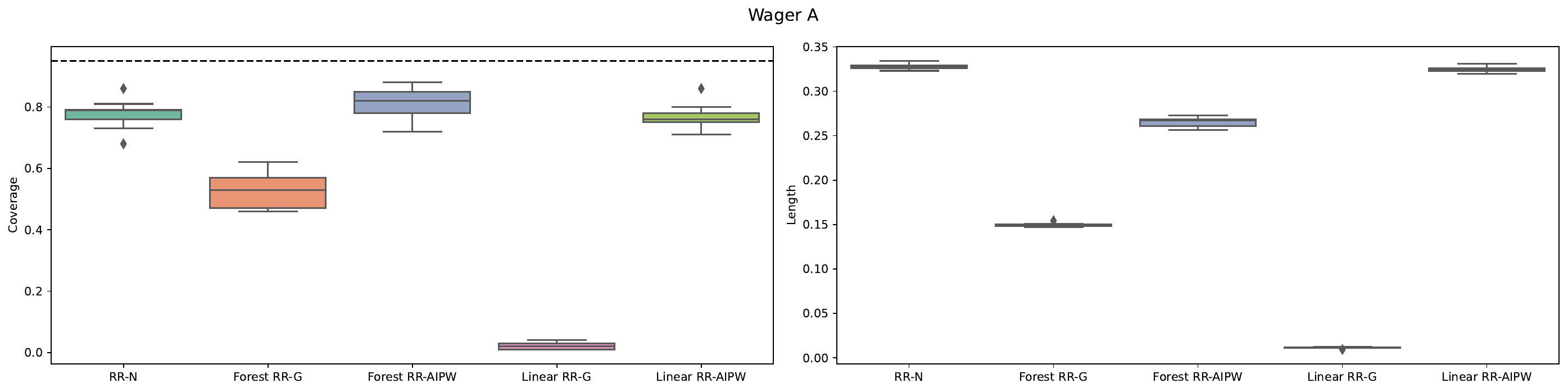}
    \caption{Average coverage (left) and average length (right) of asymptotic confidence interval derived from \Cref{risk_ratio_RCTs} and \Cref{risk_ratio_OBS} for different estimators with $n=1000$ and $300$ repetitions for a Non-Linear and Non-Logistic DGP.}
    \label{fig:Wager_CI_A}
\end{figure*}

 Results are presented in \Cref{fig:Wager_A}. At first glance, all methods seem to have similar performances. However, estimators based on parametric estimators (last four) fail to converge to the correct quantity. They present an intrinsic bias, which does not vanish as the sample size increases. This was expected as linear methods are unable to model the complex non-linear generative process of this simulation. On the other hand, methods that employ random forests estimators achieve good performances: they are consistent and unbias even for small sample sizes. Note that RIPW has a larger variance than the other methods, with a small bias for very small sample sizes. Therefore, the G-formula and the two doubly-robust estimators that use random forests are competitive in this setting. Here again, both double robust estimators give similar performances. 
 No estimator achieves 95\% coverage, which is expected given the non-linear, non-logistic DGP. Linear estimators, such as Linear RR-G and RR-AIPW, struggle to accurately estimate the nuisance functions in this context. Additionally, the limited number of observations prevents the Forest estimators from converging effectively.

\subsection{Semi-synthetic simulations}\label{an:semi-synthetic}

In our semi-synthetic experiment, we retain the original co-variates, \(X\), and first regenerate the treatment assignment, \(T\), using a logistic regression (see \Cref{an:semi-synthetic} for random forests classifier) trained on the original 17 co-variates to predict the observed treatment. Based on the estimated probabilities from this model, we sample a binary treatment from a Bernoulli distribution. Next, we simulate potential outcomes using two separate logistic regressions: one for the treated group and another for the control group, incorporating all co-variates. This process yields the potential outcomes, \(Y^{(0)}\) and \(Y^{(1)}\), for each individual, enabling us to compute the true relative risk, \(\tau_{RR} = 1.2\). Finally, we estimate the relative risk and its variance using bootstrap resampling across different sample sizes (300, 1000, and 3000). Results are displayed in \Cref{fig:semi_synthetic_linear}. All estimators except Neyman appear to converge to the true value of the Risk Ratio. As anticipated, doubly-robust estimators are consistent and efficient (smallest asymptotic variance).

 \begin{figure*}[!h]
  \includegraphics[width=\textwidth,height=4cm]{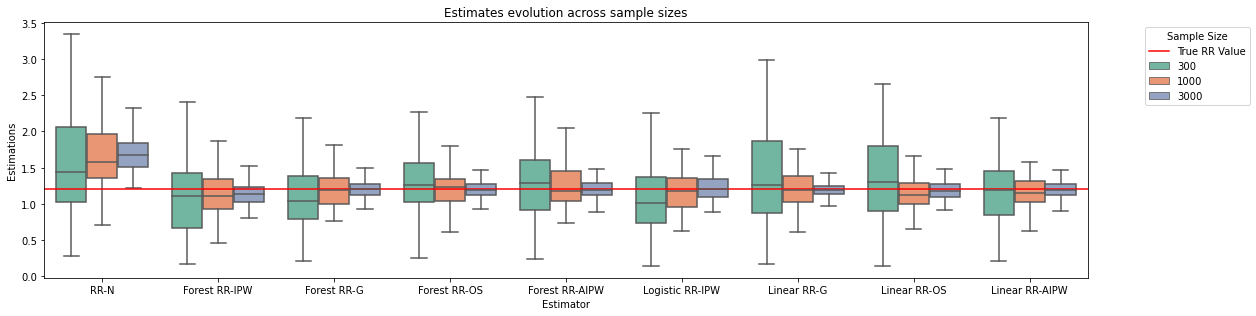}
    \caption{RR estimations with weighting, outcome based and augmented estimators as a function of the sample size for the logistic Semi-synthetic data. Parametric (Linear) and non parametric (Forest) estimations of nuisance are displayed.}
    \label{fig:semi_synthetic_linear}
\end{figure*}

In this approach, we replace logistic regression with random forest classifiers to model the treatment \( T \) and the potential outcomes \( Y^{(1)} \) and \( Y^{(0)} \). As observed, all estimators, except for RR-N, converge toward the true value of the Risk Ratio. Moreover, since both the treatment assignment and potential outcomes are generated using Random Forest, estimators that also leverage Random Forest demonstrate reduced bias and lower variance. Conversely, linear estimators exhibit higher bias, likely due to model misspecification.

 \begin{figure*}[!h]
  \includegraphics[width=\textwidth,height=4cm]{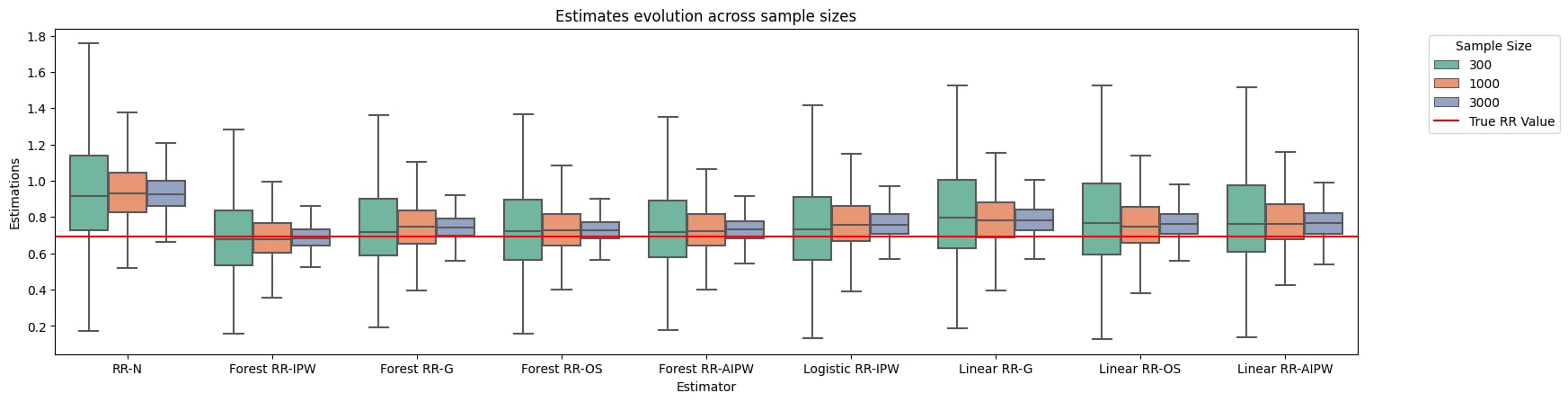}
    \caption{RR estimations with weighting, outcome based and augmented estimators as a function of the sample size for the Semi-synthetic data. Parametric (Linear) and non parametric (Forest) estimations of nuisance are displayed.}
    \label{fig:semi_synthetic}
\end{figure*}

\end{document}